\let\oldvec\vec
\RequirePackage{fix-cm}
\documentclass[final,numbook,envcountsame,smallcondensed]{svjour3}     
\let\vec\oldvec

\usepackage{amsmath,amsfonts,amssymb}
\usepackage{graphicx}
\usepackage{float}
\floatstyle{boxed}
\restylefloat{table}
\restylefloat{figure}
\usepackage{cite}
\usepackage[T1]{fontenc}
\allowdisplaybreaks
\usepackage{wrapfig}
\usepackage{proof}
\usepackage[braket]{qcircuit}
\usepackage{tikz}
\usetikzlibrary{cd,positioning}
\usepackage[only=llparenthesis,rrparenthesis,llbracket,rrbracket]{stmaryrd}
\usepackage{anyfontsize}
\SetSymbolFont{stmry}{bold}{U}{stmry}{m}{n}
\usepackage{multicol,multirow}
\usepackage{xspace}
\usepackage{adjustbox}

\newcommand\axiomd{Axiom$_{{}_{\mathsf{Dist}}}$}
\newcommand\axiomz{Axiom$_{{}_0}$}
\newcommand\cV{\mathcal V}
\newcommand\cW{\mathcal W}
\newcommand\cC{\mathcal C}

\newcommand\LambdaSM{Lambda-$\mathcal S$\xspace}
\newcommand\LambdaS{Lambda-$\mathcal S^*$\xspace}
\newcommand\xrecap[4]{\noindent {\bf #1 \ref{#3} (#2).} \emph{#4}}

\newcommand\citaCR{\cite[Thm.~7.25]{ArrighiDowekLMCS17}}
\newcommand\citaSR{\cite[Thm.~5.12]{DiazcaroDowekRinaldiBIO19}}
\newcommand\citaSN{\cite[Thm.~6.10]{DiazcaroDowekRinaldiBIO19}}

\newcommand\home[2]{[#1,#2]}
\newcommand\lra{\longrightarrow}
\newcommand\sem[1]{\left\llbracket {#1}\right\rrbracket}
\newcommand\com[1]{\mathcal{C}_{#1}}
\newcommand\B{\ensuremath{\mathbb B}}
\newcommand\types{\ensuremath{\mathcal T}}
\newcommand\qtypes{\ensuremath{\mathcal Q}}
\newcommand\btypes{\ensuremath{\mathcal B}}
\newcommand\vars{\ensuremath{\mathsf{Vars}}}
\newcommand\values{\ensuremath{\mathsf V}}
\newcommand\tbasis{\ensuremath{\mathsf B}}

\newcommand\s[1]{\ensuremath{\mathsf{#1}}}
\newcommand\head{\text{\sl head}}
\newcommand\tail{\text{\sl tail}}
\newcommand\pair[2]{({#1}+{#2})}
\newcommand\z[1][A]{\mathbf 0_{S(#1)}}
\newcommand\ite[3]{{#1}?{#2}\mathord{\cdot}{#3}}
\newcommand\C{\mathbb C}
\newcommand\I{I}
\newcommand\Id{\mathsf{Id}}
\newcommand\xlra[1]{\xrightarrow{#1}}
\newcommand\tax{\textsl{Ax}}
\newcommand\tif{\textsl{If}}
\newcommand\rbetab{(\s{\beta_b})}
\newcommand\rbetan{(\s{\beta_n})}
\newcommand\riftrue{(\s{if_{1}})}
\newcommand\riffalse{(\s{if_{0}})}
\newcommand\rlinr{(\s{lin^+_r})}
\newcommand\rlinscalr{(\s{lin^\alpha_r})}
\newcommand\rlinzr{(\s{lin^0_r})}
\newcommand\rlinl{(\s{lin^+_l})}
\newcommand\rlinscall{(\s{lin^\alpha_l})}
\newcommand\rlinzl{(\s{lin^0_l})}
\newcommand\rneut{(\s{neutral})}
\newcommand\runit{(\s{unit})}
\newcommand\rzeros{(\s{zero_\alpha})}

\newcommand\rzero{(\s{zero})}
\newcommand\rprod{(\s{prod})}
\newcommand\rdists{(\s{\alpha dist})}
\newcommand\rdistcasum{(\s{dist^+_\Uparrow})}
\newcommand\rdistcascal{(\s{dist^\alpha_\Uparrow})}
\newcommand\rdistcazeror{(\s{dist^0_{\Uparrow_r}})}
\newcommand\rdistcazerol{(\s{dist^0_{\Uparrow_\ell}})}
\newcommand\rcaneutl{(\s{neut^\Uparrow_\ell})}
\newcommand\rcaneutr{(\s{neut^\Uparrow_r})}

\newcommand\rcaneutzl{(\s{neut^\Uparrow_{0\ell}})}
\newcommand\rcaneutzr{(\s{neut^\Uparrow_{0r}})}
\newcommand\rfact{(\s{fact})}
\newcommand\rfacto{(\s{fact^1})}
\newcommand\rfactt{(\s{fact^2})}

\newcommand\rhead{(\s{head})}
\newcommand\rtail{(\s{tail})}
\newcommand\rdistzr{(\s{dist^0_r})}
\newcommand\rdistzl{(\s{dist^0_l})}
\newcommand\rdistscalr{(\s{dist^\alpha_r})}
\newcommand\rdistscall{(\s{dist^\alpha_l})}
\newcommand\rdistsumr{(\s{dist^+_r})}
\newcommand\rdistsuml{(\s{dist^+_l})}

\journalname{arXiv}

\begin{document}

\title{A categorical construction for the computational definition of vector
spaces\thanks{A.~D\'iaz-Caro has been partially supported by PICT 2015 1208,
ECOS-Sud A17C03 QuCa and PEDECIBA. O.~Malherbe has been partially supported by
MIA CSIC UdelaR.}}
\author{Alejandro D\'iaz-Caro \and Octavio Malherbe}

\institute{A.~D\'iaz-Caro \at
  Universidad Nacional de Quilmes \&\\
  Instituto de Ciencias de la Computaci\'on (CONICET-Universidad de Buenos Aires)\\
  Buenos Aires, Argentina\\
  \email{adiazcaro@icc.fcen.uba.ar}
  \and
  O.~Malherbe \at
  Departamento de Matem\'atica y Afines, CURE \& \\
  Instituto de Matem\'atica y Estad\'istica Rafael Laguardia, Facultad de Ingenier\'ia\\
  Universidad de la Rep\'ublica\\
  Montevideo, Uruguay\\
  \email{malherbe@fing.edu.uy}
}

\date{Submitted to arXiv}

\maketitle

\begin{abstract} 
  \LambdaSM is an extension to first-order lambda calculus unifying two
approaches of non-cloning in quantum lambda-calculi. One is to forbid
duplication of variables, while the other is to consider all lambda-terms as
algebraic linear functions. The type system of \LambdaSM has a constructor $S$
such that a type $A$ is considered as the base of a vector space while $S(A)$ is
its span. \LambdaSM can also be seen as a language for the computational
manipulation of vector spaces: The vector spaces axioms are given as a rewrite
system, describing the computational steps to be performed. In this paper we
give an abstract categorical semantics of \LambdaS (a fragment of \LambdaSM),
showing that $S$ can be interpreted as the composition of two functors in an
adjunction relation between a Cartesian category and an additive symmetric
monoidal category. The right adjoint is a forgetful functor $U$, which is hidden
in the language, and plays a central role in the computational reasoning.
\end{abstract}
  
\section{Introduction}\label{sec:Intro}
Algebraic lambda calculi aim to embed to the lambda calculus, the notion of
vector spaces over programs. This way a linear combination $\alpha.v+\beta.w$ of
programs $v$ and $w$, for some scalars $\alpha$ and $\beta$, is also a program~\cite{AssafDiazcaroPerdrixTassonValironLMCS14}.
This kind of construction has two independent origins. The {\em Algebraic Lambda
Calculus} (ALC for short)~\cite{VauxMSCS09} has been introduced as a fragment of
the {\em Differential Lambda Calculus}~\cite{EhrhardRegnierTCS03}, which is
itself originated from Linear Logic~\cite{GirardTCS87}. ALC can be seen as the
Differential Lambda Calculus without a differential operator. In the ALC the
notion of vector spaces is embedded in the calculus with an equational theory,
so the axioms of vector spaces, such as $\alpha.v+\beta.v = (\alpha+\beta).v$ are
seen as equalities between programs.
On the other hand, the {\em Linear Algebraic Lambda Calculus} (Lineal for
short)~\cite{ArrighiDowekLMCS17} was meant for quantum computation. The aim of
Lineal is to provide a computational definition of vector space and bilinear
functions, and so, it defines
the axioms of vector spaces as
rewrite rules, providing a confluent calculus. This way, an equality such as
$-v+v+3.w-2.w = w$ is described computationally step by step as
\begin{align*}
  (-1).v+v+3.w+(-2).w
  &\lra 0.v+3.w+(-2).w\\
  &\lra 0.v+1.w\\
  &\lra \mathbf{0}+1.w\\
  &\lra 1.w\\
  &\lra w
\end{align*}
Rules like $\alpha.v+\beta.v\lra(\alpha+\beta).v$ say that these expressions are
not \emph{the same} but one reduces to the other, and so, a computational step
has been performed. The backbone of this computation can be described as having
an element $\alpha.v+\beta.v$ without properties, which is decomposed into its
constituents parts $\alpha$, $\beta$, and $v$, and reconstructed in another way.
Otherwise, if we consider $\alpha.v+\beta.v$ being just a vector, as in the ALC, then it would
be equal  to $(\alpha+\beta).v$ and the computation needed to arrive from
the former to the latter would be ignored.
The main idea in the present paper is to study the construction of Lineal from a
categorical point of view, with an adjunction between a Cartesian closed
category, which will treat the elements as not having properties, and an
additive symmetric monoidal closed category, where the underlying properties
will allow to do the needed algebraic manipulation. A concrete example is an
adjunction between the
category $\mathbf{Set}$ of sets and the category $\mathbf{Vec}$ of vector
spaces. This way, a functor from $\mathbf{Set}$ to $\mathbf{Vec}$ will allow to
do the needed manipulation, while a forgetful functor from $\mathbf{Vec}$ to
$\mathbf{Set}$ will return the result of the computation.

The calculus \LambdaSM~\cite{DiazcaroDowekTPNC17,DiazcaroDowekRinaldiBIO19} is
a first-order typed fragment of Lineal, extended with measurements.
The type system has been designed as a quantum lambda calculus, where the main
goal was to study the non-cloning restrictions. In quantum computing a known
vector, such as a basis vector from the base considered for the measurements,
can be duplicated freely (normally the duplication process is just a preparation
of a new qubit in the same known basis state), while an unknown vector cannot. For this reason, a linear-logic
like type system has been put in place.
In linear logic we would write $A$ the types of terms that cannot be duplicated while ${!}A$
types duplicable terms. In \LambdaSM instead $A$ are the types of the terms that
represent basis vectors, while $S(A)$ are linear combinations of those (the span
of $A$). Hence, 
$A$ means that we can duplicate, while
$S(A)$ means that we cannot duplicate. Therefore, the $S$ is not the same as the bang
``$!$'', but somehow the opposite. This can be explained by the fact that
linear logic is focused on the possibility of duplication, while here we focus
on the possibility of superposition, which implies the impossibility of
duplication.

In~\cite{DiazcaroDowekTPNC17,DiazcaroDowekRinaldiBIO19} a first denotational semantics (in environment
style) is given where the type $\B$ is interpreted as $\{\ket 0,\ket 1\}$ while
$S(\B)$ is interpreted as $\mathsf{Span}(\{\ket 0,\ket 1\})=\C^2$, and, in
general, a type $A$ is interpreted as a basis while $S(A)$ is the vector space
generated by such a basis. In~\cite{DiazcaroMalherbeLSFA18,DiazcaroMalherbe2020} we went further and
gave a preliminary concrete categorical interpretation of \LambdaSM where $S$ is
a functor of an adjunction between the category $\mathbf{Set}$ and the category
$\mathbf{Vec}$. Explicitly, when we evaluate $S$ we obtain formal finite linear
combinations of elements of a set with complex numbers as coefficients and the
other functor of the adjunction, $U$, allows us to forget the vectorial
structure. In this paper, we define the abstract categorical semantics of the
fragment of \LambdaSM without measurement, which we may refer as \LambdaS, so we focus on the computational
definition of vector spaces, avoiding any interference produced by probabilistic
constructions.

The main structural feature of our model is that it is expressive enough to
describe the bridge between the property-less elements such as
$\alpha.v+\beta.v$, without any equational theory, and the result of its
algebraic manipulation into $(\alpha+\beta).v$, explicitly controlling its
interaction. In the literature, intuitionistic linear (as in linear-logic)
models are obtained by a monoidal comonad determined by a monoidal adjunction
$(S,m)\dashv (U,n)$, i.e.,~the bang $!$ is interpreted by the comonad $SU$
(see~\cite{BentonCSL94}). In a different way, a crucial ingredient of our model
is to consider the monad $US$ for the interpretation of $S$ determined by a
similar monoidal adjunction. This implies that on the one hand we have tight
control of the Cartesian structure of the model (i.e.~duplication, etc) and on
the other hand the world of superpositions lives in some sense inside the
classical world, i.e.~determined externally by classical rules until we decide
to explore it. This is given by the following composition of maps:
\[
  US(\B)\times US(\B)\xlra n U(S(\B)\otimes S(\B))\xlra{U(m)} US(\B\times \B)
\]
that allows us to operate in a monoidal structure explicitly allowing the
algebraic manipulation and then to return to the Cartesian product.

This is different from linear logic, where the $!$ stops any algebraic
manipulation, i.e.~$({!}\B)\otimes ({!}\B)$ is a product inside a monoidal
category.

\paragraph{Outline.}
The paper is structured as follows.
\begin{itemize}
\item Section~\ref{sec:calculus} gives the intuition and formalization of the
  fragment of \LambdaSM without measurements, called \LambdaS, we give some examples, and state its main
properties.
\item Section~\ref{sec:Category} presents the categorical construction for algebraic
manipulation.
\item Section~\ref{sec:DenSem} gives a denotational semantics of \LambdaS, using
  the categorical constructions from Section~\ref{sec:Category}.
\item Section~\ref{sec:properties} proves the soundness and completeness of such
  semantics.
\item Finally, we conclude in Section~\ref{sec:conclusion}. We also include an
  appendix with detailed proofs.
\end{itemize}

\section{The calculus \LambdaS}\label{sec:calculus}
In this section we define \LambdaS, a fragment of \LambdaSM, without
measurements. In addition, instead of considering the scalars in $\C$, we use
any commutative ring, which we will write $\mathcal C$, so to make the system
more general.

The syntax of terms and types is given in Figure~\ref{fig:syntax}, where we write
$\B^n$ for $\B\times\cdots\times\B$ $n$-times, with the convention that
$\B^1=\B$.
We use capital Latin letters
($A,B,C,\dots$) for general types and the capital Greek letters $\Psi$, $\Phi$,
$\Xi$, and $\Upsilon$ for qubit types.  $\btypes=\{\B^n\mid n\in\mathbb N\}$,
$\qtypes$ is the set of qubit types, and $\types$ is the set of types
($\btypes\subsetneq\qtypes\subsetneq\types$).  In the same way, $\vars$ is the
set of variables, $\tbasis$ is the set of basis terms, $\values$ the set of
values, and $\Lambda$ the set of terms. We have
$\vars\subsetneq\tbasis\subsetneq\values\subsetneq\Lambda$.

\begin{figure}[!h]
  \centering
  \[
    \begin{array}{rlr}
      \Psi & := \B^n\mid S(\Psi)\mid \Psi\times\Psi & \textrm{Qubit types (\qtypes)} \\
      A & := \Psi\mid \Psi\Rightarrow A\mid S(A) & \textrm{Types (\types)} \\                                                    
      \\
      b  & := x\mid \lambda x{:}\Psi.t\mid \ket 0\mid \ket 1\mid b\times b & \textrm{Basis terms (\tbasis)} \\
      v  & := b\mid \pair vv\mid \z\mid \alpha.v\mid v\times v & \textrm{Values (\values)} \\
      t  & := v\mid tt\mid \pair tt\mid\alpha.t\mid \ite{}tt\mid 
      t\times t\mid \head~t\mid \tail~t\mid \Uparrow_r t\mid \Uparrow_\ell
      t & \textrm{Terms ($\Lambda$)}
    \end{array}
  \]
  where $\alpha\in\mathcal C$.
  \caption{Syntax of types and terms of \LambdaS.}
  \label{fig:syntax}
\end{figure}

The intuition of these syntaxes is given by considering $\mathcal C=\C$. The type $\B$ is the type of a
specific base of $\C^2$, the base $\{\ket 0,\ket 1\}$, where we use the standard
notation from quantum computing $\ket 0$ and $\ket 1$: $\ket 0$ denotes the
vector $\left(\begin{smallmatrix}1\\0\end{smallmatrix}\right)$ and $\ket 1$
denotes the vector $\left(\begin{smallmatrix} 0\\1 \end{smallmatrix} \right)$.
This way, $\B\times \B = \{\ket 0\times\ket 0,\ket 0\times\ket 1,\ket
1\times\ket 0,\ket 1\times\ket 1\}$ is the base of $\C^4$. The type $S(\B)$ is
the type of any vector in $\C^2$, so $S$ can be seen as the span operator. For
example, $2.\ket 0+i.\ket 1$ may live inside $S(\B)$. On the other hand, a type
of the form $\B\times S(\B)$ is the type of a pair of a base vector with a
general vector, for example $\ket 0\times(\alpha.\ket 0+\beta.\ket 1)$ will have
this type. There is no type for pair of function types, only pair of qubit types
are considered. The type constructor $S$ can be used on any type, for example,
the type $S(\B\Rightarrow\B)$ is a valid type which denotes the types of
superpositions of functions, such as $2.\lambda x{:}\B.x+3.\lambda x{:}\B.\ket 0$.
We will come back to the meaning of superposed functions later.

Terms are considered modulo associativity and commutativity of the
syntactic symbol $+$.  

The term syntax is split in three: basis terms, which are values in the base
of a vector space of values. Values, which are obtained by the formal linear
combinations of basis terms, together with a null vector $\z[A]$ associated to
each type $S(A)$. And a set $\Lambda$ of general terms, which includes the values.

The syntax of terms contains:
\begin{itemize}
  \item The three basic terms for first-order lambda-calculus, namely,
    variables, abstractions and applications.
  \item Two basic terms $\ket 0$ and $\ket 1$ to represent qubits, and one test
    $\ite{}rs$ on them. We may write $\ite trs$ for $(\ite{}rs)t$, see
    Example~\ref{ex:ite} for a clarification of why to choose this
    presentation.
  \item A product $\times$ to represent associative pairs (i.e.~lists), with
    its destructors $\head$ and $\tail$. We may use the notation
    $\ket{b_1b_2\dots b_n}$ for
    $\ket{b_1}\times\ket{b_2}\times\dots\times\ket{b_n}$.
  \item Constructors to write linear combinations of terms, namely $+$ (sum)
    and $.$ (scalar multiplication), without destructor (the destructor is the
    measuring operator, which we have explicitly left out of this presentation),
    and one
    null vector $\z$ for each type $S(A)$.
  \item Two casting functions $\Uparrow_r$ and $\Uparrow_\ell$ which allows us
    to consider lists of superpositions as superpositions of lists (see
    Example~\ref{ex:cast}).
\end{itemize}

The rewrite system has not yet been exposed, however the next examples give some
intuitions and clarify the $\ite{}rs$ and the casting functions.
\begin{example}
  \label{ex:ite}
  The term $\ite{}rs$ is meant to test whether the condition is $\ket 1$ or
  $\ket 0$.  However, defining it as a function, allows us to use the algebraic
  linearity to implement the quantum-if~\cite{AltenkirchGrattageLICS05}:
  \[
    (\ite{}rs)(\alpha.\ket 1+\beta.\ket 0) = \ite{(\alpha.\ket 1+\beta.\ket
    0)}rs \lra^*\alpha.\ite{\ket 1}rs+\beta.\ite{\ket 0}rs
    \lra^*\alpha.r+\beta.s
  \]
\end{example}
\begin{example}
  \label{ex:cast}
  The term $(\frac 1{\sqrt 2}(\ket 0+\ket 1))\times\ket 0$ is the encoding of
  the qubit $\frac 1{\sqrt 2}(\ket 0+\ket 1)\otimes\ket 0$. However, while the
  qubit $\frac 1{\sqrt 2}(\ket 0+\ket 1)\otimes\ket 0$ is equal to $\frac
  1{\sqrt 2}(\ket 0\otimes\ket 0+\ket 1\otimes\ket 0)$, the term will not
  rewrite to the encoding of it, unless a casting $\Uparrow_r$ is preceding the
  term:
  \[
    \Uparrow_r(\frac 1{\sqrt 2}(\ket 0+\ket 1))\times\ket 0\lra^*\frac 1{\sqrt
    2}(\ket 0\times\ket 0+\ket 1\times\ket 0)
  \]
  The reason is that we want the term $(\frac 1{\sqrt 2}(\ket 0+\ket
  1))\times\ket 0$ to have type $S(\B)\times\B$, highlighting the fact that the
  second qubit is a basis qubit, i.e.~duplicable, while the term $\frac 1{\sqrt
  2}(\ket 0\times\ket 0+\ket 1\times\ket 0)$ will have type $S(\B\times\B)$,
  showing that the full term is a superposition where no information can be
  extracted and hence, non-duplicable.
\end{example}

The rewrite system depends on types. Indeed, $\lambda x{:}{S\Psi}.t$ follows a
call-by-name strategy, while $\lambda x{:}\B.t$, which can duplicate its
argument, must follow a call-by-base
strategy~\cite{AssafDiazcaroPerdrixTassonValironLMCS14}, that is, not only the
argument must be reduced first, but also it will distribute over linear
combinations prior to $\beta$-reduction. Therefore, we give first the type system and then the rewrite
system.

The typing relation is given in Figure~\ref{fig:types}.  Contexts, identified
by the capital Greek letters $\Gamma$, $\Delta$, and $\Theta$, are partial
functions from $\vars$ to $\types$. The contexts assigning only types in
$\btypes$ are identified with the super-index $\B$, e.g.~$\Theta^\B$.  Whenever
more than one context appear in a typing rule, their domains are considered
pair-wise disjoint. Observe that all types are linear (as in linear-logic)
except on basis types $\B^n$, which can be weakened and contracted (expressed
by the common contexts $\Theta^\B$).

Notice that rule $S_I$ makes type assignment not unique, since it makes possible to add as
many $S$ as wished. Also, there can be more than one type
derivation tree assigning the same type, for example:
\[
  \infer[\alpha_I]{\vdash 1.\z[\B]:S(S(\B))}
  {
    \infer[S_I]{\vdash\z[\B]:S(S(\B))}
    {
      \infer[\tax_{\mathbf 0}]{\vdash\z[\B]:S(\B)}{} 
    }
  }
  \qquad
  \textrm{and}
  \qquad
  \infer[S_I]{\vdash 1.\z[\B]:S(S(\B))}
  {
    \infer[\alpha_I]{\vdash 1.\z[\B]:S(\B)}
    {
      \infer[\tax_{\mathbf 0}]{\vdash\z[\B]:S(\B)}{} 
    }
  }
\]

\begin{figure}[h!]
  \centering
  \[
    \infer[^\tax] {\Theta^\B,x:\Psi\vdash x:\Psi} {}
    \qquad
    \infer[^{\tax_{\mathbf 0}}] {\Theta^\B\vdash \z:S(A)} {}
    \qquad
    \infer[^{\tax_{\ket 0}}] {\Theta^\B\vdash\ket 0:\B} {}
    \quad
    \infer[^{\tax_{\ket 1}}] {\Theta^\B\vdash\ket 1:\B} {}
  \]
  \[
    \infer[^{\alpha_I}] {\Gamma\vdash \alpha.t:S(A)} {\Gamma\vdash t:S(A)}
    \quad
    \infer[^{+_I}] {\Gamma,\Delta,\Theta^\B\vdash\pair tu:S(A)} {\Gamma,\Theta^\B\vdash t:S(A) & \Delta,\Theta^\B\vdash u:S(A)}
    \quad
    \infer[^{S_I}] {\Gamma\vdash t:S(A)} {\Gamma\vdash t:A}
  \]
  \[
    \infer[^\tif]{\Gamma\vdash\ite{}tr:\B\Rightarrow A}{\Gamma\vdash t:A &
    \Gamma\vdash r:A}
    \qquad
    \infer[^{\Rightarrow_I}] {\Gamma\vdash\lambda x{:}\Psi.t:\Psi\Rightarrow A} {\Gamma,x:\Psi\vdash t:A}
  \]
  \[
    \infer[^{\Rightarrow_E}] {\Delta,\Gamma,\Theta^\B\vdash tu:A}
    {
      \Delta,\Theta^\B\vdash u:\Psi
      &
      \Gamma,\Theta^\B\vdash t:\Psi\Rightarrow A
    }
    \qquad
    \infer[^{\Rightarrow_{ES}}] {\Delta,\Gamma,\Theta^\B\vdash tu:S(A)}
    {
      \Delta,\Theta^\B\vdash u:S(\Psi)
      &
      \Gamma,\Theta^\B\vdash t:S(\Psi\Rightarrow A)
    }
  \]
  \[
    \infer[^{\times_I}] {\Gamma,\Delta,\Theta^\B\vdash t\times u:\Psi\times\Phi} {\Gamma,\Theta^\B\vdash t:\Psi & \Delta,\Theta^\B\vdash u:\Phi}
    \qquad
    \infer[^{\times_{Er}}] {\Gamma\vdash \head~t:\B} {\Gamma\vdash t:\B^n &
    {\scriptstyle n>1}}
    \qquad
    \infer[^{\times_{El}}] {\Gamma\vdash \tail~t:\B^{n-1}} {\Gamma\vdash t:\B^n &
    {\scriptstyle n>1}}
  \]
  \[
    \infer[^{\Uparrow_\ell}]{\Gamma\vdash \Uparrow_\ell t:S(\Psi\times
      \Phi)}{\Gamma\vdash t:S(\Psi\times S^k(\Phi)) & {\scriptstyle k>0} &{\scriptstyle \Psi\neq S(A)} }
    \qquad
    \qquad
    \infer[^{\Uparrow_r}]{\Gamma\vdash \Uparrow_rt:S(\Psi\times
      \Phi)}{\Gamma\vdash t:S(S^k(\Psi)\times \Phi) & {\scriptstyle k>0} &{\scriptstyle \Psi\neq S(A)} }
  \]
  \caption{Typing relation}
  \label{fig:types}
\end{figure}

The rewrite relation is given in Figures~\ref{fig:TRSbeta} to
\ref{fig:TRScontext}.

The two beta rules (Figure~\ref{fig:TRSbeta}) are applied according to the shape
of the abstraction. If the abstraction expects an argument with a superposed type,
then the reduction follows a call-by-name strategy (rule $\rbetan$), while if
the abstraction expects a basis type, the reduction is call-by-base (rule
$\rbetab$): it $\beta$-reduces only when its argument is a basis term. However,
typing rules also allow to type an abstraction expecting an argument with basis
type, applied to a term with superposed type (cf.~Example~\ref{ex:CM}). In this
case, the $\beta$-reduction cannot occur and, instead, the application must
distribute using the rules from Figure~\ref{fig:TRSld}: the linear distribution
rules.
\begin{figure}
  \begin{equation*}
    \begin{aligned}
      \textrm{If $b$ has type $\B^n$ and $b\in\tbasis$, } (\lambda x{:}{\B^n}.t)b &\lra (b/x)t & \rbetab\\
      \textrm{If $u$ has type $S\Psi$, } (\lambda x{:}{S\Psi}.t)u &\lra (u/x)t & \rbetan
    \end{aligned}
  \end{equation*}
  \caption{Beta rules}
  \label{fig:TRSbeta}
\end{figure}
\begin{figure}
  \begin{equation*}
    \begin{aligned}
      \textrm{If $t$ has type $\B^n\Rightarrow A$, } t\pair uv &\lra \pair{tu}{tv} & \rlinr\\
      \textrm{If $t$ has type $\B^n\Rightarrow A$, } t(\alpha.u) &\lra\alpha.tu & \rlinscalr\\
      \textrm{If $t$ has type $\B^n\Rightarrow A$, } t\z[\B^n] &\lra\z &\rlinzr\\
      \pair tuv &\lra\pair{tv}{uv} & \rlinl\\
      (\alpha.t)u &\lra\alpha.tu &\rlinscall\\
      \z[\B^n\Rightarrow A]t &\lra\z[A] &\rlinzl
    \end{aligned}
  \end{equation*}
  \caption{Linear distribution rules}
  \label{fig:TRSld}
\end{figure}

Figure~\ref{fig:TRSif} gives the two rules for the conditional construction.
Together with the linear distribution rules (cf.~Figure~\ref{fig:TRSld}), these
rules implement the quantum-if (cf.~Example~\ref{ex:ite}).
\begin{figure}
  \begin{equation*}
    \begin{aligned}
      \ite{\ket 1}tr &\lra t &\riftrue
    \end{aligned}\hspace{3cm}
    \begin{aligned}
      \ite{\ket 0}tr &\lra r &\riffalse
    \end{aligned}
  \end{equation*}
  \caption{Rules of the conditional construction}
  \label{fig:TRSif}
\end{figure}

Figure~\ref{fig:TRSlists} gives the rules for lists, \rhead\ and \rtail.
\begin{figure}
  \begin{equation*}
    \begin{aligned}
      \textrm{If $h\neq u\times v$ and $h\in\tbasis$, }\head\ h\times t &\lra h & \rhead\\
      \textrm{If $h\neq u\times v$ and $h\in\tbasis$, }\tail\ h\times t &\lra t & \rtail
    \end{aligned}
  \end{equation*}
  \caption{Rules for lists}
  \label{fig:TRSlists}
\end{figure}

Figure~\ref{fig:TRSvs} deals with the vector space structure implementing a
directed version of the vector space axioms. The direction is chosen in order
to yield a canonical form~\cite{ArrighiDowekLMCS17}.
\begin{figure}
  \begin{equation*}
    \begin{aligned}
      \pair\z t &\lra t &\rneut\\
      1.t &\lra t &\runit\\
      \textrm{If $t$ has type $A$, }0.t &\lra\z &\rzeros\\
      \alpha.\z &\lra\z &\rzero\\
      \alpha.(\beta.t) &\lra (\alpha\beta).t &\rprod\\
      \alpha.\pair tu &\lra\pair{\alpha.t}{\alpha.u} &\rdists\\
      \pair{\alpha.t}{\beta.t} &\lra(\alpha+\beta).t &\rfact\\
      \pair{\alpha.t}t &\lra (\alpha+1).t &\rfacto\\
      \pair tt &\lra 2.t &\rfactt
    \end{aligned}
  \end{equation*}
  \caption{Rules implementing the vector space axioms}
  \label{fig:TRSvs}
\end{figure}

Figure~\ref{fig:TRScasts} are the rules to implement the castings. The idea is
that $\times$ does not distribute with respect to $+$, unless a casting allows
such a distribution. This way, the types $\B\times S(\B)$ and $S(\B\times\B)$
are different. Indeed, $\ket 0\times(\ket 0+\ket 1)$ has the first type but
not the second, while $\ket 0\times\ket 0+\ket 0\times\ket 1$ has the second
type but not the first. This way, the first type give us the information that
the state is separable, while the second type does not. We can choose to take the
first state as a pair of qubits forgetting the separability information, by
casting its type, in the same way as in certain programming languages an
integer can be cast to a float (and so, forgetting the information that it
was indeed an integer and not any float).

A second example is to take again Example~\ref{ex:cast}: The term $\frac 1{\sqrt 2}.(\ket 0+\ket
1)\times\ket 0$ has type $S(\B)\times\B$, expressing the fact that it is the
composition of a superposed qubit with a basis qubit. However, the term $\frac
1{\sqrt 2}.(\ket 0\times\ket 0+\ket 1\times\ket 0)$ has type
$S(\B\times\B)$, expressing the fact that it is a superposition of two qubits.
The first type give us information about the separability of the two-qubits
state, which is gathered from the fact that the term is indeed written as the
product of two qubits. Contrarily, the second term is not the product of two
qubits, and so the type cannot reflect its separability condition. In order to
not lose subject reduction, we need to cast the first term so we ``forget'' its
separability information, prior reduction.
\begin{figure}
  \begin{equation*}
    \begin{aligned}
      \Uparrow_r \pair rs\times u &\lra\pair{\Uparrow_r r\times u}{\Uparrow_r s\times u} &\rdistsumr\\
      \Uparrow_\ell u\times\pair rs &\lra\pair{\Uparrow_\ell u\times r}{\Uparrow_\ell u\times s} &\rdistsuml\\
      \Uparrow_r (\alpha.r)\times u &\lra \alpha.\Uparrow_r r\times u &\rdistscalr\\
      \Uparrow_\ell u\times(\alpha.r) &\lra \alpha.\Uparrow_r u\times r &\rdistscall\\
      \textrm{If $u$ has type $\Psi$, }\Uparrow_r \z[\Phi]\times u &\lra\z[\Phi\times\Psi] &\rdistzr\\
      \textrm{If $u$ has type $\Psi$, }\Uparrow_\ell u\times\z[\Phi] &\lra\z[\Psi\times\Phi] &\rdistzl\\
      \Uparrow\pair tu&\lra\pair{\Uparrow t}{\Uparrow u} &\rdistcasum\\
      \Uparrow(\alpha.t)&\lra\alpha.\Uparrow t &\rdistcascal\\
      \Uparrow_r\z[S(S\Psi)\times\Phi] &\lra\Uparrow_r\z[S\Psi\times\Phi] &\rdistcazeror\\
      \Uparrow_r\z[S(\B^n)\times\Phi] &\lra\z[\B^n\times\Phi] &\rcaneutzr\\
      \Uparrow_\ell\z[\Psi\times S(S\Phi)] &\lra\Uparrow_\ell\z[\Psi\times S\Phi] &\rdistcazerol\\
      \Uparrow_\ell\z[\Psi\times S(\B^n)] &\lra\z[\Psi\times\B^n] &\rcaneutzl\\
      \textrm{If }v\in\tbasis,\ \Uparrow_\ell u\times v&\lra u\times v &\rcaneutr\\
      \textrm{If }u\in\tbasis,\ \Uparrow_r u\times v&\lra u\times v &\rcaneutl
    \end{aligned}
  \end{equation*}
  \caption{Rules for castings $\Uparrow_r$ and $\Uparrow_\ell$}
  \label{fig:TRScasts}
\end{figure}

Finally, Figure~\ref{fig:TRScontext} give the contextual rules implementing the
call-by-value and call-by-name strategies.
\begin{figure}[t]
  \centering
  \[
    \begin{array}{r@{\ }lr}
      \multicolumn{3}{l}{\textrm{ If $t\lra u$, then}}\\
      \multicolumn{3}{c}{ 
	\begin{array}{c@{\qquad}c@{\qquad}c}
	  tv \lra uv & (\lambda x{:}{\B^n}.v)t\lra(\lambda x{:}{\B^n}.v)u & \pair tv\lra\pair uv \\
	  \alpha.t\lra\alpha.u& 
                                                                & t\times v\lra u\times v \\
	  v\times t\lra v\times u& \Uparrow_r t\lra\Uparrow_r u& \Uparrow_\ell t\lra\Uparrow_\ell u \\
	  \head\ t\lra\head\ u& \tail\ t\lra\tail\ u& \ite trs\lra\ite urs 
	\end{array}
      }
    \end{array}
  \]
  \caption{Contextual rules}
  \label{fig:TRScontext}
\end{figure}

\begin{example}\label{ex:CM}
  The term $\lambda x{:}\B.x\times x$ does not represent a cloning machine, but
  a CNOT\footnote{The CNOT quantum gate is such that CNOT$\ket{0x}=\ket{0x}$ and
  CNOT$\ket{1x}=\ket{1\overline x}$. Therefore, $\mathrm{CNOT}(\alpha.\ket
  0+\beta.\ket 1)\ket 0=\alpha.\mathrm{CNOT}\ket{00}+\beta.\mathrm{CNOT}\ket{10}=\alpha.\ket{00}+\beta.\ket{11}$.} with an ancillary qubit~$\ket 0$. Indeed,
  \begin{align*}
    (\lambda x{:}\B.x\times x)\tfrac 1{\sqrt 2}.({\ket 0}+{\ket 1})
    &\xlra{\rlinscalr}
    \tfrac 1{\sqrt 2}.(\lambda x{:}\B.x\times x)({\ket 0}+{\ket 1})\\
    &\xlra{\rlinr}
    \tfrac 1{\sqrt 2}.({(\lambda x{:}\B.x\times x)\ket 0}+{(\lambda x{:}\B.x\times x)\ket 1})\\
    &\xlra{\beta_b}
    \tfrac 1{\sqrt 2}.({\ket 0\times\ket 0}+{(\lambda x{:}\B.x\times x)\ket 1})\\
    &\xlra{\beta_b}
    \tfrac 1{\sqrt 2}.({\ket 0\times\ket 0}+{\ket 1\times\ket 1})
  \end{align*}

  The type derivation is as follows:
  \[
    \infer[^{\Rightarrow_{ES}}]{\vdash(\lambda x{:}\B.x\times x)\frac 1{\sqrt 2}.(\ket 0+\ket 1):S{\B^2}}
    {
      \infer[^{S_I}]{\vdash\lambda x{:}\B.x\times x:S(\B\Rightarrow\B^2)}
      {
	\infer[^{\Rightarrow_I}]{\vdash\lambda x{:}\B.x\times x:\B\Rightarrow\B^2}
	{
	  \infer[^{\times_I}]{x:\B\vdash x\times x:\B^2}
	  {
	    \infer[^\tax]{x:\B\vdash x:\B}{}
	    &
	    \infer[^\tax]{x:\B\vdash x:\B}{}
	  }
	}
      }
      &
      \infer[^{\alpha_I}]{\vdash\frac 1{\sqrt 2}.({\ket 0}+{\ket 1}):S(\B)}
      {
	\infer[^{+_I}]{\vdash{\ket 0}+{\ket 1}:S(\B)}
	{
	  \infer[^{S_I}]{\vdash\ket 0:S(\B)}{\infer[^{\tax_{\ket 0}}]{\vdash\ket 0:\B}{}}
	  &
	  \infer[^{S_I}]{\vdash\ket 1:S(\B)}{\infer[^{\tax_{\ket 1}}]{\vdash\ket 1:\B}{}}
	}
      }
    }
  \]
\end{example}

\begin{example}
  A Hadamard gate\footnote{The Hadamard quantum gate is such that $H\ket
    0=\frac{1}{\sqrt 2}(\ket 0+\ket 1)$ and $H\ket 1=\frac 1{\sqrt 2}(\ket
    0-\ket 1)$.} can be implemented by $H=\lambda x{:}\B.\ite x{\ket -}{\ket +}$,
  where $\ket +=\frac 1{\sqrt 2}.\ket 0+\frac 1{\sqrt 2}.\ket 1$ and
  $\ket -=\frac 1{\sqrt 2}.\ket 0-\frac 1{\sqrt 2}.\ket 1$. Therefore,
  $H:\B\Rightarrow S(\B)$ and we have $H\ket 0\lra^*\ket +$ and $H\ket 1\lra^*\ket
  -$.
\end{example}

Correctness has been established in previous works for slightly different
versions of \LambdaS, except for the case of confluence, which has only been
proved for Lineal. Lineal can be seen as an untyped fragment without several
constructions (in particular, without measurement), extended with higher-order computation. The proof of confluence for
\LambdaSM is delayed to future work, using the development of probabilistic
confluence from~\cite{DiazcaroMartinezLSFA17}. The proof of Subject Reduction
and Strong Normalization are straightforward modifications from the proofs of
the different presentations of \LambdaSM.

\begin{theorem}[Confluence of Lineal, \citaCR]
  Lineal is confluent.
  \qed
\end{theorem}

\begin{theorem}[Subject reduction on closed terms, \citaSR]
  For any closed terms $t$ and $r$ and type $A$, if $t\lra r$ and $\vdash t:A$,
  then $\vdash r:A$.
  \qed
\end{theorem}
\begin{theorem}[Strong normalization, \citaSN]
  If $\Gamma\vdash t:A$ then $t$ is strongly normalizing.
  \qed
\end{theorem}

\paragraph{Why first order.}
The restriction on functions to be first order answers a technical issue with
respect to the no-cloning property on quantum computing. \LambdaSM is meant for quantum computing, and,
in quantum computing, there is no universal cloning machine. Defining an affine
type system, as we did, we avoid cloning machines such as $\lambda
x^{S(\B)}.x\times x$, which cannot be typed. However, with high order it would
be possible to encode a cloning machine by encapsulating the term to be cloned
inside a lambda abstraction, as in the following example:
\[
  \lambda y^{S(\B)}.\left(\lambda x^{\B\Rightarrow S(\B)}.(x\ket 0)\times (x\ket
  0)\right)\left(\lambda z^\B.y\right)
\]
First order ensures this cannot be done.

Since the calculus is first order, it adds atomic terms ($\ket 0$ and $\ket 1$),
and so, no need to encode those.
Therefore, products of functions are not needed either, this is the reason why \LambdaSM does not include them.

\paragraph{More intuitions.}
Despite that \LambdaSM has been defined with quantum computation in mind, it can
be seen just as a calculus to manipulate vector spaces. This is the feature that
we want to highlight in this work, which is more general than just quantum
computing.
In particular, the derivation of $v+(-1).v+3.w+(-2).w = w$ given in the
introduction can be replicated with the rules from \LambdaS as follows.
\begin{align*}
  (-1).v+v+3.w+(-2).w
  &\xlra{\rfacto} 0.v+3.w+(-2).w\\
  &\xlra{\rfact} 0.v+1.w\\
  &\xlra{\rzero} \z[A]+1.w\\
  &\xlra{\rneut} 1.w\\
  &\xlra{\runit} w
\end{align*}
Also, with the help of the casting, we can write step by step distributions
between $+$ and $\otimes$. In \LambdaS we write $\times$ instead of $\otimes$ because it does not
behave as a $\otimes$ unless it is preceded by a casting $\Uparrow$. For
example, $u\otimes (v+w) = (u\otimes v)+(u\otimes w)$, but $u\times (v+w)$ does
not reduce to $(u\times v)+(u\times w)$, unless the casting is present, in which
case we have
\begin{align*}
  \Uparrow_\ell u\times(v+w)
  &\xlra{\rdistsuml} \Uparrow_\ell u\times v + \Uparrow_\ell u\times w\\
  &\xlra{\rcaneutl} (u\times v)+(u\times w) 
\end{align*}

\section{A categorical construction for algebraic manipulation}\label{sec:Category}
\subsection{Preliminaries}
In this section, we recall certain basic concepts of the theory of categories and we
establish a common notation that will help to define our work platform.
For general preliminaries and notations on categories we refer to~\cite{MacLane98}. 

\begin{definition}
    A \emph{symmetric monoidal} category, also called \emph{tensor} category, is
a category $\cV$ with an identity object $I\in\cV$, a bifunctor
$\otimes:\cV\times\cV\rightarrow \cV$ and natural isomorphisms $\lambda:A\otimes
I\rightarrow A$, $\rho:I\otimes A\rightarrow A$, $\alpha:A\otimes (B\otimes
C)\rightarrow (A\otimes B)\otimes C$, $\sigma:A\otimes B\rightarrow B\otimes A$
satisfying appropriate coherence axioms.
    
    A \emph{symmetric monoidal closed} category is a symmetric monoidal category
$\cV$ for which each functor $-\otimes B:\cV\rightarrow\cV$ has a right adjoint
$[B,-]:\cV\rightarrow\cV$, i.e., $\cV(A\otimes B,C)\cong\cV(A,[B,C])$.
\end{definition}
\begin{definition}
A \emph{Cartesian category} is a category admitting finite products (that is, products of
a finite family of objects). Equivalently, a Cartesian category is a category
admitting binary products and a terminal object (the product of the empty family
of objects). A Cartesian category can be seen as a symmetric monoidal category
with structural maps defined in an obvious way.

A \emph{Cartesian closed category} is a Cartesian category $\cC$ which is closed
as a symmetric monoidal category.
\end{definition}

\begin{definition}
    A \emph{symmetric monoidal functor} $(F,m_{A,B},m_I)$ between symmetric
monoidal categories $(\cV,\otimes,I,\alpha,\rho,\lambda,\sigma)$ and
$(\cW,\otimes',I',\alpha',\rho',\lambda',\sigma')$ is a functor
$F:\cV\rightarrow\cW$ equipped with morphisms $m_{A,B}:FA\otimes'FB\rightarrow
F(A\otimes B)$ natural in $A$ and $B$ , and for the units morphism
$m_I:I'\rightarrow F(I)$ satisfying some coherence axioms. A monoidal functor is
said to be \emph{strong} when $m_I$ and $m_{A,B}$ for every $A$ and $B$ are
isomorphisms and \emph{strict} when all the $m_{A,B}$ and $m_I$ are identities.
\end{definition}

\begin{definition}
    \label{MONOIDAL NATURAL TRANFORMATION}
    A \emph{monoidal natural transformation} $\theta:(F,m)\rightarrow (G,n)$
between monoidal functors is a natural transformation $\theta_A:FA\rightarrow
GA$ such that the following axioms hold:
$n_{A,B}\circ(\theta_{A}\otimes'\theta_{B})=\theta_{A\otimes B}\circ m_{A,B}$
and $\theta_I\circ m_I=n_I$.
\end{definition}

\begin{definition}
  Let $(\cV,\otimes,I)$ and $(\cW, \otimes', I')$ be monoidal categories. 
  We say that $((F,m),(G,n),\eta,\varepsilon)$ is a \emph{monoidal adjunction} if 
  \begin{itemize}
  \item $(F,G,\eta, \varepsilon)$ is an adjunction. 
  \item $(F,m),(G,n)$ are monoidal functors 
    \begin{center}
      \begin{tikzcd}[column sep=3mm]
        (\cV,\otimes,\I)\ar[rr,bend left,"{(F,m)}"] & \bot & (\cW,\otimes',\I')\ar[ll,bend left,"{(G,n)}"]
      \end{tikzcd}
    \end{center}
  \item $\eta:Id\Rightarrow G\circ F$ and  $\varepsilon:F\circ G\Rightarrow Id$ are monoidal natural transformations, as defined in Definition~\ref{MONOIDAL NATURAL TRANFORMATION}, 
  \end{itemize}
\end{definition}

\begin{definition}
  A \emph{preadditive category} is a category $\cC$ together with an abelian
  group structure on each set $\cC(A,B)$ of morphisms, in such a way that the
  composition mappings
  \begin{align*}
    c_{ABC}:\cC(A,B)\times\cC(B,C)&\longrightarrow \cC(A,C)\\
    (f,g) &\mapsto g\circ f
  \end{align*}
  are group homomorphisms in each variable. We shall write the group structure additively.

  An \emph{additive category} is a preadditive category with a zero object and a binary biproduct.
\end{definition}

\begin{definition}
  An \emph{additive symmetric monoidal closed category} is a category
  $(\cV,\otimes,\oplus)$ such that $(\cV,\otimes)$ is a symmetric monoidal closed
  category, $(\cV,\oplus)$ is an additive category, and $\otimes$ is
  bi-additive.
\end{definition}

\subsection{Adjunction for algebraic manipulation}
In this section we give the main categorical construction on this paper, which
is the adjunction for algebraic manipulation.
\begin{definition}\label{def:adjuncion}
  An \emph{adjunction for algebraic manipulation} is a monoidal adjunction
  \begin{center}
    \begin{tikzcd}[column sep=3mm]
      (\cC,\times,1)\ar[rr,bend left,"{(S,m)}"] &\bot & (\cV,\otimes,\I)\ar[ll,bend left,"{(U,n)}"]
    \end{tikzcd}
  \end{center}
  where
  \begin{itemize}
  \item $(\cC,\times,1)$ is a Cartesian closed category with $1$ as a terminal object.
  \item $(\mathcal V,\otimes,\oplus,\rho,\lambda,\sigma)$ is an additive
    symmetric monoidal closed category.
    
  \item The following axiom (\axiomz) is satisfied for any $f$ and $g$
    \begin{center}
      \begin{tikzcd}[row sep=10pt]
        & UV \arrow[rd, "U\mathbf 0"]  &     \\
        A \arrow[ru, "f"] \arrow[rd,swap, "g"] & & UW \\
        & UV' \arrow[ru, "U\mathbf 0",swap]  &
      \end{tikzcd}
    \end{center}
    where $\mathbf 0$ is the zero morphism of the additive category.
  \item The following axiom (\axiomd) is satisfied
    \begin{center}
      \begin{tikzcd}
        UV\times UV\times UW \arrow[r, "p\times\Id"] \arrow[d, "d"] & U(V\oplus V)\times UW \arrow[d, "n"]         \\
        UV\times UW\times UV\times UW \arrow[d, "n\times n"]        & U((V\oplus V)\otimes W) \arrow[d, "U\delta"] \\
        U(V\otimes W)\times U(V\otimes W) \arrow[r, "p"]            & U((V\otimes W)\oplus (V\otimes W))          
      \end{tikzcd}
    \end{center}
    where
    \begin{itemize}
    \item The map $d:UV\times UV\times UW\longrightarrow UV\times UW\times
      UV\times UW$ is defined by $(\Id\times\sigma\times\Id)\circ(\Id\times\Delta)$.
    \item The map $\delta$ is an isomorphism determined by the fact that
      $\otimes$ has a right adjoint.
      Explicitly, $\delta:(V\oplus V)\otimes W\longrightarrow (V\otimes
      W)\oplus(V\otimes W)$ is given by $\delta=\langle \pi_1\otimes\Id,\pi_2\otimes\Id \rangle$.
    \item The map $p$ is an isomorphism determined by the preservation of
      product of the functor $U$ given by the fact that $U$ has a left adjoint.
      Explicitly, $p_{V,W}:UV\times UW\longrightarrow U(V\oplus W)$  is given by
      $p=\phi(\langle \phi^{-1}(\pi_1),\phi^{-1}(\pi_2) \rangle_{{}_{\mathcal V}})$
      where $\phi:\mathcal V(S(UV\times UW),V\oplus W)\cong\mathcal C(UV\times
      UW,U(V\oplus W))$, in which $\pi_1:UV\times UW\longrightarrow UV$ and
      $\pi_2:UV\times UW\longrightarrow UW$ are the projection maps.
    \end{itemize}

  \item There exists an object $\B\in|\cC|$ and maps $i_1$, $i_2$ such that for
    every ${1}\xlra f A$ and ${1}\xlra g A$, there exists a unique map $\home fg$
    such that the following diagram commutes
    \begin{center}
      \begin{tikzcd}
        {1}\ar[r,"i_1"]\ar[rd,"f"'] & \B\ar[d,near start,"\home fg"] &{1}\ar[l,"i_2"']\ar[ld,"g"]\\
        &A &
      \end{tikzcd}
    \end{center}
  \end{itemize}
\end{definition}

\begin{remark}\label{rmk:categoria}~
  \begin{itemize}
  \item The object $\B$ allows us to represent the type $\B$, and the map
    $\home fg$ to interpret the \emph{if} construction (Definition~\ref{def:if}).
  \item $\cC$ is a Cartesian closed category where $\eta^A$ is the unit and $\varepsilon^A$ is the counit of $-\times A\dashv[A,-]$, from which we can define
    the curryfication ($\mathsf{curry}$) and un-curryfication
    ($\mathsf{uncurry}$) of any map.
  \item The adjunction $S\dashv U$ gives rise to a monad $(T,\eta,\mu)$ in the category
    $\cC$, where $T=US$, $\eta:\Id\to T$ is the unit of the
    adjunction, and using the counit $\varepsilon$, we obtain 
    $\mu=U\varepsilon_S:TT\to T$, satisfying unity and associativity laws (see~\cite{MacLane98}).
  \item Remember that in an additive category the morphism factoring through
    the zero object, i.e.~the zero morphisms $\mathbf 0$ are exactly the identities for the
    group structure in each $\mathcal V(A,B)$ for every $A$ and $B$.
  \item Notice that since the tensor $\otimes$ is bi-additive, it satisfies
    that $f\otimes\mathbf 0=\mathbf 0\otimes f=\mathbf 0$ for every $f$.
  \end{itemize}
\end{remark}

Intuitively, the axiom \axiomz\ carries the absorbing property of the zero morphism, to
the category $\mathcal C$. Indeed, an analogous situation to this axiom, in the category
$\mathcal V$ is
    \begin{center}
      \begin{tikzcd}[row sep=10pt]
        & V \arrow[rd, "\mathbf 0"]  &     \\
        W \arrow[rr,red,dashed,"\mathbf 0"]\arrow[ru, "f"] \arrow[rd,swap, "g"] & & W' \\
        & V' \arrow[ru, "\mathbf 0",swap]  &
      \end{tikzcd}
    \end{center}
    which is valid since the dashed arrow makes the diagram commute. However,
    using the functor $U$ we would obtain
    \begin{center}
      \begin{tikzcd}[row sep=10pt]
        & UV \arrow[rd, "U\mathbf 0"]  &     \\
        UW \arrow[rr,red,dashed,"U\mathbf 0"]\arrow[ru, "Uf"] \arrow[rd,swap, "Ug"] & & UW' \\
        & UV' \arrow[ru, "U\mathbf 0",swap]  &
      \end{tikzcd}
    \end{center}
    which is less general than \axiomz. Indeed, in \axiomz\ we allow the domain
    to be any $A$, and not necessarily of the form $UW$, capturing the absorbing 
    property of a zero morphism, but in $\mathcal C$.

The axiom \axiomd\ gives us explicitly the intuition developed in the
introduction. In the Cartesian category $\mathcal C$ we do not have all the
structure and properties as in the additive symmetric monoidal closed category
$\mathcal V$. However, we can mimic the distributivity property of $\otimes$
with respect to $\oplus$ by simply duplicating the last element and performing a
permutation, i.e.,~$\langle \langle a,b \rangle,c \rangle\mapsto\langle \langle a,b \rangle,\langle c,c \rangle \rangle\mapsto\langle \langle
a,c \rangle,\langle b,c \rangle \rangle$ mimic $(a\oplus b)\otimes c=(a\oplus
c)\oplus(b\otimes c)$.
While this property may be trivial when concrete categories are given, such as
$\mathbf{Set}$ for the Cartesian category and $\mathbf{Vec}$ for the additive
symmetric monoidal closed category, we have to axiomatize it in this abstract framework.

\begin{example}\label{ex:SetVec}
One concrete model for \LambdaS has been briefly mentioned, which is the one
presented in~\cite{DiazcaroMalherbeLSFA18,DiazcaroMalherbe2020}: an adjunction
for algebraic manipulation where $\cC=\mathbf{Set}$ and $\mathcal V=\mathbf{Vec}$.

  We must prove that those categories satisfy the requirements from
  Definition~\ref{def:adjuncion}.
  \begin{itemize}
  \item \axiomz\ is satisfied for any $f$ and $g$ since the zero morphism is absorbing in $\mathbf{Vec}$ and this property is preserved by $U$.
  \item \axiomd\ is satisfied since the concrete maps are the following:
    \begin{align*}
      \langle  a,b,c\rangle &\mapsto\langle \langle a,b \rangle, c \rangle\mapsto \langle a,b \rangle\otimes c  \mapsto
      \langle a\otimes c,b\otimes c \rangle
      \\
      \langle a,b ,c\rangle &\mapsto\langle a,c,b,c \rangle
      \mapsto
      \langle a\otimes c,b\otimes c \rangle
    \end{align*}
  \item We identify the object $\B\in|\mathbf{Set}|$ with $\{\ket 0,\ket 1\}$,
    which satisfies the required properties. 
  \end{itemize}
\end{example}

\begin{example}
More general, a family of examples is obtained by replacing {\bf Vec} by a category ${\bf Mod}_R$ of modules on a commutative ring $R$. The proof is essentially the same as the previous example.
\end{example}

\begin{example}
Let $(\cC,\times,1)$ be the category of sets $\mathbf{Set}$ and $(\mathcal
V,\otimes,\oplus)$ be the category $\mathbf{Ab}$ of abelian groups and group
homomorphisms. These categories are cartesian and symmetric monoidal closed
respectively. The tensor in $\mathbf{Ab}$
is defined by a universal property, concretely, is the quotient of the free abelian group on the direct sum determined by the subgroup that satisfies some well-know relations and where $I=\mathbb{Z}$. 
Also, $\mathbf{Ab}$
is an additive category (see~\cite{kn:B}). The functor $S$ is the free
construction $S(X)= \{\{z_x\}_{x\in X}: z_x\in \mathbb{Z}; |\{x:z_x\neq
0\}|<\omega\}$ and $U$ is the forgetful functor $U:(\mathbf {Ab} ,\otimes
_{\mathbb {Z} },\mathbb {Z} )\rightarrow (\mathbf {Set} ,\times ,\{\ast \}) $
where the mediating map $n_{A , B} : U ( A )\times U ( B )\rightarrow U ( A
\otimes B )$ sends $(a, b) \mapsto a \otimes b$ the map $n_I:\ast \mapsto 1$.
\end{example}

\begin{example}
 Let $C$ be a cocommutative cosemisimple $\mathbb{K}$-coalgebra, where
$\mathbb{K}$ is a field. We consider $(\cC,\times,1)$ to be
$\cC=\mathbf{Coalg/C}$ as the slice category of $\mathbb{K}$-cocommutative
coalgebras and morphisms of coalgebras defined as follows: objects are morphisms
of coalgebras with codomain in $C$, if $\phi:D\to C$ and $\psi:E\to C$ are
morphisms of coalgebras (as object in the slice category), morphisms
$f:(\phi)\to (\psi)$ correspond to coalgebra morphisms $f:D\to E$ such that
$\psi \circ f=\phi$. Cartesian product is given by pullbacks and $1=id_C$ the
identity morphism.

The structure $(\cV,\otimes,\oplus)$ is defined as follows: $\cV$ is the
additive (abelian) category of $C$-comodules ${\bf \mathcal{M}^C}$ (see
\cite{kn:DR}) such that the tensor is defined by an equalizer: Let $(V,v)$ and
$(W,w)$ be $C$-comodules, where $v$ and $w$ are right coactions. There is a
structure of $C$-comodule denoted by $V\otimes^C W$ in the vector space
generated by $\{x\otimes y \in V\otimes W\mid v(x)\otimes y=x\otimes \tau\left
(w(y)\right )\}$ where the coaction is defined by $\delta(x\otimes y)=x\otimes
w(y)$ (see \cite{kn:GP,HaimMalherbe16}). If $C$ is a cocommutative coalgebra, the
category $({\bf\mathcal{M}^C},\otimes^C,C)$ is symmetric monoidal. Moreover, it
is closed if and only if $C$ is cosemisimple. (see \cite{kn:GP,HaimMalherbe16}).
We define $S:\mathbf{Coalg/C} \to {\bf\mathcal{M}^C}$ to be the functor that
takes each object $\phi:D\rightarrow C$ to the comodule $(D,d)$, where $d:D\to
D\otimes C$ is the coaction defined by $d=(id_D\otimes \phi) \circ \Delta_D$ and
each morphism to its underlying morphism in ${\bf\mathcal{M}^C}$. This functor
is a strong monoidal functor and the existence of a right adjoint follows from
the special adjoint functor theorem (see \cite{MacLane98,kn:GP}) which implies
the existence of a monoidal adjunction (see \cite{kn:K}).
\end{example}

\section{Denotational semantics}\label{sec:DenSem}
\subsection{Definitions}\label{sec:Defs}
In this section we give the denotational semantics of \LambdaS by using the
adjunction for algebraic manipulation defined in the previous section.

\begin{definition}
  Types are interpreted in the category $\cC$, as follows:
  \begin{align*}
    \sem{\B} & =\B\\
    \sem{\Psi\Rightarrow A} &=\home{\sem{\Psi}}{\sem{A}}\\
    \sem{S(A)} &= US\sem{A}\\
    \sem{\Psi\times \Phi} &=\sem \Psi\times\sem \Phi
  \end{align*}
\end{definition}

\begin{remark}
  To avoid cumbersome notation, we will use the following convention: We write
  directly $USA$ for $\sem{S(A)}=US\sem A$ and $A$ for $\sem{A}$, when there
  is no ambiguity.

  In addition, we abuse notation and write $\sem{\Gamma}$ for the product of the
  interpretations of all the types in $\Gamma$. E.g.~If
  $\Gamma=x_1:\Psi_1,\dots,x_n:\Psi_n$, then
  $\sem\Gamma=\sem{\Psi_1}\times\cdots\times\sem{\Psi_n}$. We may write
  directly $\Gamma$ for $\sem\Gamma$, when there is no ambiguity.
\end{remark}

Before giving the interpretation of typing derivation trees in the model, we
need to define certain maps which will serve to implement some of the
constructions in the language.

To implement the \emph{if} construction we define the following map.
\begin{definition}\label{def:if}
  Given $t,r\in\cC(\Gamma,A)$ there exists a map $f_{t,r}\in\cC(\B,\home \Gamma A)$ defined by $f_{t,r}= \home{\hat{t}}{\hat{r}}$ where
  $\hat{t}\in\cC(1,\home \Gamma A)$ and  $\hat{r}\in\cC(1,\home \Gamma A)$
  are given by $\hat{t}=\mathsf{curry}(t\circ\pi_\Gamma)$ and
  $\hat{s}=\mathsf{curry}(r\circ\pi_\Gamma)$.
  \begin{center}
    \begin{tikzcd}
      {1}\ar[r,"i_1"]\ar[rd,"\hat t",swap] & \B\ar[d,near start,"f_{t,r}",swap] &{1}\ar[l,"i_2",swap]\ar[ld,"\hat r"]\\
      &\home\Gamma A &
    \end{tikzcd}
  \end{center}
\end{definition}

The sum in \LambdaS will be implemented internally by the map $\nabla$ issued from the 
universal property of $\oplus$. This way, we define a sum $\hat +$ in $\mathcal
C$ as follows.

\begin{definition}
  The map $\hat +$ is defined by
  \begin{center}
    \begin{tikzcd}
      UV\times UV\ar[rr,"p"]\ar[rd,"\hat +"] && U(V\oplus
      V)\ar[ld,"U\nabla"] \\
      &UV &
    \end{tikzcd}
  \end{center}
  where $p$ has been defined in Definition~\ref{def:adjuncion}.
\end{definition}

The sum $\hat +$ on $USUV\times USUV$ is performed in the following way
$USUV\times USUV\xlra{p} U(SUV\oplus SUV)\xlra{U\nabla} USUV$. Notice that the
map $\nabla$ used in this construction is fundamentally different from the map
$\nabla$ defined over $V\oplus V$. In order to perform all the sums at the same
``level'', we would need to do $USUV\times USUV\xlra{g_1} US(UV\times
UV)\xlra{USp} USU(V\oplus V)\xlra{USU\nabla}USUV$, where $g_1$ factorizes the
first $US$. We can generalize this idea to $(US)^kUV\times (US)^kUV$ with a map $g_k$
factorizing the first $k$ $(US)$s. Such a map is defined as follows.

\begin{definition}
  The map $g_k:(US)^{k}UV\times (US)^{k}UW\to (US)^k(UV\times UW)$
  is defined by
  \begin{align*}
    g_0 &= \Id\\
    g_k &=
    ((US)^{k-1}Um)\circ ((US)^{k-1}n)\circ 
    ((US)^{k-2}Um)\circ ((US)^{k-2}n)\circ 
    \dots\circ
    (Um)\circ n
  \end{align*}
\end{definition}
\begin{example}
  We can define a map $\mathsf{sum}$ on $USUSUV\times USUSUV$ by using the sum
  $\hat +$ on $UV$ as $USUS\hat +\circ g_2$, where
  $g_2=(USUm)\circ (USn)\circ (Um)\circ n$. This gives the following diagram
  \begin{center}
    \begin{tikzcd}
      USUSUV\times USUSUV\ar[r,"n"]\ar[dd,red,dashed,"g_2",bend left=83]\ar[d,"\mathsf{sum}"] & U(SUSUV\otimes SUSUV)\ar[d,"Um"]\\
      USUSUV & US(USUV\times USUV)\ar[d,"USn"]\\
      USUS(UV\times UV)\ar[u,swap,"USUS\hat +"]& USU(SUV\otimes SUV)\ar[l,swap,"USUm"]
    \end{tikzcd}
  \end{center}
\end{example}

The aim of the casting $\Uparrow_r$ is to implement the distributivity property
in $\mathcal C$ by mapping
$USA\times B$ into $US(A\times B)$. We want to perform such a property by
using the underlying distributivity property in $\mathcal V$.

In fact, the casting is defined more generally between $US(USA\times B)$ and
$US(A\times B)$. A map denoting $\Uparrow_r$ can be defined as follows.

\begin{definition}
  Let $\Uparrow_r^1$ be defined as follows
  \begin{center} 
    \begin{tikzcd}
      US(USA\times B) \arrow[rr, "\Uparrow_r^1"] \arrow[d, "US(\Id\times\eta)"] & &US(A\times B)                    \\
      US(USA\times USB) \arrow[r, "USn"] & USU(SA\otimes SB)\arrow[r,"USUm"]& USUS(A\times B) \arrow[u, "\mu"]
    \end{tikzcd}
  \end{center}

  We generalize $\Uparrow_r^1$ to the case $US((US)^kA\times B)$ with the map
  $\Uparrow^k_r:US((US)^kA\times B)\rightarrow
  US(A\times B)$ is defined by
  \[
    \Uparrow^k_r = \Uparrow_r^1\circ\cdots\circ\Uparrow_r^1
  \]

  Analogously, we define $\Uparrow^k_\ell:US(A\times (US)^kB)\rightarrow US(A\times B)$.
\end{definition}

Using all the previous definitions, we can finally give the interpretation of a
type derivation tree in our model.
If $\Gamma\vdash t:A$ with a derivation $\pi$, we
write it generically $\sem\pi$ as
$\Gamma\xlra{t_A} A$. When $A$ is clear from the context, we may write just $t$
for $t_A$.
Also, each interpretation depends on a choice of scalars, i.e., a function $c :
\mathcal C\to\mathcal V(I,I)$; without loss of generality we denote the values $c(\alpha)$ with the same letter $\alpha$.
\begin{definition}
  If $\pi$ is a type derivation tree, we define $\sem\pi$ inductively as follows,
  \begin{align*}
    &\sem{\vcenter{\infer[^{\mathsf{Ax}}]{\Gamma^\B,x:\Psi\vdash x:\Psi}{}}} = \Gamma^\B\times\Psi\xlra{{!}\times\Id}{1}\times\Psi\approx\Psi\\
    &\sem{\vcenter{\infer[^{Ax_{\mathbf 0}}]{\Gamma^\B\vdash\z:S(A)}{}}} = \Gamma^\B\xlra{!}{1}\xlra{\eta}US1\xlra{U\mathbf 0}USA\\
    &\sem{\vcenter{\infer[^{Ax_{\ket 0}}]{\Gamma^\B\vdash\ket 0:\B}{}}}=\Gamma^\B\xlra{!}{1}\xlra{i_1}\B\\
    &\sem{\vcenter{\infer[^{Ax_{\ket 1}}]{\Gamma^\B\vdash\ket 1:\B}{}}}=\Gamma^\B\xlra{!}{1}\xlra{i_2}\B\\
    &\sem{\vcenter{\infer[^{\alpha_I}]{\Gamma\vdash \alpha.t:S^m(A)}{\Gamma\vdash t:S^m(A)}}} =
      \begin{aligned}[t]
        &\Gamma \xlra{t}(US)^mA \xlra{(US)^{m-1}U\lambda} (US)^{m-1}U(SA\otimes\I)\\
        &\xlra{(US)^{m-1}U(\Id\otimes\alpha)} (US)^{m-1}U(SA\otimes\I)\\
        &\xlra{(US)^{m-1}U\lambda^{-1}} (US)^mA
      \end{aligned}\\
    &\sem{\vcenter{\infer[^{+_I}]{\Gamma,\Delta,\Xi^\B\vdash t+r:S^m(A)}{\Gamma,\Xi^\B\vdash t:S^m(A) &  \Delta,\Xi^\B\vdash r:S^m(A)}}} =
                                                                                                        \begin{aligned}[t]
                                                                                                          &\Gamma\!\times\!\Delta\!\times\!\Xi^\B\xlra d\Gamma\!\times\!\Xi^\B\!\times\!\Delta\!\times\!\Xi^\B\\
                                                                                                          &\xlra{t\times r} (US)^mA\times (US)^mA\\
                                                                                                          &\xlra {g_{m-1}} (US)^{m-1}(USA\times USA)\\
                                                                                                          &\xlra{(US)^{m-1}\hat +} (US)^mA
                                                                                                        \end{aligned}\\
    &\sem{\vcenter{\infer[^{S_I}]{\Gamma\vdash t:S(A)}{\Gamma\vdash t:A}}}=\Gamma\xlra{t}A\xlra{\eta}USA\\
    &\sem{\vcenter{\infer[^{\mathsf{If}}]{\Gamma\vdash\ite{}{t}{r}:\B\Rightarrow A} {\Gamma\vdash t:A & \Gamma\vdash r:A}}}= \Gamma\xlra{\mathsf{curry}(\mathsf{uncurry}(f_{t,r})\,\circ\,\mathsf{swap})}[\B,A]\\
    &\sem{\vcenter{\infer[^{\Rightarrow_I}] {\Gamma\vdash\lambda x{:}\Psi.t:\Psi\Rightarrow A} {\Gamma,x:\Psi\vdash t:A}}} =\Gamma\xlra{\eta^\Psi}[\Psi,\Gamma\times\Psi]\xlra{[\Id,t]}[\Psi,A]\\
    &\sem{\vcenter{\infer[^{\Rightarrow_E}]{\Delta,\Gamma,\Xi^\B\vdash tu:A} {\Delta,\Xi^\B\vdash u:\Psi & \Gamma,\Xi^\B\vdash t:\Psi\Rightarrow A}}} =
                                                                                                           \begin{aligned}[t]
                                                                                                             &\Delta\times\Gamma\times\Xi^\B\xlra d \Delta\times\Xi^\B\times\Gamma\times\Xi^\B\\
                                                                                                             &\xlra{u\times t}\Psi\times[\Psi,A]\xlra{\varepsilon^\Psi}A
                                                                                                           \end{aligned}\\
    &\sem{\vcenter{\infer[^{\Rightarrow_{ES}}]{\Delta,\Gamma,\Xi^\B\vdash
      tu:S(A)} {\Delta,\Xi^\B\vdash u:S\Psi & \Gamma,\Xi^\B\vdash t:S(\Psi\Rightarrow A)}}} =
                                              \begin{aligned}[t]
                                                &\Delta\!\times\!\Gamma\!\times\!\Xi^\B\xlra d\Delta\!\times\!\Xi^\B\!\times\!\Gamma\!\times\!\Xi^\B\\
                                                &\xlra{u\times t}US\Psi\times US[\Psi,A]\\
                                                &\xlra{n} U(S\Psi\otimes S([\Psi,A]))\\
                                                &\xlra{Um} US(\Psi\times[\Psi,A])\\
                                                &\xlra{US\varepsilon^ \Psi}USA
                                              \end{aligned}\\
    &\sem{\vcenter{\infer[^{\times_I}]{\Gamma,\Delta,\Xi^\B\vdash t\times u:\Psi\times \Phi} {\Gamma,\Xi^\B\vdash t:\Psi & \Delta,\Xi^\B\vdash u:\Phi}}} =
                                                                                                                           \Gamma\!\times\!\Delta\!\times\!\Xi^\B\xlra d\Gamma\!\times\!\Xi^\B\!\times\!\Delta\!\times\!\Xi^\B\xlra{t\times u} \Psi\!\times\! \Phi\\
    &    \sem{\vcenter{\infer[^{\times_{Er}}]{\Gamma\vdash\head\ t:\B}{\Gamma\vdash t:\B^n}}} =\Gamma\xlra t\B^n\xlra{\pi_1}\B\\
    &    \sem{\vcenter{\infer[^{\times_{El}}]{\Gamma\vdash\tail\ t:\B^{n-1}}{\Gamma\vdash t:\B^n}}} =\Gamma\xlra t\B^n\xlra{\pi_2}\B^{n-1}\\
    &    \sem{\vcenter{\infer[^{\Uparrow_r}]{\Gamma\vdash\Uparrow_r  t:S(\Psi\times \Phi)} {\Gamma\vdash t:S(S^k(\Psi)\times \Phi)}}} =
      \Gamma\xlra t US((US)^k\Psi\times \Phi)\xlra{\Uparrow_r^k} {US(\Psi\times\Phi)}\\
    &    \sem{\vcenter{\infer[^{\Uparrow_\ell}]{\Gamma\vdash\Uparrow_\ell  t:S(\Psi\times \Phi)} {\Gamma\vdash t:S(\Psi\times S^k(\Phi))}}}  =
      \Gamma\xlra t US(\Psi\times (US)^k\Phi)\xlra{\Uparrow_\ell^k} US(\Psi\times\Phi)
  \end{align*}
\end{definition}

\subsection{Properties}\label{sec:properties}
In this section we prove that the given denotational semantics is sound
(Theorem~\ref{thm:soundness}) and complete (Theorem~\ref{thm:Completeness}).

Proposition~\ref{prop:eqDer} allows us to write the semantics of a sequent,
independently of its derivation. Hence, due to this independence, we can write
$\sem{\Gamma\vdash t:A}$, without ambiguity.
\begin{proposition}
  [Independence of derivation]
  \label{prop:eqDer}
  If $\Gamma\vdash t:A$ can be derived with two different derivations $\pi$ and
  $\pi'$, then $\sem{\pi}=\sem{\pi'}$
\end{proposition}
\begin{proof}
  Without taking into account rules $\Rightarrow_E$, $\Rightarrow_{ES}$ and
  $S_I$,
  the typing system is syntax directed.
  In the case of the application (rules $\Rightarrow_E$ and $\Rightarrow_{ES}$),
  they can be interchanged only in a few specific cases.

  Hence, we give a rewrite system on trees such that each time a rule $S_I$
  can be applied before or after another rule, we chose a direction to rewrite the
  tree to one of these forms. Similarly, we chose a direction for rules
  $\Rightarrow_E$ and $\Rightarrow_{ES}$.
  Then we prove that every rule preserves the
  semantics of the tree. This rewrite system is clearly confluent and
  normalizing, hence for each tree $\pi$ we can take the semantics of its normal
  form, and so every sequent will have one way to calculate its semantics, i.e.~as
  the semantics of the normal tree.

  The full proof is given in the appendix.
  \qed
\end{proof}

\begin{lemma}
  [Substitution]\label{lem:substitution}
  If $\Gamma',x:\Psi,\Gamma\vdash t:A$ and $\vdash r:\Psi$, then the following diagram
  commutes:
  \begin{center}
    \begin{tikzcd}
      \Gamma'\times\Gamma\ar[r,"(r/x)t"]\ar[d,"\lambda_\times"] & A\\
      \Gamma'\times{1}\times\Gamma\ar[r,"\Id\times r\times\Id"] &\Gamma'\times\Psi\times\Gamma\ar[u,"t"]
    \end{tikzcd}
  \end{center}
  That is,
  $\sem{\Gamma',\Gamma\vdash(r/x)t:A}=\sem{\Gamma',x:\Psi,\Gamma\vdash t:A}\circ(\sem{\vdash r:\Psi}\times\Id)$.
\end{lemma}
\begin{proof}
  By induction on the derivation of $\Gamma',x:\Psi,\Gamma\vdash t:A$.
  The full proof is given in the appendix.
  \qed
\end{proof}

\begin{theorem}
  [Soundness]\label{thm:soundness}
  If $\vdash t:A$, and $t\lra r$,
  then
  $\sem{\vdash t:A} = \sem{\vdash r:A}$.
\end{theorem}
\begin{proof}
  By induction on the rewrite relation, using the first derivable type for each
  term.
  The full proof is given in the appendix.
  \qed
\end{proof}

In order to prove completeness (Theorem~\ref{thm:Completeness}), we use an adaptation to
\LambdaS of Tait's proof for strong normalization~\cite{TaitJSL67} (cf.~\cite[Chapter~6]{Girard89} for reference).

\begin{definition}
  Let $\mathcal A,\mathcal B$ be sets of closed terms. We define the following operators on them:

  \begin{itemize}
  \item\textit{Closure by antireduction:} $\overline{\mathcal A}=\{t\mid
    t\lra^*r_, \textrm{ with }r\in \mathcal A \textrm{ and }FV(t)=\emptyset\}$.
  \item\textit{Product:} $\mathcal A\times \mathcal B=\{t\times u\mid t\in \mathcal
    A\textrm{ and }u\in \mathcal B\}$.
  \item\textit{Arrow:} $\mathcal A\Rightarrow \mathcal B=\{t\mid\forall u\in\mathcal
    A, tu\in \mathcal B\}$.
  \item\textit{Span:} $S\mathcal A=\{\sum_i\alpha_ir_i\mid r_i\in\mathcal A\}$
    where $\alpha r$ is a notation for $\alpha.r$ when $\alpha\neq 1$, or $1.r$
    or just $r$ when $\alpha=1$. Also, we use the convention that
    $\sum_{i=1}^1\alpha_ir_i=\alpha_ir_i$ and $0r=\z[A]$ for any $r$.
  \end{itemize}

  The set of computational closed terms of type $A$ (denoted $\com A$), is defined by 
  \[
    \begin{aligned}[t]
      \com{\B} &=\overline{\{\ket 0,\ket 1\}}\\
      \com{A\times B} &=\overline{\com A\times\com B}
    \end{aligned}
    \qquad\qquad
    \begin{aligned}[t]
      \com{\Psi\Rightarrow A} &=\overline{\com\Psi\Rightarrow\com A}\\
      \com{S(A)} &=\overline{S\com{A}}
    \end{aligned}
  \]

  A substitution $\sigma$ is valid with respect to a context $\Gamma$ (notation
  $\sigma\vDash\Gamma$) if for each $x:A\in\Gamma$, $\sigma x\in\com A$.
\end{definition}

\begin{lemma}[Adequacy]\label{lem:Adequacy}
  If $\Gamma\vdash t:A$ and $\sigma\vDash\Gamma$, then $\sigma t\in\com A$.
\end{lemma}
\begin{proof}
  By induction on the derivation of $\Gamma\vdash t:A$.
  The detailed proof can be found in the appendix.
  \qed
\end{proof}

\begin{theorem}[Completeness]\label{thm:Completeness}
  If $\sem{\vdash t:S(\B^n)}=\sem{\vdash r:S(\B^n)}$,
  then for any concrete model interpretation injective on values
  there exists $s$ such that $t\lra^* s$ and $r\lra^*s$.
\end{theorem}

\begin{proof}
  By Lemma~\ref{lem:Adequacy},
  $t\in\com{S(\B^n)}=\overline{S\com{\B^n}}=\overline{S(\overline{\B^n})}$.
  Hence, $t\lra^*\psi$, with $\vdash\psi:S(\B^n)$ and
  $\psi=\sum_i\alpha_i\ket{b_{i1}}\times\dots\times\ket{b_{in}}$. Then, by
  Theorem~\ref{thm:soundness}, $\sem{\vdash
    t:S(\B^n)}=\sem{\vdash\psi:S(\B^n)}$. Analogously, $r\lra^*\phi$ and so by
  Theorem~\ref{thm:soundness}, $\sem{\vdash
    r:S(\B^n)}=\sem{\vdash\phi:S(\B^n)}$, with
  $\phi=\sum_j\beta_j\ket{b_{j1}}\times\dots\times\ket{b_{jn}}$. Therefore,
  since $\sem{\vdash t:S(\B^n)}=\sem{\vdash r:S(\B^n)}$, we have
  $\sem{\vdash\psi:S(\B^n)}=\sem{\vdash\phi:S(\B^n)}$, and so, since the model
  is injective on values, we have $\psi=\phi=s$.
  \qed
\end{proof}

\begin{remark}
It is easy to verify that the model from Example~\ref{ex:SetVec} is injective on values, and
hence, the model is Complete (Theorem~\ref{thm:Completeness}).

On the other hand, there exist non injective concrete models. An easy example
is the degenerated concrete model where both $\mathcal C$ and $\mathcal V$ are
the terminal category (the category with a unique object and morphism), which is
trivially both, Cartesian closed, and additive symmetric monoidal closed, and
both axioms are satisfied trivially.
\end{remark}

\section{Conclusion}\label{sec:conclusion}
In this paper, we have given an abstract categorical semantics of \LambdaS, a fragment of
\LambdaSM without measurements, and we have proved that it is sound
(Theorem~\ref{thm:soundness}) and complete (Theorem~\ref{thm:Completeness}).
Such a semantics highlights the dynamics of the calculus: The algebraic
rewriting (linear distribution, vector space axioms, and typing casts rules)
emphasize the standard behavior of vector spaces, in a computational way: the
vector space axioms give rise to computational steps. We have enforced this
computational steps by interpreting the calculus into a Cartesian category $\mathcal C$,
without distributivity properties, and defining and using an adjunction for
algebraic manipulation between this category $\mathcal C$ and an additive symmetric
monoidal closed category $\mathcal V$ with all the properties needed for the vectorial
space axioms.
This way, in order to transform an element from the category $\mathcal C$, we use the
adjunction to carry these elements to $\mathcal V$, where the proper transformation
properties are in place.

As an immediate future work, we are willing to pursue a complete semantics for
quantum computing, for which we need to add back the measurement operator, and
define a notion of a norm, maybe
following~\cite{DiazcaroGuillermoMiquelValironLICS19}.

\paragraph{Acknowledgments} We thank the anonymous reviewer for the suggestion
with some examples on concrete models.

\bibliographystyle{spmpsci}
\bibliography{biblio}
 
\appendix
\section{Detailed proofs}\label{ap:appendix}
\xrecap{Proposition}{Independence of derivation}{prop:eqDer}{
  If $\Gamma\vdash t:A$ can be derived with two different derivations $\pi$ and
  $\pi'$, then $\sem{\pi}=\sem{\pi'}$
}
\begin{proof}
  Without taking into account rules $\Rightarrow_E$, $\Rightarrow_{SE}$ and
  $S_I$, 
  the typing system is syntax directed.
  In the case of the application (rules $\Rightarrow_E$ and $\Rightarrow_{SE}$),
  they can be interchanged only in few specific cases.

  Hence, we give a rewrite system on trees such that each time a rule $S_I$ 
  can be applied before or after another rule, we chose a direction to rewrite the
  three to one of these forms. Similarly we chose a direction for rules
  $\Rightarrow_E$ and $\Rightarrow_{ES}$.
  Then we prove that every rule preserves the
  semantics of the tree. This rewrite system is clearly confluent and
  normalizing, hence for each tree $\pi$ we can take the semantics of its normal
  form, and so every sequent will have one way to calculate its semantics: as
  the semantics of the normal tree.

  In order to define the rewrite system, we first analyze the typing rules
  containing only one premise, and check whether these rules allow for a
  previous and posterior rule $S_I$. 
  If both are allowed, we choose a direction for the rewrite rule. Then we continue with rules with more than one
  premise and check under which conditions a commutation of rules is possible,
  choosing also a direction.

  \noindent {Rules with one premise:}
  \begin{itemize}
  \item Rule $\alpha_I$:
    \begin{align}
      \label{rule:alpha-SI}
      \vcenter{\infer[^{\alpha_I}]{\Gamma\vdash\alpha.t:S(S(A))}{
        \infer[^{S_I}]{\Gamma\vdash t:S(S(A))}{\Gamma\vdash t:S(A)}
        }}
      &\lra
        \vcenter{\infer[^{S_I}]{\Gamma\vdash\alpha.t:S(S(A))}{
        \infer[^{\alpha_I}]{\Gamma\vdash\alpha.t:S(A)}{\Gamma\vdash t:S(A)}
        }}
    \end{align}
  \item Rules $\Rightarrow_I$, $\times_{E_r}$, $\times_{E_l}$,
    $\Uparrow_r$, and $\Uparrow_\ell$: These rules end with a specific
    types not admitting two $S$ in the head position (i.e.~$\B^j\times S(\B^{n-j})$,
    $\Psi\Rightarrow A$, $\B$, $\B^{n-1}$, and $S(\Psi\times\Phi)$) hence removing an $S$ or adding an $S$ would not allow the rule to be applied, and hence, these rules followed or
    preceded by $S_I$ 
    cannot commute.
  \end{itemize}
  {Rules with more than one premise:}
  \begin{itemize}
  \item Rule $+_I$:
    \begin{align}\label{rule:sum-SI}\nonumber
      \vcenter{
      \infer[^{+_I}]{\Gamma,\Delta,\Xi^\B\vdash(t+u):S(S(A))}{
      \infer[^{S_I}]{\Gamma,\Xi^\B\vdash t:S(S(A))}{\Gamma,\Xi^\B\vdash t:S(A)}
      &
        \infer[^{S_I}]{\Delta,\Xi^\B\vdash u:S(S(A))}{\Delta,\Xi^\B\vdash r:S(A)}
        }
        }\\
      \lra\ &
              \vcenter{
              \infer[^{S_I}]{\Gamma,\Delta,\Xi^\B\vdash(t+u):S(S(A))}{
              \infer[^{+_I}]{\Gamma,\Delta,\Xi^\B\vdash (t+u):S(A)}{
              \Gamma,\Xi^\B\vdash t:S(A) & \Delta,\Xi^\B\vdash u:S(A)
                                           }
                                           }
                                           }
    \end{align}
  \item Rules $\Rightarrow_E$ and $\Rightarrow_{ES}$:
    \begin{align}\label{rule:arrow-SI}\nonumber
      \vcenter{\infer[^{\Rightarrow_{ES}}]{{\Delta,\Gamma,\Xi^\B\vdash tu:S(A)}}{
      \infer[^{S_I}]{\Delta,\Xi^\B\vdash u:S\Psi}{\Delta,\Xi^\B\vdash u:\Psi}
      &
        \infer[^{S_I}]{\Gamma,\Xi^\B\vdash t:S(\Psi\Rightarrow A)}{\Gamma,\Xi^\B\vdash t:\Psi\Rightarrow A}
        }
        }\\
      \lra\ &
              \vcenter{
              \infer[^{S_I}]{\Delta,\Gamma,\Xi^\B\vdash tu:S(A)}{
              \infer[^{\Rightarrow_E}]{\Delta,\Gamma,\Xi^\B\vdash tu:A}{
              \Delta,\Xi^\B\vdash u:\Psi
            &
              \Gamma,\Xi^\B\vdash t:\Psi\Rightarrow A
              }
              }
              }
    \end{align}
  \item Rules $\mathit{If}$ and $\times_I$: These rules end with a specific 
    types not admitting two $S$ in the head position (i.e.~$\B\Rightarrow A$ and
    $\Psi\times\Phi$),
    hence removing an $S$ or adding an $S$ would not allow the rule to be
    applied, and hence, these rules followed or preceded by $S_I$ 
    cannot commute.
  \end{itemize}
  The confluence of this rewrite system is easily inferred from the fact that
  there are not critical pairs. The strong normalization follows from the fact that the
  trees are finite and all the rewrite rules push the $S_I$ 
  to the
  root of the trees. 

  It only remains to check that each rule preserves the semantics.
  \begin{itemize}
  \item Rule \eqref{rule:alpha-SI}: The following diagram gives the semantics of
    both trees (we only treat, without lost of generality, the case where $A\neq S(A')$). 
    \begin{center}
      \begin{tikzcd}[column sep=1.3cm]
        \Gamma\ar[d,"t"]\\
        USA\ar[d,"U\lambda"]\ar[r,"\eta"]\ar[drr,dashed,red,bend left=5,"f",pos=0.6]& USUSA\ar[r,"USU\lambda"]\ar[drr,dashed,red,"USf",bend left=5,pos=0.6] &USU(SA\otimes\I)\ar[r,"USU(\Id\otimes\alpha)"] & USU(SA\otimes\I)\ar[d,"USU\lambda^{-1}"]\\
        U(SA\otimes\I)\ar[r,"U(\Id\otimes\alpha)"]& U(SA\otimes\I)\ar[r,"U\lambda^{-1}"]& USA\ar[r,"\eta"]&USUSA
      \end{tikzcd}
    \end{center}
    Let $h = \lambda^{-1}\circ(\Id\otimes\alpha)\circ\lambda$ and $f=Uh$. The
    diagram commutes by naturality of $\eta$  with respect to $f$.

  \item Rule \eqref{rule:sum-SI}: The following diagram gives the semantics of
    both trees (we only treat, without lost of generality, the case where $A\neq S(A')$).
    \begin{center}
      \begin{tikzcd}[execute at end picture={
          \path (\tikzcdmatrixname-1-3) -- (\tikzcdmatrixname-2-3) coordinate[pos=0.5] (aux)
          (aux) -- (\tikzcdmatrixname-2-4) node[midway,blue]{\small \eqref{a.2:def}};
          \path (\tikzcdmatrixname-1-2) -- (\tikzcdmatrixname-2-2) coordinate[pos=0.5] (aux1)
          (aux1) -- (aux) node[midway,blue]{\small \eqref{a.2:nat}};
          \path (\tikzcdmatrixname-2-2) -- (\tikzcdmatrixname-3-2) coordinate[pos=0.5] (aux1)
          (\tikzcdmatrixname-2-3) -- (\tikzcdmatrixname-3-3) coordinate[pos=0.5] (aux2)
          (aux1) -- (aux2) node[midway,blue]{\small\eqref{a.2:naturality}};
        }]
        \Gamma\times\Xi^\B\times\Delta\times\Xi^\B\ar[r,"t\times r"] & (USA)^2\ar[r,"\eta^2"]\ar[d,"g_0=\Id"] & (USUSA)^2\ar[d,"g_1"]\ar[rd,dashed,red,"n"] \\
        \Gamma\times\Delta\times\Xi^\B\times\Xi^\B \ar[u,"\Id\times\sigma\times\Id"]
        &(USA)^2\ar[d,"\hat +"]\ar[r,dashed,red,"\eta"]&US(USA)^2\ar[d,"US\hat +"] &
        \color{red}U(SUSA\otimes SUSA)\ar[l,dashed,red,"Um"]\\
        \Gamma\times\Delta\times\Xi^\B\ar[u,"\Id\times\Delta"] &USA\ar[r,"\eta"]&USUSA
      \end{tikzcd}
    \end{center}
    \begin{enumerate}
    \item\label{a.2:def} Definition of $g_1$.
    \item\label{a.2:nat} $\eta$ is a monoidal natural transformation.
    \item\label{a.2:naturality} Naturality of $\eta$ with respect to $\hat +$.
    \end{enumerate}
  \item Rule \eqref{rule:arrow-SI}: The following diagram gives the semantics
    of both trees.
    \begin{center}
      \begin{tikzcd}[column sep=1.8cm,
        execute at end picture={
          \path
          (\tikzcdmatrixname-2-1) -- (\tikzcdmatrixname-3-1) coordinate[pos=0.5](aux1)
          (\tikzcdmatrixname-2-3) -- (\tikzcdmatrixname-3-3) coordinate[pos=0.5](aux2)
          (aux1) -- (aux2) node[near start,sloped,blue]{\small\eqref{a.3:nat}}
          (aux1) -- (aux2)  node[near end,sloped,blue]{\small\eqref{a.3:mon}};
        }]
        \Delta\times\Xi^\B\times\Gamma\times\Xi^\B\ar[d,"u\times t"] &
        \Delta\times\Gamma\times\Xi^\B\ar[l,swap,"(\Id\times\sigma\times\Id)\circ(\Id\times\Delta)"]\\
        \Psi\times\home\Psi A\ar[d,"\varepsilon^\Psi"]\ar[rrd,dashed,red,"\eta"]\ar[r,"\eta^2"]&US\Psi\times US(\home\Psi A)\ar[r,"n"] &U(S\Psi\otimes S(\home\Psi A))\ar[d,"Um"]\\
        A\ar[r,"\eta"]& US(A)&US(\Psi\times\home\Psi A)\ar[l,"US(\varepsilon^\Psi)"]\\
      \end{tikzcd}
    \end{center}
    \begin{enumerate}
    \item\label{a.3:nat} Naturality of $\eta$ with respect to $\varepsilon^\Psi$.
    \item\label{a.3:mon} $\eta$ is a monoidal natural transformation.
    \qed
    \end{enumerate}
  \end{itemize}
\end{proof}
\begin{lemma}[Weakening]\label{lem:W}
  If $\Gamma\vdash t:A$, then $\Gamma,\Delta^\B\vdash t:A$. Moreover,
  $\sem{\Gamma,\Delta^\B\vdash t:A}=\sem{\Gamma\vdash t:A}\circ(\Id\times{!})$.
\end{lemma}
\begin{proof}
  It is easy to show that a tree deriving $\Gamma\vdash t:A$ can be transformed
  into a tree deriving $\Gamma,\Delta^\B\vdash t:A$ just by adding $\Delta^\B$
  to the contexts in its axioms.
  Moreover, since $FV(t)\cap\Delta^\B=\emptyset$, we have
  $\sem{\Gamma,\Delta^\B\vdash t:A}=\sem{\Gamma\vdash t:A}\circ(\Id\times{!})$.
  \qed
\end{proof}
\xrecap{Lemma}{Substitution}{lem:substitution}
{
  If $\Gamma',x:\Psi,\Gamma\vdash t:A$ and $\vdash r:\Psi$, then the following diagram
  commutes:
}
\begin{center}
  \begin{tikzcd}
    \Gamma'\times\Gamma\ar[r,"(r/x)t"]\ar[d,"\lambda_\times\times\Id",swap] & A\\
    \Gamma'\times{1}\times\Gamma\ar[r,"\Id\times r\times\Id",swap] &\Gamma'\times\Psi\times\Gamma\ar[u,"t"]
  \end{tikzcd}
\end{center}
\emph{  That is,
  $\sem{\Gamma'\times\Gamma\vdash(r/x)t:A}=\sem{\Gamma,x:\Psi,\Gamma'\vdash t:A}\circ(\Id\times\sem{\vdash r:\Psi}\times\Id)$.
}
\begin{proof}
  By induction on the derivation of $\Gamma',x:\Psi,\Gamma\vdash t:A$.
  Also, we take the rules $\alpha_I$ and $+_I$ with $m=1$, the generalization is straightforward.
  \begin{itemize}
  \item $\vcenter{\infer{\Delta^\B,x:\Psi\vdash x:\Psi}{}}$
    \begin{center}
      \begin{tikzcd}[
        execute at end picture={
          \path (\tikzcdmatrixname-1-1) -- (\tikzcdmatrixname-3-1) coordinate[pos=0.5](aux)
          (aux) -- (\tikzcdmatrixname-2-2) node[midway,blue]{\small\eqref{substax:A}};
          \path (\tikzcdmatrixname-1-1) -- (\tikzcdmatrixname-1-3) coordinate[pos=0.5](aux)
          (aux) -- (\tikzcdmatrixname-2-2) node[midway,blue]{\small\eqref{substax:B}};
          \path (\tikzcdmatrixname-1-3) -- (\tikzcdmatrixname-3-3) coordinate[pos=0.5](aux1)
          (\tikzcdmatrixname-2-2) -- (\tikzcdmatrixname-3-1) coordinate[pos=0.5](aux2)
          (aux1) -- (aux2) node[midway,blue]{\small\eqref{substax:C}};
          \path (\tikzcdmatrixname-1-3) -- (\tikzcdmatrixname-3-3) node[midway,xshift=4cm]{\parbox{0.4\textwidth}{
              \begin{enumerate}
              \item\label{substax:A} Naturality of the projection.
              \item\label{substax:B} Lemma~\ref{lem:W}.
              \item\label{substax:C} Functoriality of the product.
              \end{enumerate}
            }};
        }]
        \Delta^\B \arrow[dd, "\lambda_\times",swap] \arrow[rr, "(r/x)x"] \arrow[rd, "!", red,dashed] & & \Psi \\
        &\color{red} 1 \arrow[ru, "r", red,dashed] & \\
        \Delta^\B\times 1 \arrow[rr, "\Id\times r"'] \arrow[ru, "!", red,dashed] & & \Delta^\B\times\Psi \arrow[uu, "x{\color{red} = !\times\Id}",swap]
      \end{tikzcd}
    \end{center}

  \item
    $\vcenter{\infer{\Gamma,x:\Psi,\Gamma\vdash\alpha.t:S(A)}{\Gamma',x:\Psi,\Gamma\vdash
        t:S(A)}}$

    \begin{center}
      \begin{tikzcd}[column sep=1cm,
        execute at end picture={
          \path (\tikzcdmatrixname-1-2) -- (\tikzcdmatrixname-4-2) coordinate[pos=0.5](aux)
          (\tikzcdmatrixname-3-1) -- (aux) node[midway,blue,xshift=-5mm]{\small\eqref{substSI:A}}
          (aux) -- (\tikzcdmatrixname-3-3) node[midway,blue,xshift=5mm]{\small\eqref{substSI:B}}
          (aux) -- (aux) node[midway,blue,yshift=-1.8cm]{\small\eqref{substSI:C}}
          (\tikzcdmatrixname-1-2) -- (\tikzcdmatrixname-5-2) node[midway,xshift=6.1cm]{\parbox{0.4\textwidth}{
              \begin{enumerate}
              \item\label{substSI:A} Definition of the map $\alpha.(r/x)t$.
              \item\label{substSI:B} Definition of the map $\alpha.t$.
              \item\label{substSI:C} Induction hypothesis.
              \end{enumerate}
            }};
        }]
        & USA & \\
        &\color{red} U(SA\otimes\I) \arrow[u,pos=0.1, "U\lambda^{-1}", red,dashed] & \\
        \Gamma'\times\Gamma \arrow[rdd,sloped, "\lambda_\times\times\Id"'] \arrow[ruu, sloped,"(r/x)(\alpha.t){\color{red}=\alpha.(r/x)t}"] \arrow[rd,sloped, "(r/x)t"', red,dashed] &\color{red} U(SA\otimes\I) \arrow[u, "U(\Id\otimes\alpha)", red,dashed] & \Gamma'\times\Psi\times\Gamma \arrow[luu,sloped,swap, "\alpha.t"'] \arrow[ld,sloped,swap, "t", red,dashed] \\
        &\color{red} USA \arrow[u, "U\lambda", red,dashed]  & \\
        & \Gamma'\times 1\times\Gamma \arrow[ruu,sloped, "\Id\times r\times\Id"'] & 
      \end{tikzcd}
    \end{center}

  \item $\vcenter{\infer{\Gamma',x:\Psi,\Gamma,\Delta,\Xi^\B\vdash
        t+u:S(A)}{\Gamma',x:\Psi,\Gamma,\Xi^\B\vdash t:S(A) &
        \Delta,\Xi^\B\vdash u:S(A)}}$\qquad
    \parbox{.43\textwidth}{We only treat the case when $x\in FV(t)$, the cases $x\in FV(u)$ and $x\in FV(u)\cap FV(t)$ are analogous.}
    
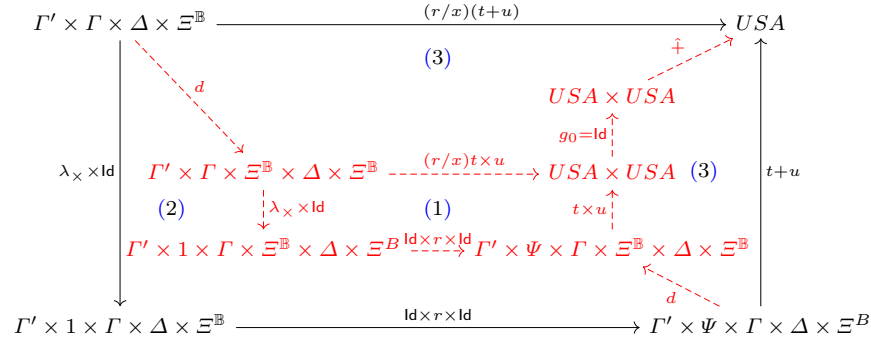
\begin{center}
      \begin{tikzcd}[
        execute at end picture={
          \path (\tikzcdmatrixname-3-2) -- (\tikzcdmatrixname-4-2) coordinate[pos=0.5](aux1)
          (\tikzcdmatrixname-3-3) -- (\tikzcdmatrixname-4-3) coordinate[pos=0.5](aux2)
          (aux1) -- (aux2) node[midway,sloped,blue]{\small\eqref{subssum}}
          (aux1) -- (aux2) node[midway,sloped,blue,xshift=-3.5cm]{\small\eqref{subsnat}}
          (aux1) -- (aux2) node[midway,sloped,blue,xshift=3.5cm,yshift=5mm]{\small\eqref{subsdef}}
          (aux1) -- (aux2) node[midway,sloped,blue,yshift=2cm]{\small\eqref{subsdef}};
        }]
        \Gamma'\times\Gamma\times\Delta\times\Xi^\B \arrow[dddd,swap, "\lambda_\times\times\Id"] \arrow[rdd, "d", red,dashed] \arrow[rrr, "(r/x)(t+u)"] &[-65pt] & &[-65pt] USA \\
        & &\color{red} USA\times USA \arrow[ru, "\hat +", red,dashed] & \\
        & \color{red}\Gamma'\times\Gamma\times\Xi^\B\times\Delta\times\Xi^\B \arrow[r, "(r/x)t\times u", red,dashed] \arrow[d, "\lambda_\times\times\Id", red,dashed] &\color{red} USA\times USA \arrow[u, "g_0=\Id", red,dashed] &\\
        &\color{red} \Gamma'\times
        1\times\Gamma\times\Xi^\B\times\Delta\times\Xi^B \arrow[r, "\Id\times
        r\times \Id", red,dashed]   &\color{red}
        \Gamma'\times\Psi\times\Gamma\times\Xi^\B\times\Delta\times\Xi^\B
        \arrow[u, "t\times u", red,dashed] & \\
        \Gamma'\times 1\times\Gamma\times\Delta\times\Xi^\B \arrow[rrr, "\Id\times r\times \Id"] & & & \Gamma'\times\Psi\times\Gamma\times\Delta\times\Xi^B \arrow[lu, "d", red,dashed] \arrow[uuuu,swap,"t+u"]
      \end{tikzcd}
    \end{center}
    \begin{enumerate}
    \item\label{subssum} Induction hypothesis.
    \item\label{subsnat} Naturality of $d$.
    \item\label{subsdef} Definition of $+$.
    \end{enumerate}

  \item $\vcenter{\infer{\Gamma',x:\Psi,\Gamma\vdash\ite{}ts:\B\Rightarrow A}{\Gamma',x:\Psi,\Gamma\vdash t:A &
        \Gamma',x:\Psi,\Gamma\vdash s:A}}$

    \begin{center}
      \begin{tikzcd}[
        execute at end picture={
          \path (\tikzcdmatrixname-1-2)  -- (\tikzcdmatrixname-2-2) node[midway,xshift=4.5cm]{\parbox{0.55\textwidth}{
              where $(r/x)G = \mathsf{curry}(\mathsf{uncurry}(f_{(r/x)t,(r/x)s})\circ\mathsf{swap})$ and $G =\mathsf{curry}(\mathsf{uncurry}(f_{t,s})\circ\mathsf{swap})$.\\
              By the induction hypothesis, $(r\times\Id)\circ t=(r/x)t$ and $(r\times\Id)\circ s=(r/x)s$, hence, $(r\times\Id)\circ f_{t,s}=f_{(r/x)t,(r/x)s}$ and so $(r/x)G=(r\times\Id)\circ G$, which makes the diagram commute.
            }};
        }]
        \Gamma'\times\Gamma\ar[d,"\lambda_\times\times\Id",swap]\ar[r,"(r/x)G"]   & \lbrack\B,A\rbrack\\
        \Gamma'\times{1}\times\Gamma\ar[r,"\Id\times r\times\Id",swap] & \Gamma'\times\Psi\times\Gamma\ar[u,"G",swap]
      \end{tikzcd}
    \end{center}

  \item $\vcenter{\infer{\Gamma',x:\Psi,\Gamma\vdash \lambda y{:}\Phi.t:\Phi\Rightarrow A}{\Gamma',x:\Psi,\Gamma,y:\Phi\vdash t:A}}$

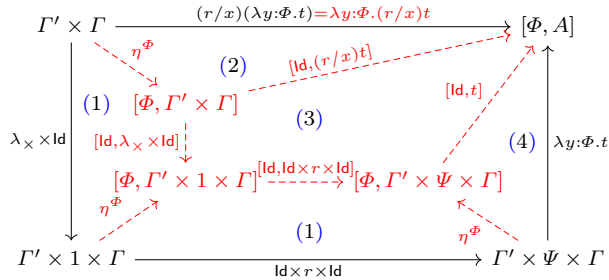
    \begin{center}
      \begin{tikzcd}[column sep=1cm,
        execute at end picture={
          \path
          (\tikzcdmatrixname-1-1) -- (\tikzcdmatrixname-4-1) coordinate[pos=0.5](aux1)
          (\tikzcdmatrixname-1-4) -- (\tikzcdmatrixname-4-4) coordinate[pos=0.5](aux2)
          (aux1) -- (aux2) node[midway,blue,xshift=-1cm,yshift=1cm]{\small\eqref{substAI:A}}
          (aux1) -- (aux2) node[midway,blue,xshift=2.8cm]{\small\eqref{substAI:B}}
          (aux1) -- (aux2) node[midway,blue,yshift=.3cm]{\small\eqref{substAI:C}}
          (aux1) -- (aux2) node[midway,blue,yshift=-1.2cm]{\small\eqref{substAI:D}}
          (aux1) -- (aux2) node[midway,blue,xshift=-2.8cm,yshift=5mm]{\small\eqref{substAI:D}}
          (\tikzcdmatrixname-1-4) -- (\tikzcdmatrixname-4-4) node[midway,xshift=2.7cm]{\parbox{0.28\textwidth}{
              \begin{enumerate}
              \item\label{substAI:D} Naturality of $\eta^\Phi$.
              \item\label{substAI:A} Definition of the map $\lambda y{:}\Phi.(r/x)t$.
              \item\label{substAI:C} Induction hypothesis and the functoriality of $[\Phi,\--]$.
              \item\label{substAI:B} Definition of the map $\lambda y{:}\Phi.t$.
              \end{enumerate}
            }}
          ;
        }]
        \Gamma'\times\Gamma \arrow[rrr, "(r/x)(\lambda y{:}\Phi.t){\color{red}=\lambda y{:}\Phi.(r/x)t}"] \arrow[ddd, "\lambda_\times\times\Id",swap] \arrow[rd, "\eta^\Phi", red,dashed] &[-40pt] & &[-40pt] {[\Phi,A]} \\
        &\color{red} {[\Phi,\Gamma'\times\Gamma]} \arrow[rru, "{[\Id,(r/x)t]}", red,dashed,sloped] \arrow[d, "{[\Id,\lambda_\times\times\Id]}", red,dashed,swap] & & \\
        &\color{red} {[\Phi,\Gamma'\times 1\times\Gamma]} \arrow[r, "{[\Id,\Id\times r\times\Id]}", red,dashed] &\color{red} {[\Phi,\Gamma'\times\Psi\times\Gamma]} \arrow[ruu, "{[\Id,t]}", red,dashed] & \\
        \Gamma'\times 1\times\Gamma \arrow[rrr,swap, "\Id\times r\times\Id"] \arrow[ru, "\eta^\Phi", red,dashed] & & & \Gamma'\times\Psi\times\Gamma \arrow[uuu,swap, "\lambda y{:}\Phi.t"] \arrow[lu, "\eta^\Phi", red,dashed]
      \end{tikzcd}
    \end{center}

  \item $\vcenter{\infer{\Delta,\Gamma',x:\Psi,\Gamma,\Xi^\B\vdash tu:A}
      {
        \Delta,\Xi^\B\vdash u:\Phi
        &
        \Gamma',x:\Psi,\Gamma,\Xi^\B\vdash t:\Phi\Rightarrow A
      }}$
    \begin{center}
      \begin{tikzcd}[
        execute at end picture={
          \path
          (\tikzcdmatrixname-1-1) -- (\tikzcdmatrixname-2-1) coordinate[pos=0.5](aux1)
          (\tikzcdmatrixname-1-2) -- (\tikzcdmatrixname-2-2) coordinate[pos=0.5](aux2)
          (aux1) -- (aux2) node[midway,blue]{\small\eqref{substAE:B}}
          (\tikzcdmatrixname-2-1) -- (\tikzcdmatrixname-3-1) coordinate[pos=0.5](aux1)
          (\tikzcdmatrixname-2-2) -- (\tikzcdmatrixname-3-2) coordinate[pos=0.5](aux2)
          (aux1) -- (aux2) node[midway,blue]{\small\eqref{substAE:C}}
          (aux1) -- (aux2) node[midway,blue,xshift=-4cm]{\small\eqref{substAE:A}}
          (aux1) -- (aux2) node[midway,blue,xshift=3.5cm]{\small\eqref{substAE:D}}
          (\tikzcdmatrixname-3-1) -- (\tikzcdmatrixname-4-1) coordinate[pos=0.5](aux1)
          (\tikzcdmatrixname-3-2) -- (\tikzcdmatrixname-4-2) coordinate[pos=0.5](aux2)
          (aux1) -- (aux2) node[midway,blue]{\small\eqref{substAE:A}}
          ;
        }]
        \Delta\times\Gamma'\times\Gamma\times\Xi^\B \arrow[r, "(r/x)(tu){\color{red}=(r/x)tu}"] \arrow[ddd, "\lambda_\times\times\Id"', out=180,in=180] \arrow[d, "d", red,dashed] & A \\
        \color{red}\Delta\times\Xi^\B\times\Gamma'\times\Gamma\times\Xi^\B \arrow[d, "\Id\times\lambda_\times\times\Id", red,dashed] \arrow[r, "u\times (r/x)t", red,dashed] &\color{red} {\Phi\times[\Phi,A]} \arrow[u, "\varepsilon^\Phi", red,dashed] \\
        \color{red}\Delta\times\Xi^\B\times\Gamma'\times 1\times\Gamma\times\Xi^\B \arrow[r,yshift=1.5mm, "\Id\times\Id\times r\times\Id", red,dashed] &\color{red} \Delta\times\Xi^\B\times\Gamma'\times\Psi\times\Gamma\times\Xi^\B \arrow[u, "u\times t", red,dashed] \\
        \Delta\times\Gamma'\times 1\times\Gamma\times\Xi^\B \arrow[r, "\Id\times r\times\Id"'] \arrow[u,swap, "d", red,dashed] & \Delta\times\Gamma'\times\Psi\times\Gamma\times\Xi^\B \arrow[uuu, "tu"', out=0,in=0] \arrow[u, "d", red,dashed]
      \end{tikzcd}
    \end{center}
    \begin{enumerate}
    \item\label{substAE:A} Naturality of $d$.
    \item\label{substAE:B} Definition of the map $(r/x)tu$.
    \item\label{substAE:C} Induction hypothesis and functoriality of the product.
    \item\label{substAE:D} Definition of the map $tu$.
    \end{enumerate}

  \item $\vcenter{\infer{\Delta',x:\Psi,\Delta,\Gamma,\Xi^\B\vdash tu:A}
      {
        \Delta',x:\Psi,\Delta,\Xi^\B\vdash u:\Phi
        &
        \Gamma,\Xi^\B\vdash t:\Phi\Rightarrow A
      }}$
    \qquad
    Analogous to previous case. 

  \item $\vcenter{\infer{\Delta,\Gamma',x:\Psi,\Gamma,\Xi^\B\vdash tu:S(A)}
      {
        \Delta,\Xi^\B\vdash u:S\Phi
        &
        \Gamma',x:\Psi,\Gamma,\Xi^\B\vdash t:S(\Phi\Rightarrow A)
      }}$
    \begin{center}
      \begin{tikzcd}
        [column sep=5mm,
        execute at end picture={
          \path
          (\tikzcdmatrixname-4-2) -- (\tikzcdmatrixname-5-2) coordinate[pos=0.5](aux1)
          (\tikzcdmatrixname-4-3) -- (\tikzcdmatrixname-5-3) coordinate[pos=0.5](aux2)
          (aux1) -- (aux2) node[midway,blue]{\small\eqref{substAES:C}}
          (aux1) -- (aux2) node[midway,blue,yshift=2.5cm]{\small\eqref{substAES:B}}
          (aux1) -- (aux2) node[midway,blue,yshift=1.2cm,xshift=3.8cm]{\small\eqref{substAES:D}}
          (aux1) -- (aux2) node[midway,blue,xshift=-3.8cm]{\small\eqref{substAES:A}}
          (aux1) -- (aux2) node[midway,blue,yshift=-1.3cm]{\small\eqref{substAES:A}}
          ;
        }]
        \Delta\times\Gamma'\times\Gamma\times\Xi^\B \arrow[rrr, "(r/x)(tu){\color{red}=(r/x)tu}"] \arrow[ddddd, "\lambda_\times\times\Id",swap] \arrow[rddd, "d", red,dashed] &[-55pt] & &[-55pt]] USA \\
        & &\color{red} {US(\Phi\times[\Phi,A])} \arrow[ru, "US\varepsilon^\Phi", red,dashed] & \\
        & &\color{red} {U(S\Phi\otimes S[\Phi,A])} \arrow[u, "Um", red,dashed] & \\
        &\color{red} \Delta\times\Xi^\B\times\Gamma'\times\Gamma\times\Xi^\B \arrow[d, "\Id\times\lambda_\times\times\Id", red,dashed] \arrow[r, "u\times (r/x)t", red,dashed] &\color{red} {US\Phi\times US[\Phi,A]} \arrow[u, "n", red,dashed] & \\
        &\color{red} \Delta\times\Xi^\B\times\Gamma'\times 1\times\Gamma\times\Xi^\B \arrow[r,"\raisebox{1.5mm}{$\Id\times\Id\times r\times\Id$}", red,dashed] &\color{red} \Delta\times\Xi^\B\times\Gamma'\times\Psi\times\Gamma\times\Xi^\B \arrow[u, "u\times t", red,dashed] & \\
        \Delta\times\Gamma'\times 1\times\Gamma\times\Xi^\B \arrow[rrr,swap, "\Id\times r\times\Id"] \arrow[ru, "d", red,dashed] & & & \Delta\times\Gamma'\times\Psi\times\Gamma\times\Xi^\B \arrow[uuuuu,swap, "tu"] \arrow[lu, "d", red,dashed]
      \end{tikzcd}
    \end{center}
    \begin{enumerate}
    \item\label{substAES:A} Naturality of $d$.
    \item\label{substAES:B} Definition of the map $(r/x)tu$.
    \item\label{substAES:C} Induction hypothesis and functoriality of the product.
    \item\label{substAES:D} Definition of the map $tu$.
    \end{enumerate}
  \item $\vcenter{\infer{\Delta'x:\Psi,\Delta,\Gamma,\Xi^\B\vdash tu:S(A)}
      {
        \Delta',x:\Psi,\Delta,\Xi^\B\vdash u:S\Phi
        &
        \Gamma,\Xi^\B\vdash t:S(\Phi\Rightarrow A)
      }}$
    \qquad
    Analogous to previous case.

  \item $\vcenter{\infer{\Gamma',x:\Psi,\Gamma,\Delta,\Xi^\B\vdash t\times u:\Phi\times
        \Upsilon}{\Gamma',x:\Psi,\Gamma,\Xi^\B\vdash t:\Phi & \Delta,\Xi^\B\vdash u:\Upsilon}}$
    \begin{center}
      \begin{tikzcd}[column sep=1cm,
        execute at end picture={
          \node at (5,.3) {\parbox{.35\textwidth}{This diagram commutes by the
              induction hypothesis and the functoriality of the product}};
        }]
        \Gamma'\times\Gamma\times\Delta\times\Xi^\B \arrow[r, "\begin{array}{c}(r/x)(t\times u)\\{\color{red}=\,(r/x)t\times u}\end{array}"] \arrow[d, "\lambda_\times\times\Id{\color{red}=\lambda_\times\times\Id\times\Id}",swap] & \Psi\times\Upsilon \\
        \Gamma'\times 1\times\Delta\times\Gamma\times\Xi^\B \arrow[r,swap,
        "\begin{array}{c}\Id\times r\times\Id\\ {\color{red}=\,\Id\times r\times\Id\times\Id}\end{array}"] & \Gamma'\times\Psi\times\Delta\times\Gamma\times\Xi^\B \arrow[u, "t\times u"]
      \end{tikzcd}
    \end{center}

  \item $\vcenter{\infer{\Gamma,\Delta',x:\Psi,\Delta,\Xi^\B\vdash t\times u:\Phi\times
        \Upsilon}{\Gamma,\Xi^\B\vdash t:\Phi & \Delta',x:\Psi,\Delta,\Xi^\B\vdash u:\Upsilon}}$
    \qquad
    Analogous to previous case.

  \item $\vcenter{\infer{\Gamma',x:\Psi,\Gamma\vdash\head\ t:\B}{\Gamma',x:\Psi,\Gamma\vdash t:\B^n}}$
    \begin{center}
      \begin{tikzcd}[
        execute at end picture={
          \path
          (\tikzcdmatrixname-1-1) -- (\tikzcdmatrixname-3-1) coordinate[pos=0.5](aux1)
          (\tikzcdmatrixname-1-3) -- (\tikzcdmatrixname-3-3) coordinate[pos=0.5](aux2)
          (aux1) -- (aux2) node[midway,blue,xshift=-1.5cm]{\small\eqref{substhead:A}}
          (aux1) -- (aux2) node[midway,blue,xshift=1.5cm]{\small\eqref{substhead:C}}
          (aux1) -- (aux2) node[midway,blue,yshift=.5cm]{\small\eqref{substhead:B}}
          (aux1) -- (aux2) node[midway,xshift=6cm]{\parbox{0.4\textwidth}{
              \begin{enumerate}
              \item \label{substhead:A} Induction hypothesis.
              \item \label{substhead:B} Definition of the map $\head\ (r/x)t$.
              \item \label{substhead:C} Definition of the map $\head\ t$.
              \end{enumerate}
            }}
          ;
        }]
        \Gamma'\times\Gamma \arrow[dd,swap, "\lambda_\times\times\Id"] \arrow[rr, "(r/x)(\head\ t){\color{red}=\head\ (r/x)t}"] \arrow[rd, "(r/x)t", red,dashed] & & \B \\
        & \color{red}\B^n \arrow[ru, "\pi_1", red,dashed] & \\
        \Gamma'\times 1\times\Gamma \arrow[rr, "\Id\times r\times\Id",swap] & & \Gamma'\times\Psi\times\Gamma \arrow[uu, "\head\ t",swap] \arrow[lu, "t", red,dashed,swap]
      \end{tikzcd}
    \end{center}

  \item $\vcenter{\infer{\Gamma',x:\Psi,\Gamma\vdash\tail\ t:\B^{n-1}}{\Gamma',x:\Psi,\Gamma\vdash t:\B^n}}$
    \begin{center}
      \begin{tikzcd}[
        execute at end picture={
          \path
          (\tikzcdmatrixname-1-1) -- (\tikzcdmatrixname-3-1) coordinate[pos=0.5](aux1)
          (\tikzcdmatrixname-1-3) -- (\tikzcdmatrixname-3-3) coordinate[pos=0.5](aux2)
          (aux1) -- (aux2) node[midway,blue,xshift=-1.5cm]{\small\eqref{substtail:A}}
          (aux1) -- (aux2) node[midway,blue,xshift=1.5cm]{\small\eqref{substtail:C}}
          (aux1) -- (aux2) node[midway,blue,yshift=.5cm]{\small\eqref{substtail:B}}
          (aux1) -- (aux2) node[midway,xshift=6cm]{\parbox{0.4\textwidth}{
              \begin{enumerate}
              \item \label{substtail:A} Induction hypothesis.
              \item \label{substtail:B} Definition of the map $\tail\ (r/x)t$.
              \item \label{substtail:C} Definition of the map $\tail\ t$.
              \end{enumerate}
            }}
          ;
        }]
        \Gamma'\times\Gamma \arrow[dd,swap, "\lambda_\times\times\Id"] \arrow[rr, "(r/x)(\tail\ t){\color{red}=\tail\ (r/x)t}"] \arrow[rd, "(r/x)t", red,dashed] & & \B^{n-1} \\
        &\color{red} \B^n \arrow[ru, "\pi_2", red,dashed] & \\
        \Gamma'\times 1\times\Gamma \arrow[rr, "\Id\times r\times\Id",swap] & & \Gamma'\times\Psi\times\Gamma \arrow[uu, "\tail\ t",swap] \arrow[lu, "t", red,dashed,swap]
      \end{tikzcd}
    \end{center}
  \item $\vcenter{\infer{\Gamma',x:\Psi,\Gamma\vdash\Uparrow_r t:S(\Phi\times \Upsilon)}{\Gamma',x:\Psi,\Gamma\vdash t:S(S^k(\Phi)\times \Upsilon)}}$
    \begin{center}
      \begin{tikzcd}[
        execute at end picture={
          \path
          (\tikzcdmatrixname-1-1) -- (\tikzcdmatrixname-3-1) coordinate[pos=0.5](aux1)
          (\tikzcdmatrixname-1-3) -- (\tikzcdmatrixname-3-3) coordinate[pos=0.5](aux2)
          (aux1) -- (aux2) node[midway,blue,xshift=-2cm]{\small\eqref{substuparrowr:A}}
          (aux1) -- (aux2) node[midway,blue,xshift=2cm]{\small\eqref{substuparrowr:C}}
          (aux1) -- (aux2) node[midway,blue,yshift=.5cm]{\small\eqref{substuparrowr:B}}
          (aux1) -- (aux2) node[midway,xshift=5.5cm]{\parbox{0.38\textwidth}{
              \begin{enumerate}
              \item \label{substuparrowr:A} Induction hypothesis.
              \item \label{substuparrowr:B} Definition of the map $\Uparrow_r(r/x)t$.
              \item \label{substuparrowr:C} Definition of the map $\Uparrow_r t$.
              \end{enumerate}
            }}
          ;
        }]
        \Gamma'\times\Gamma \arrow[dd,swap, "\lambda_\times\times\Id"] \arrow[rr, "(r/x)(\Uparrow_r t){\color{red}=\Uparrow_r (r/x)t}"] \arrow[rd, "(r/x)t", red,dashed] &[-10pt] &[-10pt] US(\Phi\times\Upsilon) \\
        & \color{red}US((US)^k\Phi\times\Upsilon) \arrow[ru, "\Uparrow_r", red,dashed] & \\
        \Gamma'\times 1\times\Gamma \arrow[rr, "\Id\times r\times\Id",swap] & & \Gamma'\times\Psi\times\Gamma \arrow[uu, "\Uparrow_r t",swap] \arrow[lu, "t", red,dashed,swap]
      \end{tikzcd}
    \end{center}

  \item $\vcenter{\infer{\Gamma',x:\Psi,\Gamma\vdash\Uparrow_\ell t:S(\Phi\times \Upsilon)}{\Gamma',x:\Psi,\Gamma\vdash t:S(\Phi\times S(\Upsilon))}}$
    \qquad
    Analogous to previous case.

  \item $\vcenter{\infer{\Gamma',x:\Psi,\Gamma\vdash
        t:S(A)}{\Gamma',x:\Psi,\Gamma\vdash t:A}}$

    \begin{center}
      \begin{tikzcd}[
        execute at end picture={
          \path
          (\tikzcdmatrixname-1-1) -- (\tikzcdmatrixname-3-1) coordinate[pos=0.5](aux1)
          (\tikzcdmatrixname-1-3) -- (\tikzcdmatrixname-3-3) coordinate[pos=0.5](aux2)
          (aux1) -- (aux2) node[midway,blue,xshift=-1.5cm]{\small\eqref{substS:A}}
          (aux1) -- (aux2) node[midway,blue,xshift=1.5cm]{\small\eqref{substS:C}}
          (aux1) -- (aux2) node[midway,blue,yshift=.5cm]{\small\eqref{substS:B}}
          (aux1) -- (aux2) node[midway,xshift=6cm]{\parbox{0.4\textwidth}{
              \begin{enumerate}
              \item \label{substS:A} Induction hypothesis.
              \item \label{substS:B} Definition of the map $(r/x)t^{USA}$.
              \item \label{substS:C} Definition of the map $t^{USA}$.
              \end{enumerate}
            }}
          ;
        }]
        \Gamma'\times\Gamma \arrow[dd,swap, "\lambda_\times\times\Id"] \arrow[rr, "(r/x)t^{USA}"] \arrow[rd, "(r/x)t^A", red,dashed] & & USA \\
        &\color{red} A \arrow[ru, "\eta", red,dashed] & \\
        \Gamma'\times 1\times\Gamma \arrow[rr, "\Id\times r\times\Id",swap] & & \Gamma'\times\Psi\times\Gamma \arrow[uu, "t^{USA}",swap] \arrow[lu, "t^A", red,dashed,swap]
      \end{tikzcd}
    \end{center}
    \qed
  \end{itemize}
\end{proof}
\xrecap{Theorem}{Soundness}{thm:soundness}{
  If $\vdash t:A$, and $t\lra r$,
  then
  $\sem{\vdash t:A} = \sem{\vdash r:A}$.
}
\begin{proof}
  By induction on the rewrite relation, using the first derivable type for each
  term.
  We take the rules $\alpha_I$ and $+_I$ with $m=1$, the generalization is straightforward.
  \begin{description}
  \item[$(\mathsf{comm})$] $(t+r)=(r+t)$. We have
    \[
      \vcenter
      {
        \infer{\vdash (t+r):S(A)}
        {
          \vdash t:S(A) & \vdash r:S(A) 
        }
      }
      \qquad\textrm{and}\qquad
      \vcenter
      {
        \infer{\vdash (r+t):S(A)}
        {
          \vdash r:S(A) & \vdash t:S(A)
        }
      }
    \]
    Then,
    
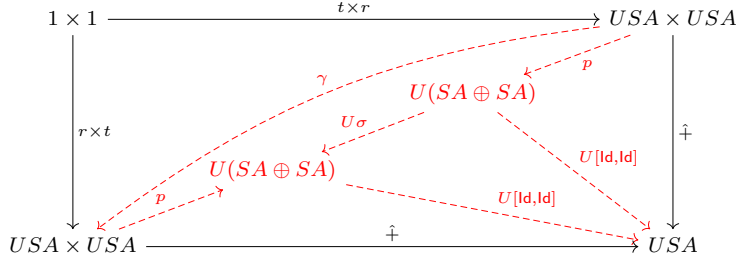
\begin{center}
      \begin{tikzcd}
        1\times 1 \arrow[rrr, "t\times r"] \arrow[ddd, "r\times t"]  & & & USA\times USA \arrow[ddd, "\hat +"] \arrow[ld, "p", red,dashed] \arrow[lllddd, "\gamma"', red,dashed, bend right=15] \\
        & & \color{red}U(SA\oplus SA) \arrow[rdd, "{U[\Id,\Id]}", red,dashed] \arrow[ld, "U\sigma"', red,dashed] & \\
        &\color{red} U(SA\oplus SA) \arrow[rrd, "{U[\Id,\Id]}", red,dashed] & & \\
        USA\times USA \arrow[rrr, "\hat +"] \arrow[ru, "p", red,dashed] & & & USA                                                                                                            
      \end{tikzcd}
    \end{center}
    Where $\gamma$ is the symmetry on $\times$ and $\sigma$ the symmetry on $\oplus$.

  \item[$(\mathsf{asso})$] $((t+r)+s)=(t+(r+s))$. We have
    \[
      \vcenter{
        \infer{\vdash((t+r)+s):S(A)}
        {
          \infer{\vdash(t+r):S(A)}{\vdash t:S(A) &\vdash r:S(A)}
          &
          \vdash s:S(A)
        }
      }
      \qquad\textrm{and}\qquad
      \vcenter{
        \infer{\vdash(t+(r+s)):S(A)}
        {
          \vdash t:S(A)
          &
          \infer{\vdash (r+s):S(A)}{\vdash r:S(A) & \vdash s:S(A)}
        }
      }
    \]
    Then
    
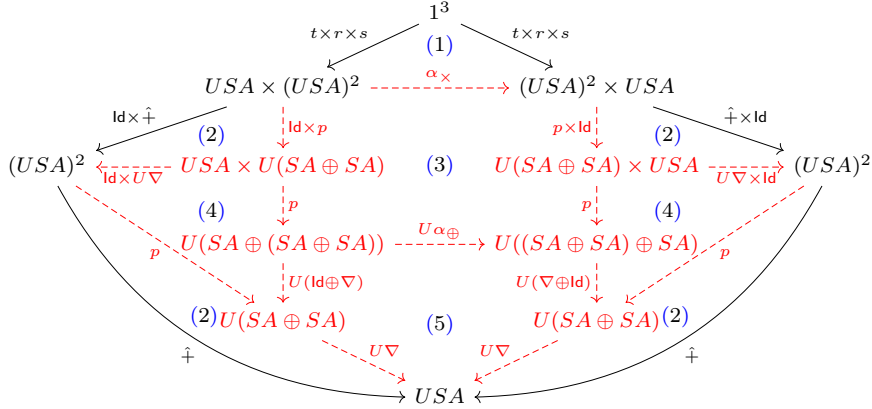
\begin{center}
      \begin{tikzcd}[column sep=5mm,
        execute at end picture={
          \path (\tikzcdmatrixname-2-3) -- (\tikzcdmatrixname-4-3) coordinate[pos=0.5](aux1)
          (\tikzcdmatrixname-2-5) -- (\tikzcdmatrixname-4-5) coordinate[pos=0.5](aux2)
          (aux1) -- (aux2) node[midway,sloped,blue]{\small\eqref{asso:func}};
          \path (aux1) -- (aux2) node[midway,sloped,blue,yshift=1.6cm]{\small\eqref{asso:alpha}};
          \path (aux1) -- (aux2) node[midway,sloped,blue,yshift=0.4cm,xshift=-3cm]{\small\eqref{asso:def}};
          \path (aux1) -- (aux2) node[midway,sloped,blue,yshift=0.4cm,xshift=3cm]{\small\eqref{asso:def}};
          \path (aux1) -- (aux2) node[midway,sloped,blue,yshift=-0.6cm,xshift=3cm]{\small\eqref{asso:nat}};
          \path (aux1) -- (aux2) node[midway,sloped,blue,yshift=-0.6cm,xshift=-3cm]{\small\eqref{asso:nat}};
          \path (aux1) -- (aux2) node[midway,sloped,blue,yshift=-2.1cm]{\small\eqref{asso:trivial}};
          \path (aux1) -- (aux2) node[midway,sloped,blue,yshift=-2cm,xshift=-3.1cm]{\small\eqref{asso:def}};
          \path (aux1) -- (aux2) node[midway,sloped,blue,yshift=-2cm,xshift=3.1cm]{\small\eqref{asso:def}};
        }
        ]
        &  & &[-10pt] 1^3 \arrow[rd, "t\times r\times s"] \arrow[ld,swap, "t\times r\times s"] &[-10pt] &  & \\
        &  & USA\times(USA)^2 \arrow[lld, "\Id\times\hat +"'] \arrow[d, "\Id\times p", red,dashed] \arrow[rr, "\alpha_\times", red,dashed] &                                                                     & (USA)^2\times USA \arrow[rrd, "\hat +\times\Id"] \arrow[d, "p\times\Id"', red,dashed]            &  &                                                                        \\
        (USA)^2 \arrow[rrrddd, "\hat +"', bend right] \arrow[rrdd,swap, "p", red,dashed] &  &\color{red} USA\times U(SA\oplus SA) \arrow[ll,swap, "{\Id\times U\nabla}"', red,dashed] \arrow[d, "p", red,dashed] & &\color{red} U(SA\oplus SA)\times USA \arrow[d, "p"', red,dashed] \arrow[rr, swap,"{U\nabla\times\Id}", red,dashed] &  & (USA)^2 \arrow[lllddd, "\hat +", bend left] \arrow[lldd, "p"',swap, red,dashed] \\
        &  &\color{red} U(SA\oplus(SA\oplus SA)) \arrow[rr, "U\alpha_\oplus", red,dashed] \arrow[d, "{U(\Id\oplus\nabla)}", red,dashed] & & \color{red}U((SA\oplus SA)\oplus SA) \arrow[d, "{U(\nabla\oplus\Id)}"', red,dashed] &  & \\
        &  &\color{red} U(SA\oplus SA) \arrow[rd, "{U\nabla}", red,dashed] & &\color{red} U(SA\oplus SA) \arrow[ld, "{U\nabla}"', red,dashed] &  & \\
        &  & & USA & &  & 
      \end{tikzcd}
    \end{center}
    \begin{enumerate}
    \item\label{asso:alpha} Naturality of $\alpha_\times$.
    \item\label{asso:def} Definition of $\hat +$.
    \item\label{asso:func} $U$ monoidal functor.
    \item\label{asso:nat} Naturality of $p$ with respect to $\nabla$.
    \item\label{asso:trivial} Associativity property of $\oplus$.
    \end{enumerate}
  \item[($\beta_b$)] If $b$ has type $\B^n$ and $b\in\tbasis$, then $(\lambda x{:}{\B^n}.t)b\lra (b/x)t$.
    We have,
    \[
      \vcenter{\infer{\vdash(\lambda x{:}{\B^n}.t)b:A}
        {
          \infer{\vdash\lambda x{:}{\B^n}.t:\B^n\Rightarrow A}
          {x:{\B^n}\vdash t:A}
          &
          \vdash b:\B^n
        }}
      \qquad\textrm{and}\qquad
      \vcenter{\infer{\vdash (b/x)t:A}{}}
    \]
    Then,
    \begin{center}
      \begin{tikzcd}
        1 \arrow[rr, "(b/x)t"] \arrow[rd, "b"] &                      & A \\
        & \B^n \arrow[ru, "t"] &  
      \end{tikzcd}
    \end{center}
    This diagram commutes because of Lemma~\ref{lem:substitution}.

  \item[($\beta_n$)] If $u$ has type $S\Psi$, then $(\lambda x{:}{S\Psi}.t)u\lra  (u/x)t$.
    We have,
    \[
      \vcenter{
        \infer{\vdash(\lambda x{:}{S\Psi}.t)u:A}
        {
          \infer{\vdash\lambda x{:}{S\Psi}.t:S\Psi\Rightarrow A}
          {x:{S\Psi}\vdash t:A}
          &
          \vdash u:S\Psi
        }
      }
      \qquad\textrm{and}\qquad
      \vcenter{
        \infer{\vdash (u/x)t:A}{}
      }
    \]
    Then,
    \begin{center}
      \begin{tikzcd}
        1 \arrow[rr, "(u/x)t"] \arrow[rd, "u"] &                      & A \\
        & US\Psi \arrow[ru, "t"] &  
      \end{tikzcd}
    \end{center}
    This diagram commutes because of Lemma~\ref{lem:substitution}.

  \item[($\mathsf{If}_1$)] $\ite{\ket 1}tr\lra t$. We have,
    \[
      \vcenter{
        \infer{\vdash\ite{\ket 1}tr:A}
        {
          \infer{\vdash\ite{}tr:\B\Rightarrow A}{\vdash t:A & \vdash r:A}
          &
          \infer{\vdash\ket 1:\B}{}
        }
      }
      \qquad\textrm{and}\qquad
      \vcenter{\infer{\vdash t:A}{}}
    \]
    Then,
    \begin{center}
      \begin{tikzcd}[column sep=4cm]
        {1}^2\ar[r,"{i_2\times\mathsf{curry}(\mathsf{uncurry}(f_{t,r})\circ\mathsf{swap})}"]
        &\B\times\home\B A\ar[d,"\varepsilon"]\\
        {1}\ar[u,dashed,"\lambda_\times"]\ar[r,"t"] & A
      \end{tikzcd}
    \end{center}
    Notice that
    $\mathsf{curry}(\mathsf{uncurry}(f_{t,r})\circ\mathsf{swap})$ transforms the arrow
    $\B\xlra{f_{t,r}}\home{1} A$ (which is the arrow $\ket 0\mapsto r$, $\ket
    1\mapsto t$) into an arrow ${1}\xlra{}\home\B A$, and
    hence, $\varepsilon\circ(i_2\times\mathsf{curry}(\mathsf{uncurry}(f_{t,r})\circ\mathsf{swap}))=t$.

  \item[($\mathsf{If}_0$)] Analogous to ($\mathsf{If}_1$).

  \item[($\mathsf{lin}_r^+$)] If $t$ has type $\B^n\Rightarrow A$, then
    $t(u+v)\lra tu+tv$.
    We have,
    \[
      \vcenter{
        \infer{\vdash t(u+v):S(A)}
        {
          \infer{\vdash t:S(\B^n\Rightarrow A)}{\vdash t:\B^n\Rightarrow A}
          &
          \infer{\vdash u+v:S{\B^n}}
          {
            \vdash u:S{\B^n} & \vdash v:S{\B^n}
          }
        }}
    \]
    and
    \[
      \vcenter{\infer{\vdash tu+tv:S(A)}
        {
          \infer{\vdash tu:S(A)}
          {
            \infer{\vdash t:S(\B^n\Rightarrow A)}{\vdash t:\B^n\Rightarrow A} & \vdash u:S{\B^n}
          }
          &
          \infer{\vdash tv:S(A)}
          {
            \infer{\vdash t:S(\B^n\Rightarrow A)}{\vdash t:\B^n\Rightarrow A}  & \vdash v:S{\B^n}
          }
        }}
    \]
    \begin{center}
      \begin{tikzcd}[
        execute at end picture={
        \path
        (\tikzcdmatrixname-1-1) -- (\tikzcdmatrixname-9-1) coordinate[pos=0.5] (aux1)
        (\tikzcdmatrixname-1-3) -- (\tikzcdmatrixname-9-3) coordinate[pos=0.5] (aux2)
        (aux1) -- (aux2) node[midway,sloped,blue]{\small\eqref{linsumr:1}}
        (aux1) -- (aux2) node[midway,sloped,blue,xshift=-2cm]{\small\eqref{linsumr:2}}
        (aux1) -- (aux2) node[midway,sloped,blue,xshift=-2cm,yshift=-3.5cm]{\small\eqref{linsumr:2}}
        (aux1) -- (aux2) node[midway,sloped,blue,xshift=-2cm,yshift=-4.5cm]{\small\eqref{linsumr:3}}
        (aux1) -- (aux2) node[midway,sloped,blue,yshift=-5cm]{\small\eqref{linsumr:3}}
        (aux1) -- (aux2) node[midway,sloped,blue,xshift=4cm]{\small\eqref{linsumr:4}}
        (aux1) -- (aux2) node[midway,sloped,blue,xshift=4.5cm,yshift=1cm]{\small\eqref{linsumr:4}}
        (aux1) -- (aux2) node[midway,sloped,blue,yshift=4.5cm,xshift=-2cm]{\small\eqref{linsumr:8}} 
        (aux1) -- (aux2) node[midway,sloped,blue,yshift=4.5cm]{\small\eqref{linsumr:5}} 
        (aux1) -- (aux2) node[midway,sloped,blue,yshift=2.5cm]{\small\eqref{linsumr:6}} 
        (aux1) -- (aux2) node[midway,sloped,blue,yshift=2.5cm,xshift=2cm]{\small\eqref{linsumr:7}} 
        (aux1) -- (aux2) node[midway,sloped,blue,yshift=2.5cm,xshift=3.3cm]{\small\eqref{linsumr:8}} 
        ;
        }]
        {US{\B^n}\times US[\B^n,A]} \arrow[rr, "n"] &[-25pt] &[-45pt] {{U(S{\B^n}\otimes S([\B^n,A]))}} \arrow[rd,sloped, "Um"] &[-30pt]  \\
        &&&{{US(\B^n\times[\B^n,A])}} \arrow[d, "US\varepsilon^{\B^n}"]\\
        (US\B^n)^2\times US\home{\B^n}A \arrow[rrdddddd,"d",red,dashed,bend left] \arrow[rddd, "\Id\times\Delta", red,dashed] \arrow[uu, "\hat +\times\Id"] \arrow[r, "p\times\Id", red,dashed] & {\color{red}\begin{array}{c}U(S\B^n\oplus S\B^n)\\\times US[\B^n,A]\end{array}} \arrow[luu, "U\nabla\times\Id"',sloped,swap, red,dashed] \arrow[d, "n", red,dashed] & & USA \\
        & {\color{red}\hspace{1cm}\begin{array}{c}U((S\B^n\oplus S\B^n)\\\otimes S[\B^n,A])\end{array}} \arrow[ruuu,sloped,out=45,in=225, "U(\nabla\otimes\Id)", red,dashed,swap] \arrow[rd, "U\delta", red,dashed] & & (USA)^2 \arrow[u,swap, "\hat +"] \\
        & & {\color{red}\begin{array}{l}U((S\B^n\\\ \ \otimes
                          S[\B^n,A])\phantom{)}\\\oplus(S\B^n\\\ \ \otimes S[\B^n,A]))\end{array}} \arrow[uuuu, "U\nabla", red,dashed] & (USA)^2 \arrow[u,swap, "g_0=\Id"] \\
        & \color{red}(US{\B^n})^2\times (US\home{\B^n}A)^2 \arrow[rddd, "\Id\times\sigma\times\Id",sloped, red,dashed] & &\left(\begin{array}{l}US(\B^n\\\times\home{\B^n}A)\end{array}\right)^2 \arrow[u,"(US\varepsilon^{\B^n})^2",swap] \arrow[uuuu,swap, "\hat +", red,dashed, bend left] \\
        {(US{\B^n})^2\times[\B^n,A]} \arrow[uuuu, "g_0\times \eta"] \arrow[rd, "\Id\times\Delta",sloped, red,dashed] & & & \\
        1 \arrow[u, "u\times v\times t"] \arrow[d,"u\times t\times v\times t"] \arrow[r, "u\times v\times t\times t", red,dashed] & {\color{red} {(US{\B^n})^2\times[\B^n,A]^2}} \arrow[ld, "\Id\times\sigma\times\Id",sloped,swap, red,dashed] \arrow[uu, "\Id\times\eta^2", red,dashed] & & \left(\begin{array}{l}U(S{\B^n}\otimes\\ S(\home{\B^n}A))\end{array}\right)^2 \arrow[uu,swap,"(Um)^2"] \arrow[luuuuuuu, "\hat +",swap,out=135,in=315, red,dashed] \arrow[luuu, "p",out=180,in=270,swap, red,dashed] \\
        (US{\B^n}\times\home{\B^n}A)^2 \arrow[rr, "(\Id\times\eta)^2"] & &
        \left(\begin{array}{l} US{\B^n}\times\\ US\home{\B^n}A\end{array}\right)^2 \arrow[ru, "n\times n",sloped] & 
\end{tikzcd}
    \end{center}
    \begin{enumerate}
    \item\label{linsumr:8} Definition of $\hat +$.
    \item\label{linsumr:5} Naturality of $\nabla$.
    \item\label{linsumr:6} \axiomd.
    \item\label{linsumr:7} 
      Let $f=\pi_1\otimes\Id$, $g=\pi_2\otimes\Id$, and $A=B=S\B^n\otimes [\B^n,A]$.
      Hence,
      \begin{align*}
        \nabla\circ\delta &=[\Id,\Id]\circ\langle f,g\rangle\\
                          &=[\Id,\Id]\circ\Id\circ\langle f,g \rangle\\
                          &=[\Id,\Id]\circ(i_A\circ\pi_A+i_B\circ\pi_B)\circ\langle  f,g \rangle\\
                          &=(([\Id,\Id]\circ i_A\circ\pi_A)+([\Id,\Id]\circ i_B\circ\pi_B))\circ\langle f,g \rangle\\
                          &=(\pi_A+\pi_B)\circ\langle f,g \rangle\\
                          &=\pi_A\circ\langle f,g \rangle+\pi_B\circ\langle f,g \rangle\\
                          &=f+g\\
                          &=\pi_1\otimes\Id+\pi_2\otimes\Id\\
                          &=\nabla\otimes\Id
      \end{align*}
    \item\label{linsumr:4} Naturality of $\hat +$.
    \item\label{linsumr:2} Naturality of $\Delta$.
    \item\label{linsumr:1} Definition of $d$.
    \item\label{linsumr:3} Naturality of $\sigma$.
    \end{enumerate}

  \item[$(\mathsf{lin}_r^ \alpha)$] If $t$ has type $\B^n\Rightarrow A$, then
    $t(\alpha.u)\lra \alpha.(tu)$. We have,
    \[
      \vcenter{
        \infer{\vdash t(\alpha.u):S(A)}
        {
          \infer{\vdash t:S(\B^n\Rightarrow A)}{\vdash t:\B^n\Rightarrow A}
          &
          \infer{\vdash\alpha.u:S{\B^n}}{\vdash u:S{\B^n}}
        }
      }
      \qquad\textrm{and}\qquad
      \vcenter{
        \infer{\vdash\alpha.(tu):S(A)}
        {
          \infer{\vdash tu:S(A)}
          {
            \infer{\vdash t:S(\B^n\Rightarrow A)}{\vdash t:\B^n\Rightarrow A}
            &
            \vdash u:S{\B^n}
          }
        }
      }
    \]
    Then,
    \begin{center}
      \begin{tikzcd} [
        execute at end picture={
          \path
          (\tikzcdmatrixname-1-1) -- (\tikzcdmatrixname-9-1) coordinate[pos=0.5](aux1)
          (\tikzcdmatrixname-1-4) -- (\tikzcdmatrixname-9-4) coordinate[pos=0.5](aux2)
          (aux1) -- (aux2) node[midway,blue,yshift=5mm]{\small\eqref{linralpha:B}}
          (aux1) -- (aux2) node[midway,blue,xshift=8mm,yshift=3.5cm]{\small\eqref{linralpha:Azero}}
          (aux1) -- (aux2) node[midway,blue,xshift=-1cm,yshift=2.5cm]{\small\eqref{linralpha:A}}
          (aux1) -- (aux2) node[midway,blue,xshift=-17mm,yshift=1.5cm]{\small\eqref{linralpha:B}}
          (aux1) -- (aux2) node[midway,blue,xshift=25mm,yshift=1.8cm]{\small\eqref{linralpha:A}}
          (aux1) -- (aux2) node[midway,blue,xshift=27mm,yshift=1cm]{\small\eqref{linralpha:A}}
          (aux1) -- (aux2) node[midway,blue,xshift=-35mm,yshift=5mm]{\small\eqref{linralpha:F}}
          (aux1) -- (aux2) node[midway,blue,xshift=-5mm,yshift=-1.5cm]{\small\eqref{linralpha:C}}
          (aux1) -- (aux2) node[midway,blue,xshift=27mm,yshift=-2.1cm]{\small\eqref{linralpha:E}}
          (aux1) -- (aux2) node[midway,blue,xshift=-3.7cm,yshift=-2.1cm]{\small\eqref{linralpha:G}}
          (aux1) -- (aux2) node[midway,blue,xshift=20mm,yshift=-3.7cm]{\small\eqref{linralpha:G}}
          (aux1) -- (aux2) node[midway,blue,xshift=-20mm,yshift=-3.1cm]{\small\eqref{linralpha:C}}
          ;
        }]
        1\arrow[rr, "u\times t"]  &[-35pt] &[-40pt]{US{\B^n}\times{[\B^n,A]}} \arrow[dr, "U\lambda\times\eta"] \arrow[lld,swap, "\Id\times\eta"] &[-20pt] {\ } \\
        {US{\B^n}\times US{[\B^n,A]}} \arrow[rrr, "U\lambda\times\Id", red,dashed] \arrow[d, "n"'] & & & {U(S{\B^n}\otimes\I)\times US{[\B^n,A]}} \arrow[d,"U(\Id\otimes\alpha)\times\Id"] \arrow[ld, "n", red,dashed] \\
        {U(S{\B^n}\otimes S{[\B^n,A]})} \arrow[rr, "U(\lambda\otimes\Id)", red,dashed] \arrow[d, "Um"'] \arrow[rdd, "U\lambda", red,dashed] & &\color{red} {U(S\B^n\otimes\I\otimes S[\B^n,A])} \arrow[ldd, "U(\Id\otimes\sigma)",red,dashed,out=215,in=90,sloped] \arrow[d, "U(\Id\otimes\alpha\otimes\Id)", red,dashed] & {U(S{\B^n}\otimes\I)\times US{[\B^n,A]}} \arrow[d, "U\lambda^{-1}\times\Id"] \arrow[ld, "n", red,dashed] \\
        {US(\B^n\times{[\B^n,A]})} \arrow[dd, "US\varepsilon^{\B^n}"'] \arrow[rddd, "U\lambda",out=290,in=168, red,dashed] & & \color{red}{U(S\B^n\otimes\I\otimes S[\B^n,A])} \arrow[dd, "U(\Id\otimes\sigma)", red,dashed] \arrow[rdd, "U(\lambda^{-1}\otimes\Id)", red,dashed] & {US{\B^n}\times US{[\B^n,A]}} \arrow[dd, "n"] \\
        & \color{red}{U(S\B^n\otimes S[\B^n,A]\otimes\I)} \arrow[dd, "U(m\otimes\Id)"', red,dashed] \arrow[rd, "U(\Id\otimes\Id\otimes\alpha)",sloped,out=290,in=100,red,dashed] & & \\
        USA \arrow[dd, "U\lambda"'] & &\color{red} {U(S\B^n\otimes S[\B^n,A]\otimes\I)} \arrow[dd, "U(m\otimes\Id)", red,dashed] \arrow[r, "U\lambda^{-1}", red,dashed] & {U(S{\B^n}\otimes S[\B^n,A])} \arrow[dd, "Um"] \\
        & \color{red}{U(S(\B^n\times[\B^n,A])\otimes\I)} \arrow[ld,sloped,,swap, "U(S\varepsilon^{\B^n}\otimes\Id)", red,dashed] \arrow[rd, "U(\Id\otimes\alpha)", red,dashed] & & \\
        U(SA\otimes\I) \arrow[d, "U(\Id\otimes\alpha)"'] & &\color{red} {U(S(\B^n\times[\B^n,A])\otimes\I)} \arrow[r, "U\lambda^{-1}", red,dashed] \arrow[lld, "U(S\varepsilon^{\B^n}\otimes\Id)", red,dashed] & {US(\B^n\times{[\B^n,A]})} \arrow[d, "US\varepsilon^{\B^n}"] \\
        U(SA\otimes\I) \arrow[rrr, "U\lambda^{-1}"'] & & & USA                                                                                                   
      \end{tikzcd}
    \end{center}
    \begin{enumerate}
    \item\label{linralpha:Azero} Functoriality of $\times$.
    \item\label{linralpha:A} Naturality of $n$.
    \item\label{linralpha:B} Coherence of $\sigma$.
    \item\label{linralpha:F} Naturality of $\lambda$ with respect to $m$.
    \item\label{linralpha:C} Functoriality of $\otimes$.
    \item\label{linralpha:G} Naturality of $\lambda$ with respect to $S\varepsilon^{\B^n}$.
    \item\label{linralpha:E} Naturality of $\lambda^{-1}$ with respect to $m$.
    \end{enumerate}

  \item[$(\mathsf{lin}^0_r)$] If $t$ has type $\B^n\Rightarrow A$, then
    $t\z[\B^n]\lra \z$. We have,
    \[
      \vcenter{
        \infer{\vdash t\z[\B^n]:S(A)}
        {
          \infer{\vdash t:S(\B^n\Rightarrow A)}{\vdash t:\B^n\Rightarrow A}
          &
          \infer{\vdash \z[\B^n]:S{\B^n}}{}
        }
      }
      \qquad\textrm{and}\qquad
      \vcenter{\infer{\vdash\z:S(A)}{}}
    \]
    Then,
    \begin{center}
      \begin{tikzcd}[row sep=1cm,
        execute at end picture={
          \path
          (\tikzcdmatrixname-1-1) -- (\tikzcdmatrixname-5-1) coordinate[pos=0.5](aux1)
          (\tikzcdmatrixname-1-4) -- (\tikzcdmatrixname-5-4) coordinate[pos=0.5](aux2)
          (aux1) -- (aux2) node[midway,blue,yshift=1cm,xshift=-3.6cm]{\small\eqref{exaxiom:B}}
          (aux1) -- (aux2) node[midway,blue,yshift=8mm]{\small\eqref{exaxiom:D}}
          (aux1) -- (aux2) node[midway,blue,xshift=-2.2cm,yshift=5mm]{\small\eqref{exaxiom:F}}
          (aux1) -- (aux2) node[midway,blue,yshift=2.5mm,xshift=2.5cm]{\small\eqref{exaxiom:E}} 
          (aux1) -- (aux2) node[midway,blue,xshift=3.8cm]{\small\eqref{linrzero:D}}
          (aux1) -- (aux2) node[midway,blue,yshift=-1cm]{\small\eqref{exaxiom:G}}
          (aux1) -- (aux2) node[midway,blue,yshift=-2.5cm,xshift=-3.5cm]{\small\eqref{linrzero:E}}
          (aux1) -- (aux2) node[midway,blue,yshift=2.3cm]{\small\eqref{exaxiom:D}}
          ;
        }]
        {US{\B^n}\times{[\B^n,A]}} \arrow[r, "\Id\times \eta",sloped] &[-7pt] {US{\B^n}\times US{[\B^n,A]}} \arrow[r, "n",sloped] &[-8pt] {U(S{\B^n}\otimes S{[\B^n,A]})} \arrow[r, "Um",sloped] &[-7pt] {US(\B^n\times{[\B^n,A]})} \arrow[dddd, "US\varepsilon^{\B^n}",sloped] \\
         & \color{red}US1\times US1 \arrow[u, "U\mathbf 0\times USt", red,dashed,sloped] \arrow[r, "n", red,dashed,sloped] & \color{red}U(S1\otimes S1) \arrow[u, "U(\mathbf 0\otimes St)", red,dashed,sloped] & \\
         & \color{red}US1\times U\I \arrow[r, "n", red,dashed,sloped] \arrow[u, "\Id\times Um_{\I}"', red,dashed,sloped]  & \color{red}U(S1\otimes\I) \arrow[u, "U(\Id\otimes m_{\I})", red,dashed,sloped]  & \\
         US1\times 1 \arrow[uuu, "U\mathbf 0\times t",sloped] \arrow[rrd, "\lambda_\times^{-1}"', red,dashed, bend right=5,sloped] \arrow[ru, "\Id\times n_1",swap, red,dashed,sloped] \arrow[ruu, "\Id\times\eta",swap, red,dashed,sloped] & & & \\
         1 \approx 1^2 \arrow[rr, "\eta",sloped,near start] \arrow[u,swap, "\eta\times\Id",swap,sloped] & & US1\arrow[ull,"\lambda_\times",red,dashed,sloped,bend right=5]\arrow[uu,"U\lambda",red,dashed]\arrow[r,"U\mathbf 0"]\arrow[uuuu,"U\mathbf 0",red,dashed,bend right=50] & USA 
\end{tikzcd}
    \end{center}
    \begin{enumerate}
    \item\label{exaxiom:B} Naturality of $\eta$ and functoriality of the product.
    \item\label{exaxiom:D} Naturality of $n$.
    \item\label{exaxiom:F} Property of monoidal adjunctions.
    \item\label{exaxiom:E} Notice that $\mathbf 0\otimes (St\circ m_I)=\mathbf 0$, hence, this diagram commutes by property of $\mathbf 0$.
    \item\label{linrzero:D} Property of the morphism $\mathbf 0$.
    \item\label{exaxiom:G} Property of monoidal categories.
    \item\label{linrzero:E} Naturality of $\lambda_\times$.
    \end{enumerate}

  \item[$(\mathsf{lin}^+_l)$] $(t+u)v\lra (tv+uv)$. This case is analogous
    to $(\mathsf{lin}^+_r)$. Notice that the axiom is valid only for $\hat
    +\times\Id$, not for $\Id\times\hat +$, however, one can be transformed
    into the other by using a swap.

  \item[$(\mathsf{lin}^\alpha_l)$] $(\alpha.t)u\lra \alpha.(tu)$. This case
    is analogous to $(\mathsf{lin}^\alpha_r)$.

  \item[$(\mathsf{lin}^0_l)$]  $\z[\B^n\Rightarrow A]t\lra \z$. This case is
    analogous to $(\mathsf{lin}^0_r)$.

  \item[$(\mathsf{neutral})$] $(\z+t)\lra  t$. We have
    \[
      \infer{\vdash\z+t:S(A)}
      {
        \infer{\vdash\z:S(A)}{}
        &
        \vdash t:S(A)
      }
      \qquad\textrm{and}\qquad
      \vdash t:S(A)
    \]
    Then, 
    \begin{center}
      \begin{tikzcd}[
        execute at end picture={
          \path
          (\tikzcdmatrixname-1-1) -- (\tikzcdmatrixname-3-1) coordinate[pos=0.5](aux1)
          (\tikzcdmatrixname-1-4) -- (\tikzcdmatrixname-3-4) coordinate[pos=0.5](aux2)
          (aux1) -- (aux2) node[midway,blue,xshift=2.5mm,yshift=-5mm]{\small\eqref{neutral:A}}
          (aux1) -- (aux2) node[midway,blue,xshift=-2cm,yshift=-5mm]{\small\eqref{neutral:B}}
          (aux1) -- (aux2) node[midway,blue,yshift=-1.2cm]{\small\eqref{neutral:C}}
          (aux1) -- (aux2) node[midway,blue,yshift=-5mm,xshift=2.4cm]{\small\eqref{neutral:D}}
          (aux1) -- (aux2) node[midway,blue,xshift=-1.4cm,yshift=6mm]{\small\eqref{neutral:E}}
          (aux1) -- (aux2) node[midway,blue,xshift=1.9cm,yshift=6mm]{\small\eqref{neutral:G}}
          (aux1) -- (aux2) node[midway,blue,xshift=4.1cm,yshift=6mm]{\small\eqref{neutral:H}}
          ;
        }]
        1\times 1 \arrow[r, "\eta\times t"] \arrow[rd, "t\times t", red,dashed] & US1\times USA \arrow[r, "U\mathbf 0\times\Id"] & USA\times USA \arrow[r, "g_0=\Id"] & USA\times USA \arrow[dd, "\hat +", bend left=70] \arrow[d, "p", red,dashed] \\
        &\color{red} USA\times USA \arrow[ru, "U\mathbf 0\times U\Id", red,dashed] \arrow[r, "p", red,dashed] &\color{red} U(SA\oplus SA) \arrow[r, "U(\mathbf 0\oplus\Id)", red,dashed] &\color{red} U(SA\oplus SA) \arrow[d, "U\nabla", red,dashed] \\
        1 \arrow[uu, "\lambda_\times{\color{red} = \Delta}"] \arrow[rrr, "t", bend right=10,swap] \arrow[rr, "t", red,dashed] & &\color{red} USA \arrow[r, "U(\mathbf 0+\Id)", red,dashed] \arrow[u, "U\Delta", red,dashed] \arrow[lu, "\Delta", red,dashed] & USA                                                                 
      \end{tikzcd}
    \end{center}
    \begin{enumerate}
    \item\label{neutral:E} By \axiomz\ and functoriality of product.
    \item\label{neutral:G} Naturality of $p$.
    \item\label{neutral:H} Definition of $\hat +$.
    \item\label{neutral:B} Naturality of $\Delta$.
    \item\label{neutral:A} $U$ preserves product.
    \item\label{neutral:D} Property of sum in an additive category.
    \item\label{neutral:C} The map $\mathbf 0$ is neutral with respect to the sum.
    \end{enumerate}

  \item[$(\mathsf{unit})$] $1.t\lra t$. We have
    \[
      \infer{\vdash 1.t:S(A)}
      {
        \vdash t:S(A)
      }
      \qquad\textrm{and}\qquad
      \vdash t:S(A)
    \]
    Then,
    \begin{center}
      \begin{tikzcd}[column sep=15mm]
        USA\ar[r,"U\lambda"]\ar[rrd,red,dashed,"\Id"] &U(SA\otimes\I)\ar[r,"U(\Id\otimes 1)"] & U(SA\otimes\I)\ar[d,"U\lambda^{-1}"]\\
        1\ar[u,"t"]\ar[rr,"t"] & & USA
      \end{tikzcd}
    \end{center}

  \item[$(\mathsf{zero}_\alpha)$] If $t$ has type $A$, $0.t\lra \z$. We have
    \[
      \vcenter{\infer{\vdash 0.t:S(A)}{\vdash t:S(A)}}
      \qquad\textrm{and}\qquad
      \vcenter{\infer{\vdash\z:S(A)}{}}
    \]
    Then,
    \begin{center}
      \begin{tikzcd}[column sep=15mm,row sep=1cm,
        execute at end picture={
          \path
	  (\tikzcdmatrixname-1-1) -- (\tikzcdmatrixname-2-1) coordinate[pos=0.5](aux1)
          (\tikzcdmatrixname-1-3) -- (\tikzcdmatrixname-2-3) coordinate[pos=0.5](aux2)
	  (aux1) -- (aux2) node[midway,blue,xshift=-2cm,yshift=-3mm]{\small\eqref{zeroalpha:A}}
	  (aux1) -- (aux2) node[midway,blue,xshift=2cm]{\small\eqref{zeroalpha:D}}
          (aux1) -- (aux2) node[midway,blue,yshift=4mm]{\small\eqref{zeroalpha:D}}
	  ;
        }]
	USA\ar[r,"U\lambda"]\ar[rrd,red,dashed,"U\mathbf 0",out=315,in=165]\ar[rr,bend right=20,red,dashed,"U\mathbf 0"] & U(SA\otimes\I)\ar[r,"U(\Id\otimes 0){\color{red}=U\mathbf 0}"] & U(SA\otimes\I)\ar[d,"U(\lambda^{-1})"]\\
	1\ar[r,"\eta"]\ar[u,"t"] & US1\ar[r,"U\mathbf 0"] & USA
      \end{tikzcd}
    \end{center}
    \begin{enumerate}
    \item\label{zeroalpha:A} \axiomz.
    \item\label{zeroalpha:D} Property of the map $\mathbf 0$.
    \end{enumerate}

  \item[$(\mathsf{zero})$] $\alpha.\z\lra \z$. We have
    \[
      \vcenter{
        \infer{\vdash\alpha.\z:S(A)}
        {
          \infer{\vdash\z:S(A)}{}
        }
      }
      \qquad\textrm{and}\qquad
      \vcenter{
        \infer{\vdash\z:S(A)}{}
      }
    \]
    Then
    \begin{center}
      \begin{tikzcd}[
	execute at end picture={
	  \path
	  (\tikzcdmatrixname-1-1) -- (\tikzcdmatrixname-2-1) coordinate[pos=0.5] (aux1)
	  (\tikzcdmatrixname-1-3) -- (\tikzcdmatrixname-2-3) coordinate[pos=0.5] (aux2)
	  (aux1) -- (aux2) node[midway,blue] {\small\eqref{zero:A}}
	  ;
	  \path
	  (\tikzcdmatrixname-1-4) -- (\tikzcdmatrixname-2-4) coordinate[pos=0.3] (aux1)
	  (\tikzcdmatrixname-1-5) -- (\tikzcdmatrixname-2-5) coordinate[pos=0.3] (aux2)
	  (aux1) -- (aux2) node[midway,blue] {\small\eqref{zero:B}}
	  ;
      }]
	1 \arrow[d, "\eta"] \arrow[r, "\eta"] & US1 \arrow[r, "U\mathbf 0"] \arrow[rrrd, "U\mathbf 0", red,dashed] & US(A) \arrow[r, "U(\lambda)"] & U(S(A)\otimes\I) \arrow[r, "U(\Id\otimes\alpha)"] & U(SUS(A)\otimes\I) \arrow[d, "U(\lambda^{-1})"] \\
	US1 \arrow[rrrr, "U\mathbf 0"] & & {\ }& {\ } & US(A)
\end{tikzcd}
    \end{center}
    \begin{enumerate}
      \item\label{zero:A} Same arrows.
      \item\label{zero:B} Property of the morphism $\mathbf 0$.
    \end{enumerate}

  \item[$(\mathsf{prod})$] $\alpha.(\beta.t)\lra (\alpha\beta).t$. We have
    \[
      \vcenter{
        \infer{\vdash\alpha.(\beta.t):S(A)}
        {
          \infer{\vdash\beta.t:S(A)}
          {\vdash t:S(A)} 
        }
      }
      \qquad\textrm{and}\qquad
      \vcenter{
        \infer{\vdash(\alpha\beta).t:S(A)}
        {\vdash t:S(A)}
      }
    \]
    Then
    \begin{center}
      \begin{tikzcd}[column sep=15mm]
        U(SA\otimes\I)\ar[rr,"U(\Id\otimes\beta)"] && U(SA\otimes\I)\ar[r,"U\lambda^{-1}"]\ar[rd,red,dashed,"\Id"] & USA\ar[d,"U\lambda"]\\
        USA\ar[u,"U\lambda"] &&& U(SA\otimes\I)\ar[d,"U(\Id\otimes\alpha)"]\\
        {1}\ar[d,"t"]\ar[u,"t"] &&& U(SA\otimes\I)\ar[d,"U\lambda^{-1}"]\\
        USA\ar[r,"U\lambda"] &
        U(SA\otimes\I)\ar[r,"U(\Id\otimes(\alpha.\beta))"]& U(SA\otimes\Id)\ar[r,"U\lambda^{-1}"]& USA
      \end{tikzcd}
    \end{center}
    This diagram commutes by functoriality of $\otimes$.

  \item[$(\alpha\mathsf{dist})$] $\alpha.(t+u)\lra \alpha.t+\alpha.u$. We
    have
    \[
      \vcenter{
        \infer{\vdash\alpha.(t+u):S(A)}
        {
          \infer{\vdash t+u:S(A)}
          {\vdash t:S(A) & \vdash u:S(A)}
        }
      }
      \qquad\textrm{and}\qquad
      \vcenter{
        \infer{\vdash\alpha.t+\alpha.u:S(A)}
        {
          \infer{\vdash\alpha.t:S(A)}{\vdash t:S(A)}
          &
          \infer{\vdash\alpha.u:S(A)}{\vdash u:S(A)}
        }
      }
    \]
    Then
    \begin{center}
      \begin{tikzcd}[column sep=13mm]
        (USA)^2\ar[r,"\hat +"] & USA\ar[r,"U\lambda"]& U(SA\otimes\I)\ar[d,"U(\Id\otimes\alpha)"]\\
        (USA)^2\ar[r,"(U\lambda)^2"]\ar[u,"g_0"] &
        (U(SA\otimes\I))^2\ar[ru,red,dashed,"\hat +"]\ar[dd,"(U(\Id\otimes\alpha))^2"]& U(SA\otimes\I)\ar[d,"U\lambda^{-1}"]\\
        {1}^2\ar[u,"t\times u"] && USA\\
        {1}\ar[u,"\lambda_\times"] & (U(SA\otimes\I))^2\ar[ruu,red,dashed,"\hat +",swap]\ar[r,"(U\lambda^{-1})^2"]& USA^2\ar[u,"\hat +"]\\
      \end{tikzcd}
    \end{center}
    The three subdiagrams are valid by the naturality of $\hat +$.

  \item[$(\mathsf{fact})$] $(\alpha.t+\beta.t)\lra (\alpha+\beta).t$. We have
    \[
      \vcenter{
        \infer{\vdash(\alpha.t+\beta.t):S(A)}
        {
          \infer{\vdash\alpha.t:S(A)}{\vdash t:S(A)}
          &
          \infer{\vdash\beta.t:S(A)}{\vdash t:S(A)}
        }
      }
      \qquad\textrm{and}\qquad
      \vcenter{
        \infer{\vdash(\alpha+\beta).t:S(A)}{\vdash t:S(A)}
      }
    \]
    Then
    \begin{center}
      \begin{tikzcd}[
        execute at end picture={
          \path
          (\tikzcdmatrixname-1-2) -- (\tikzcdmatrixname-5-2) coordinate[pos=0.5](aux1)
          (\tikzcdmatrixname-1-5) -- (\tikzcdmatrixname-5-5) coordinate[pos=0.5](aux2)
          (aux1) -- (aux2) node[midway,blue,xshift=-2mm]{\small\eqref{soundfact:A}}
          (aux1) -- (aux2) node[midway,blue,yshift=1.7cm,xshift=-2mm]{\small\eqref{soundfact:B}}
          (aux1) -- (aux2) node[midway,blue,yshift=3mm,xshift=3.3cm]{\small\eqref{soundfact:C}}
          (aux1) -- (aux2) node[midway,blue,xshift=2.3cm]{\small\eqref{soundfact:A}}
          (aux1) -- (aux2) node[midway,blue,xshift=4cm]{\small\eqref{soundfact:D}}
          (aux1) -- (aux2) node[midway,blue,xshift=-5.45cm,yshift=5mm]{\small\eqref{soundfact:E}}
          (aux1) -- (aux2) node[midway,blue,xshift=-4.75cm,yshift=1cm]{\small\eqref{soundfact:E}}
          (aux1) -- (aux2) node[midway,blue,xshift=-3cm]{\small\eqref{soundfact:A}}
          (aux1) -- (aux2) node[midway,blue,xshift=-4cm]{\small\eqref{soundfact:F}}
          (aux1) -- (aux2) node[midway,blue,yshift=-1.5cm,xshift=-2mm]{\small\eqref{soundfact:G}}
          ;
        }]
        &[-30pt] (USA)^2 \arrow[r, "(U\lambda)^2"] &[-10pt] U(SA\otimes\I)\times U(SA\otimes\I) \arrow[r,"\raisebox{1.5mm}{$U(\Id\otimes\alpha)\times U(\Id\otimes\beta)$}"] &[-2pt] U(SA\otimes\I)\times U(SA\otimes\I) \arrow[r, "(U\lambda^{-1})^2"] \arrow[dddd, "\hat +"', red,dashed, bend left=73] \arrow[d, "p", red,dashed, bend left] & (USA)^2 \arrow[dd, "g_0=\Id",swap] \\
        1^2 \arrow[ru, "t^2"]  &  &\color{red} U((SA\otimes\I)\oplus(SA\otimes\I)) \arrow[r,"\raisebox{1.5mm}{$U((\Id\otimes\alpha)\oplus(\Id\otimes\beta))$}", red,dashed] \arrow[u, "p^{-1}", red,dashed] &\color{red} U((SA\otimes\I)\oplus(SA\otimes\I)) \arrow[ddd, "U\nabla", red,dashed, bend left=62] \arrow[u, "p^{-1}", red,dashed] & \\
        & & & & (USA)^2 \arrow[dd, "\hat +",swap]  \\
        1 \arrow[d, "t",swap] \arrow[uu, "\lambda_\times{\color{red}=\Delta}",sloped] & &\color{red}U(SA\otimes(\I\oplus\I)) \arrow[r, "U(\Id\otimes(\alpha\oplus\beta))", red,dashed] \arrow[uu, "U\delta", red,dashed]  &\color{red} U(SA\otimes(\I\oplus\I)) \arrow[uu, "U\delta", red,dashed] \arrow[d, "U(\Id\otimes\nabla)"', red,dashed] & \\
        USA \arrow[rr,swap, "U\lambda"] \arrow[ruuuu, "\Delta", red,dashed] &\ & U(SA\otimes\I) \arrow[r,swap, "U(\Id\otimes(\alpha+\beta))"] \arrow[uuuu, "\Delta", red,dashed, bend left=80] \arrow[uuu, "U\Delta", red,dashed,bend left=65] \arrow[u, "U(\Id\otimes\Delta)"', red,dashed] & U(SA\otimes\I) \arrow[r, "U\lambda^{-1}"'] & USA
      \end{tikzcd}
    \end{center}
    \begin{enumerate}
    \item \label{soundfact:E} Naturality of $\Delta$.
    \item \label{soundfact:F} Functor $U$ preserves product.
    \item \label{soundfact:A} Distributivity property given by the fact that the tensor is a left adjoint.
    \item \label{soundfact:B} Naturality of $p^{-1}$.  
    \item \label{soundfact:G} Additivity of the category.
    \item \label{soundfact:C} Definition of $\hat +$.
    \item \label{soundfact:D} Naturality of $\hat +$.
    \end{enumerate}

  \item[$(\mathsf{fact}^1)$] $(\alpha.t+t)\lra (\alpha+1).t$. This case is a
    particular case of $\mathsf{fact}$.
  \item[$(\mathsf{fact}^2)$] $(t+t)\lra 2.t$. This case is a particular case
    of $\mathsf{fact}$.

  \item[$(\mathsf{head})$] If $h\neq u\times v$, and $h\in\tbasis$, $\head\ h\times t\lra h$. We have
    \[
      \vcenter{
        \infer{\vdash\head\ h\times t:\B}{
          \infer{\vdash h\times t:\B^n}{
            \vdash  h:\B & \vdash t:\B^{n-1}
          }
        }
      }
      \qquad\textrm{and}\qquad
      \vdash h:\B
    \]
    Then
    \begin{center}
      \begin{tikzcd}
        {1}^2\ar[r,"h\times t"] &\B^n\ar[d,"\pi_1"]\\
        {1}\ar[u,"\lambda_\times"]\ar[r,"h"] & \B
      \end{tikzcd}
    \end{center}
    This diagram commutes since $\pi_1$ is just the projection.

  \item[$(\mathsf{tail})$] If $h\neq u\times v$, and $h\in\tbasis$, $\tail\ h\times t\lra t$. We have
    \[
      \vcenter{
        \infer{\vdash\tail\ h\times t:\B^{n-1}}{
          \infer{\vdash h\times t:\B^n}{
            \vdash  h:\B & \vdash t:\B^{n-1}
          }
        }
      }
      \qquad\textrm{and}\qquad
      \vdash t:\B^{n-1}
    \]
    Then
    \begin{center}
      \begin{tikzcd}
        {1}^2\ar[r,"h\times t"] &\B^n\ar[d,"\pi_2"]\\
        {1}\ar[u,"\lambda_\times"]\ar[r,"t"] & \B^{n-1}
      \end{tikzcd}
    \end{center}
    This diagram commutes since $\pi_2$ is just the projection.

  \item[($\mathsf{dist}_r^+$)] $\Uparrow_r ((r+s)\times
    u)\lra \Uparrow_r (r\times u)+\Uparrow_r (s\times u)$.

    We have,
    \[
      \vcenter{
        \infer{\vdash\Uparrow_r ((r+s)\times u):S(\Psi\times \Phi)}
        {
          \infer{\vdash(r+s)\times u:S(S\Psi\times \Phi)}
          {
            \infer{\vdash (r+s)\times u:S\Psi\times \Phi}
            {
              \infer{\vdash r+s:S\Psi}{\vdash r:S\Psi & \vdash s:S\Psi}
              &
              \vdash u:\Phi
            }
          }
        }
      }
      \quad\textrm{and}\quad
      \vcenter{
        \infer{\vdash\Uparrow_r (r\times u)+\Uparrow_r (s\times u):S(\Psi\times \Phi)}
        {
          \infer{\vdash\Uparrow_r (r\times u):S(\Psi\times \Phi)}
          {
            \infer{\vdash r\times u:S(S\Psi\times \Phi)}
            {
              \infer{\vdash r\times u:S\Psi\times \Phi}
              {
                \vdash r:S\Psi
                &
                \vdash u:\Phi
              }
            }
          }
          &
          \infer{\vdash\Uparrow_r (s\times u):S(\Psi\times \Phi)}
          {
            \infer{\vdash s\times u:S(S\Psi\times \Phi)}
            {
              \infer{\vdash s\times u:S\Psi\times \Phi}
              {
                \vdash s:S\Psi
                &
                \vdash u:\Phi
              }
            }
          }
        }
      }
    \]
    Then
    
    \begin{center}
      \begin{tikzcd}[column sep=12mm,
        execute at end picture={
          \path (\tikzcdmatrixname-2-2) -- (\tikzcdmatrixname-4-2) coordinate[pos=0.5] (aux1)
          (\tikzcdmatrixname-2-3) -- (\tikzcdmatrixname-4-3) coordinate[pos=0.5] (aux2)
          (aux1) -- (aux2) node[midway,sloped,blue]{\small\eqref{distrplus:A} }
          (aux1) -- (aux2) node[xshift=-4cm,yshift=1.7cm,midway,sloped,blue]{\small\eqref{distrplus:B}}
          (\tikzcdmatrixname-1-2) -- (\tikzcdmatrixname-2-2) coordinate[pos=0.5] (aux1)
          (\tikzcdmatrixname-1-4) -- (\tikzcdmatrixname-2-4) coordinate[pos=0.5] (aux2)
          (aux1) -- (aux2) node[midway,sloped,blue]{\small\eqref{distrplus:C} };
          \path (aux1) -- (aux2)
          node[midway,blue,xshift=2cm,yshift=-1cm]{\small\eqref{distrplus:D}};
          \path (\tikzcdmatrixname-2-1) -- (\tikzcdmatrixname-4-1) coordinate[pos=0.5] (aux1)
          (\tikzcdmatrixname-2-2) -- (\tikzcdmatrixname-4-2) coordinate[pos=0.5] (aux2)
          (aux1) -- (aux2) node[midway,sloped,blue]{\small \eqref{distrplus:E}};
          \path (\tikzcdmatrixname-4-1) -- (\tikzcdmatrixname-5-1) coordinate[pos=0.5] (aux1)
          (\tikzcdmatrixname-4-2) -- (\tikzcdmatrixname-5-2) coordinate[pos=0.5] (aux2)
          (aux1) -- (aux2) node[midway,sloped,blue]{\small \eqref{distrplus:C}};
          \path (\tikzcdmatrixname-3-3) -- (\tikzcdmatrixname-4-3) coordinate[pos=0.5] (aux1)
          (\tikzcdmatrixname-3-4) -- (\tikzcdmatrixname-4-4) coordinate[pos=0.5] (aux2)
          (aux1) -- (aux2) node[midway,sloped,blue]{\small \eqref{distrplus:H}};
          \path (\tikzcdmatrixname-4-2) -- (\tikzcdmatrixname-5-2) coordinate[pos=0.5] (aux1)
          (\tikzcdmatrixname-4-3) -- (\tikzcdmatrixname-5-3) coordinate[pos=0.5] (aux2)
          (aux1) -- (aux2) node[midway,sloped,blue]{\small \eqref{distrplus:C}};
          \path (\tikzcdmatrixname-4-3) -- (\tikzcdmatrixname-5-3) coordinate[pos=0.5] (aux1)
          (\tikzcdmatrixname-4-4) -- (\tikzcdmatrixname-5-4) coordinate[pos=0.5] (aux2)
          (aux1) -- (aux2) node[midway,sloped,blue]{\small \eqref{distrplus:D}};
        }]
        US\Psi\times\Phi \arrow[r, "\eta"] \arrow[rrd,red,dashed,"\Id\times\eta"] &[-30pt] US(US\Psi\times\Phi) \arrow[r, "US(\Id\times\eta)"]  &[-20pt] US(US\Psi\times US\Phi) \arrow[r, "US(n)"] &[-25pt] USU(S\Psi\otimes S\Phi) \arrow[d, "USUm"]\\
        (US\Psi)^2\times\Phi \arrow[u, "\hat +\times\Id"] \arrow[dd, "d"'] \arrow[r,red,dashed, "\Id\times\eta"] & \color{red}(US\Psi)^2\times US\Phi \arrow[dd,red,dashed, "d"] \arrow[r,red,dashed, "\hat +\times\Id"'] & \color{red}US\Psi\times US\Phi \arrow[d,red,dashed, "n"] & USUS(\Psi\times\Phi) \arrow[d, "\mu"] \\
        & &\color{red} U(S\Psi\otimes S\Phi)
        \arrow[ruu,red,dashed,"\eta"] \arrow[r,red,dashed, "Um"] &
        US(\Psi\times\Phi)\arrow[u,red,dashed,"\eta",bend left] \\
        (US\Psi\times\Phi)^2 \arrow[d, "\eta^2"] \arrow[r,red,dashed,"(\Id\times\eta)^2"] & \color{red}(US\Psi\times US\Phi)^2 \arrow[r,red,dashed, "n^2"] \arrow[d,red,dashed, "\eta^2"'] &\color{red} (U(S\Psi\otimes S\Phi))^2 \arrow[d,red,dashed, "\eta^2"] \arrow[r,red,dashed, "(Um)^2"] \arrow[u,red,dashed, "\hat +"] & (US(\Psi\times\Phi))^2 \arrow[u, "\hat +"] \arrow[d,red,dashed, "\eta^2", swap,bend right] \\
        (US(US\Psi\times\Phi))^2 \arrow[r,swap, "\raisebox{1.5mm}{$(US(\Id\times\eta))^2$}"]                                                & (US(US\Psi\times US\Phi))^2 \arrow[r, "\raisebox{1.5mm}{$(USn)^2$}",swap]                                   & (USU(S\Psi\otimes S\Phi))^2 \arrow[r, "\raisebox{1.5mm}{$(USUm)^2$}",swap]                                     & (USUS(\Psi\times\Phi))^2 \arrow[u, "\mu^2",swap]                                         
      \end{tikzcd}
    \end{center}
    \begin{enumerate}
    \item \label{distrplus:B} Functoriality of the product.
    \item \label{distrplus:C} Naturality of $\eta$.
    \item \label{distrplus:D} Naturality of $\eta$ (remark that by the axioms of monads, $\mu\circ\eta=\Id$).
    \item \label{distrplus:E} Naturality of $d$.
    \item \label{distrplus:A} \axiomd\ and definition of $\hat +$.
    \item \label{distrplus:H} Naturality of $\hat +$.
    \end{enumerate}
  \item[($\mathsf{dist}_l^+$)] $\Uparrow_\ell u\times
    (r+s)\lra \Uparrow_\ell (u\times r)+\Uparrow_\ell (u\times s)$.
    Analogous to case $(\mathsf{dist}_r^+)$

  \item[$\rdistscalr$] $\Uparrow_r (\alpha.r)\times u\lra \alpha.\Uparrow_r r\times u$.
    We have,
    \[
      \vcenter{
        \infer{\vdash\Uparrow_r (\alpha.r)\times u:S(\Psi\times \Phi)}
        {
          \infer{\vdash(\alpha.r)\times u:S(S\Psi\times \Phi)}
          {
            \infer{\vdash (\alpha.r)\times u:S\Psi\times \Phi}
            {
              \infer{\vdash \alpha.r:S\Psi}{\vdash r:S\Psi}
              &
              \vdash u:\Phi
            }
          }
        }
      }
      \qquad\textrm{and}\qquad
      \vcenter{
        \infer{\vdash\alpha.\Uparrow_r r\times u:S(\Psi\times \Phi)}
        {
          \infer{\vdash\Uparrow_r (r\times u):S(\Psi\times \Phi)}
          {
            \infer{\vdash r\times u:S(S\Psi\times \Phi)}
            {
              \infer{\vdash r\times u:S\Psi\times \Phi}
              {
                \vdash r:S\Psi
                &
                \vdash u:\Phi
              }
            }
          }
        }
      }
    \]
    Then
    \begin{center}
      \begin{tikzcd}[
        execute at end picture={
          \path
          (\tikzcdmatrixname-1-1) -- (\tikzcdmatrixname-8-1) coordinate[pos=0.5](aux1)
          (\tikzcdmatrixname-1-4) -- (\tikzcdmatrixname-8-4) coordinate[pos=0.5](aux2)
          (aux1) -- (aux2) node[midway,blue,yshift=2cm]{\small\eqref{distalphar:G}}
          (aux1) -- (aux2) node[midway,blue,yshift=-1cm]{\small\eqref{distalphar:D}}
          (aux1) -- (aux2) node[midway,blue,xshift=3cm]{\small\eqref{distalphar:D}}
          (aux1) -- (aux2) node[midway,blue,xshift=-3cm,yshift=-2cm]{\small\eqref{distalphar:E}}
          (aux1) -- (aux2) node[midway,blue,yshift=-2cm]{\small\eqref{distalphar:F}}
          (aux1) -- (aux2) node[midway,blue,yshift=-2cm,xshift=2.5cm]{\small\eqref{distalphar:C}}
          (aux1) -- (aux2) node[midway,blue,yshift=-2cm,xshift=4.2cm]{\small\eqref{distalphar:A}}
          (aux1) -- (aux2) node[midway,blue,xshift=-3cm,yshift=-3cm]{\small\eqref{distalphar:D}}
          (aux1) -- (aux2) node[midway,blue,yshift=-3cm]{\small\eqref{distalphar:D}}
          (aux1) -- (aux2) node[midway,blue,xshift=3cm,yshift=-3cm]{\small\eqref{distalphar:D}}
          ;
        }]
      US\Psi\times \Phi \arrow[r, "\eta"] &[-15pt] US(US\Psi\times \Phi) \arrow[r, "U(\Id\times\eta)"] & [-8pt] US(US\Psi\times US\Phi) \arrow[r, "USn"] & [-15pt] USU(S\Psi\otimes S\Phi) \arrow[d, "USUm"] \\
      U(S\Psi\otimes\I)\times\Phi \arrow[u,swap,"U\lambda^{-1}\times\Id"] &\color{red} US(U(S\Psi\otimes\I)\times\Phi) \arrow[d, "US(\Id\times\eta)", red,dashed] & & USUS(\Psi\times \Phi) \arrow[d, "\mu"] \\
      &\color{red} US(U(S\Psi\otimes\I)\times US\Phi) \arrow[d, "USn", red,dashed] &\color{red} USUS(\Psi\times\Phi) \arrow[r, "\mu", red,dashed] & US(\Psi\times \Phi) \arrow[l, "\eta", red,dashed, bend left] \\
      &\color{red} USU(S\Psi\otimes\I\otimes S\Phi) \arrow[d, "USU(\Id\otimes\sigma)", red,dashed] &\color{red} USU(S(\Psi\times\Phi)\otimes\I) \arrow[u, "USU\lambda^{-1}", red,dashed] & U(S(\Psi\times\Phi)\otimes\I) \arrow[u, "U\lambda^{-1}"] \\
      &\color{red} USU(S\Psi\otimes S\Phi\otimes\I) \arrow[r, "\raisebox{1.5mm}{$USU(m\otimes\Id)$}", red,dashed] &\color{red} USU(S(\Psi\times\Phi)\otimes\I) \arrow[u, "USU(\Id\otimes\alpha)", red,dashed] & U(S(\Psi\times\Phi)\otimes\I) \arrow[u, "U(\Id\otimes\alpha)"] \arrow[l, "\eta", red,dashed] \\
      U(S\Psi\otimes\I)\times\Phi \arrow[r, "\Id\times\eta", red,dashed] \arrow[uuuu,sloped, "U(\Id\otimes\alpha)\times\Id"] \arrow[ruuuu,out=45,in=180,"\eta", red,dashed] &\color{red} U(S\Psi\otimes\I)\times US\Phi \arrow[r, "n", red,dashed] &\color{red} U(S\Psi\otimes\I\otimes S\Phi) \arrow[d,swap,"U(\lambda^{-1}\otimes\Id)", red,dashed] \arrow[r, "U(\Id\otimes\sigma)", red,dashed] & U(S\Psi\otimes S\Phi\otimes\I) \arrow[u, "U(m\otimes\Id)", red,dashed] \\
      US\Psi\times\Phi \arrow[u, "U\lambda\times\Id"] \arrow[rdd, "\eta",bend right] \arrow[r, "\Id\times\eta", red,dashed] &\color{red} US\Psi\times US\Phi \arrow[r, "n", red,dashed] \arrow[u,"U\lambda\times\Id", red,dashed] \arrow[d, "\eta", red,dashed] &\color{red} U(S\Psi\otimes S\Phi) \arrow[r, "Um", red,dashed] \arrow[d, "\eta", red,dashed] \arrow[ru, "U\lambda", red,dashed] & US(\Psi\times\Phi) \arrow[uu, "U\lambda", bend right=70,near end] \arrow[d, "\eta"', red,dashed, bend right] \\
      {1}^2 \arrow[u, "r\times u"] & US(US\Psi\times US\Phi) \arrow[r, "USn"] & US(U(S\Psi\otimes S\Phi)) \arrow[r, "USUm"] & USUS(\Psi\times\Phi) \arrow[u, "\mu"] \\
      & US(US\Psi\times\Phi) \arrow[u, "U(\Id\times\eta)"] & & 
\end{tikzcd}
    \end{center}
    \begin{enumerate}
    \item\label{distalphar:G} See next diagram.
    \item\label{distalphar:D} Naturality of $\eta$.
    \item\label{distalphar:E} Functoriality of product.
    \item\label{distalphar:F} Naturality of $n$.
    \item\label{distalphar:C} Coherence.
    \item\label{distalphar:A} Naturality of $\lambda$.
    \end{enumerate}

    \begin{center}
      \adjustbox{scale=.95,center}{
      \begin{tikzcd}
        [
        execute at end picture={
          \path (\tikzcdmatrixname-5-2) -- (\tikzcdmatrixname-1-4) node[midway,blue,yshift=-4mm]{\small\eqref{diag3:4}};
          \path (\tikzcdmatrixname-1-3) -- (\tikzcdmatrixname-4-3) coordinate[pos=0.5] (aux1)
          (\tikzcdmatrixname-1-4) -- (\tikzcdmatrixname-5-4) coordinate[pos=0.5] (aux2)
          (aux1) -- (aux2) node[midway,blue,xshift=6mm,yshift=-5mm]{\small\eqref{diag3:5}};
          \path (\tikzcdmatrixname-1-1) -- (\tikzcdmatrixname-2-1) coordinate[pos=0.5](aux1)
          (\tikzcdmatrixname-1-2) -- (\tikzcdmatrixname-3-2) coordinate[pos=0.5](aux2)
          (aux1) -- (aux2) node[midway,blue]{\small\eqref{diag3:1}};
          \path (aux1) -- (aux2) node[midway,blue,yshift=-1.6cm]{\small\eqref{diag3:1}};
          \path (aux1) -- (aux2) node[midway,blue,yshift=-2.7cm]{\small\eqref{diag3:2}};
          \path (aux1) -- (aux2) node[midway,blue,xshift=4.5cm]{\small\eqref{diag3:2}};
          \path (aux1) -- (aux2) node[midway,blue,xshift=7cm]{\small\eqref{diag3:3}};
          \path (aux1) -- (aux2) node[midway,blue,yshift=-4cm]{\small\eqref{diag3:3}};
          \path (aux1) -- (aux2) node[midway,blue,yshift=-3.8cm,xshift=4cm]{\small\eqref{diag3:6}};
          \path (aux1) -- (aux2) node[midway,blue,yshift=-3.8cm,xshift=7cm]{\small\eqref{diag3:7}};
        }]
        US\Psi\times\Phi\ar[r,"\eta"] & US(US\Psi\times\Phi)\ar[r,"U(\Id\times\eta)"] & US(US\Psi\times US\Phi)\ar[r,"USn"] & USU(S\Psi\otimes S\Phi)\ar[d,"USUm"]\\
        U(S\Psi\otimes\I)\times\Phi\ar[rd,red,dashed,"\eta"]\ar[u,"U\lambda^{-1}\times\Id"] &&& USUS(\Psi\times\Phi)\ar[d,"\eta"] \\
        &\color{red}US(U(S\Psi\otimes\I)\times\Phi)\ar[uu,red,dashed,"US((U\lambda^{-1})\times\Id)"]\ar[d,red,dashed,"US(\Id\times\eta)"]&& US(\Psi\times\Phi)\\
        U(S\Psi\otimes\I)\times\Phi\ar[uu,"U(\Id\otimes\alpha)\times\Id"]\ar[d,"\eta"]
        &\color{red}US(U(S\Psi\otimes\I)\times
        US\Phi)\ar[ruuu,red,dashed,"US((U\lambda^{-1})\times\Id)",out=20,in=270,sloped]\ar[d,red,dashed,"USn"] &\color{red}USU(S\Psi\otimes S\Phi\otimes\I)\ar[ruuu,red,dashed,sloped,"USU\lambda^{-1}"]\ar[rd,red,dashed,sloped,"USU(m\otimes\Id)",swap]&USUS(\Psi\times\Phi)\ar[u,"\eta"]\\
        US(U(S\Psi\otimes\I)\times\Phi)\ar[d,"US(\Id\times\eta)"]\ar[ruu,red,dashed,"US(U(\Id\otimes\alpha)\times\Id)",sloped] &\color{red}USU(S\Psi\otimes\I\otimes S\Phi)\ar[ru,red,dashed,sloped,"USU(\Id\otimes\sigma)"]\ar[rruuuu,red,dashed,sloped,"USU(\lambda^{-1}\otimes\Id)"]&& USU(S(\Psi\times\Phi)\otimes\I)\ar[u,"USU\lambda^{-1}"] \\
        US(U(S\Psi\otimes\I)\times US\Phi)\ar[r,swap,"USn"]\ar[ruu,red,dashed,"US(U(\Id\otimes\alpha)\times\Id)",sloped,bend
        right=8] & USU(S\Psi\otimes\I\otimes S\Phi)\ar[u,red,dashed,"USU(\Id\otimes\alpha\otimes\Id)"]\ar[r,swap,"USU(\Id\otimes\sigma)"] & USU(S\Psi\otimes S\Phi\otimes\I)\ar[r,swap,yshift=-1mm,"USU(m\otimes\Id)"]\ar[uu,red,dashed,"USU(\Id\otimes\alpha)"] &USU(S(\Psi\times\Phi)\otimes\I)\ar[u,"USU(\Id\otimes\alpha)"]
      \end{tikzcd}
      }
    \end{center}
    \begin{enumerate}
    \item\label{diag3:1} Naturality of $\eta$.
    \item\label{diag3:2} Functoriality of the product.
    \item\label{diag3:3} Naturality of $n$.
    \item\label{diag3:4} Coherence
    \item\label{diag3:5} Naturality of $\lambda$ and functoriality of $USU$.
    \item\label{diag3:6} Naturality of $\sigma$.
    \item\label{diag3:7} Functoriality of the tensor.
    \end{enumerate}

  \item[\rdistscall]  $\Uparrow_\ell
    u\times(\alpha.r)\lra \alpha.\Uparrow_\ell u\times r$. Analogous to
    case $\rdistscalr$.

  \item[\rdistzr] If $u$ has type $\Phi$, $\Uparrow_r\z[\Psi]\times
    u\lra\z[\Psi\times \Phi]$. We have
    \[
      \vcenter{
        \infer{\vdash\Uparrow_r\z[\Psi]\times u:S(\Psi\times \Phi)}
        {
          \infer{\vdash\z[\Psi]\times u:S(S\Psi\times \Phi)}
          {
            \infer{\vdash\z[\Psi]\times u:S\Psi\times \Phi}
            {
              \infer{\vdash\z[\Psi]:S\Psi}{}
              &
              \vdash u:\Phi
            }
          }
        }
      }
      \qquad\textrm{and}\qquad
      \vcenter{\infer{\vdash\z[\Psi\times \Phi]:S(\Psi\times \Phi)}{}}
    \]

    Then
    \begin{center}
      \begin{tikzcd}[row sep=1cm,
        execute at end picture={
          \path
          (\tikzcdmatrixname-1-1) -- (\tikzcdmatrixname-6-1) coordinate[pos=0.5](aux1)
          (\tikzcdmatrixname-1-4) -- (\tikzcdmatrixname-6-4) coordinate[pos=0.5](aux2)
          (aux1) -- (aux2) node[midway,blue,xshift=-3cm,yshift=1cm] {\small\eqref{distrzero:B}}
          (aux1) -- (aux2) node[midway,blue,xshift=-2.5cm,yshift=3cm] {\small\eqref{distrzero:C}}
          (aux1) -- (aux2) node[midway,blue,yshift=2.5cm] {\small\eqref{distrzero:D}}
          (aux1) -- (aux2) node[midway,blue,xshift=2.5cm,yshift=1.5cm] {\small\eqref{distrzero:C}}
          (aux1) -- (aux2) node[midway,blue,xshift=-2.5cm,yshift=-.7cm] {\small\eqref{distrzero:A}}
          (aux1) -- (aux2) node[midway,blue,xshift=-3mm,yshift=-.7cm] {\small\eqref{distrzero:E}}
          (aux1) -- (aux2) node[midway,blue,xshift=2.5cm,yshift=-2cm] {\small\eqref{distrzero:C}}
          (aux1) -- (aux2) node[midway,blue,xshift=-3cm,yshift=-1.8cm] {\small\eqref{distrzero:F}}
          (aux1) -- (aux2) node[midway,blue,xshift=-1cm,yshift=-3cm] {\small\eqref{distrzero:C}}
          ;
        }]
        1^2 \arrow[r, "\eta\times u"] \arrow[rdd, "\eta^2", red,dashed] & US1\times \Phi \arrow[d, "\Id\times\eta", red,dashed] \arrow[r, "U\mathbf 0\times\Id"] & US\Psi\times\Phi \arrow[r, "\eta"] \arrow[ddd, "\Id\times\eta", red,dashed] & US(US\Psi\times\Phi) \arrow[ddd, "US(\Id\times\eta)"] \\
        &\color{red} US1\times US\Phi \arrow[rdd, "U\mathbf 0\times\Id", red,dashed] & & \\
        1 \arrow[uu, "\Delta"] \arrow[d, "\eta"] &\color{red} (US1)^2 \arrow[u, "\Id\times USu",sloped, red,dashed] \arrow[d, "n", red,dashed] & & \\
        US1 \arrow[rd, "U\mathbf 0", red,dashed] \arrow[d, "U\mathbf 0"] \arrow[ru, "\Delta", red,dashed] &\color{red} U(S1\otimes S1) \arrow[d, "U(\mathbf 0\otimes Su)",sloped, red,dashed] &\color{red} US\Psi\times US\Phi \arrow[r, "\eta", red,dashed] \arrow[ld, "n", red,dashed]   & US(US\Psi\times US\Phi) \arrow[dd, "USn"]             \\
        US(\Psi\times\Phi) \arrow[d, "\eta", red,dashed, bend left=10] & \color{red}U(S\Psi\otimes S\Phi) \arrow[rrd, "\eta", red,dashed] \arrow[l, "Um", red,dashed] & & \\
        USUS(\Psi\times\Phi) \arrow[u, "\mu"] & & & USU(S\Psi\otimes S\Phi) \arrow[lll, "USUm"] 
      \end{tikzcd}
    \end{center}
    \begin{enumerate}
    \item \label{distrzero:C} Naturality of $\eta$.
    \item \label{distrzero:D} Functoriality of the product.
    \item \label{distrzero:B} Naturality of $\Delta$.
    \item \label{distrzero:A} $U(\mathbf 0\otimes Su)=U\mathbf 0$, hence, we conclude by \axiomz\ with the maps $n\circ\Delta$ and $\Id$.
    \item \label{distrzero:E} Naturality of $n$.
    \item \label{distrzero:F} Property of map $\mathbf 0$.
    \end{enumerate}

  \item[\rdistzl] If $u$ has type $\Psi$, $\Uparrow_\ell
    u\times\z[\Phi]\lra\z[\Psi\times\Phi]$. Analogous to case $\rdistzr$.

  \item[($\mathsf{dist}_\Uparrow^+$)] 
    $\Uparrow (t+u)\lra (\Uparrow t+\Uparrow u)$.
    We only give the details for $\Uparrow_r$, the case $\Uparrow_\ell$ is
    analogous.
    \[
      \vcenter{
        \infer{\vdash\Uparrow_r (t+u):S(\Psi\times \Phi)}
        {
          \infer{\vdash t+u:S(S\Psi\times \Phi)}
          {
            \vdash t:S(S\Psi\times \Phi)
            &
            \vdash u:S(S\Psi\times \Phi)
          }
        }
      }
      \qquad\textrm{and}\qquad
      \vcenter{
        \infer{\vdash\Uparrow_r  t+\Uparrow_r u:S(\Psi\times \Phi)}
        {
          \infer{\vdash\Uparrow_r  t:S(\Psi\times \Phi)}
          {\vdash t:S(S\Psi\times \Phi)}
          &
          \infer{\vdash\Uparrow_r  u:S(\Psi\times \Phi)}
          {\vdash u:S(S\Psi\times \Phi)}
        }
      }
    \]
    Then
    \begin{center}
      \begin{tikzcd}[column sep=1.8cm]
        (US(US\Psi\times \Phi))^2 \arrow[d, "(U(\Id\times\eta))^2"] \arrow[rr, "\hat +"] &  & US(US\Psi\times \Phi) \arrow[d, "U(\Id\times\eta)"] \\
        (US(US\Psi\times US\Phi))^2 \arrow[d, "(USn)^2"] \arrow[rr,red,dashed, "\hat +"]            &  & US(US\Psi\times US\Phi) \arrow[d, "USn"]            \\
        (USU(S\Psi\otimes S\Phi))^2 \arrow[d, "(USUm)^2"] \arrow[rr,red,dashed, "\hat +"]           &  & USU(S\Psi\otimes S\Phi) \arrow[d, "USUm"]           \\
        (USUS(\Psi\times \Phi))^2 \arrow[d, "\mu^2"] \arrow[rr,red,dashed, "\hat +"]                &  & USUS(\Psi\times \Phi) \arrow[d, "\mu"]              \\
        (US(\Psi\times \Phi))^2 \arrow[rr, "\hat +"]         &  & US(\Psi\times \Phi)                                
      \end{tikzcd}
    \end{center}
    This diagram commutes by naturality of $\hat +$.

  \item[($\mathsf{dist}_\Uparrow^\alpha$)] $\Uparrow (\alpha.t)\lra \alpha.\Uparrow t$.
    We only give the details for $\Uparrow_r$, the case $\Uparrow_\ell$ is
    similar.
    \[
      \vcenter{
        \infer{\vdash\Uparrow_r (\alpha.t):S(\Psi\times \Phi)}
        {
          \infer{\vdash\alpha.t:S(S\Psi\times \Phi)}
          {
            \vdash t:S(S\Psi\times \Phi)
          }
        }
      }
      \qquad\textrm{and}\qquad
      \vcenter{
        \infer{\vdash\alpha.\Uparrow_r t:S(\Psi\times \Phi)}
        {
          \infer{\vdash\Uparrow_r  t:S(\Psi\times \Phi)}
          {\vdash t:S(S\Psi\times \Phi)}
        }
      }
    \]

    Then
    \begin{center}
      \adjustbox{scale=.9,center}{
      \begin{tikzcd}[
        execute at end picture={
          \path
          (\tikzcdmatrixname-1-1) -- (\tikzcdmatrixname-5-1) coordinate[pos=0.5](aux1)
          (\tikzcdmatrixname-1-4) -- (\tikzcdmatrixname-5-4) coordinate[pos=0.5](aux2)
          (aux1) -- (aux2) node[midway,blue]{\small\eqref{distalphaup:A}}
          (aux1) -- (aux2) node[midway,blue,yshift=1.6cm]{\small\eqref{distalphaup:A}}
          (aux1) -- (aux2) node[midway,blue,yshift=-1.6cm]{\small\eqref{distalphaup:A}}
          (aux1) -- (aux2) node[midway,blue,xshift=-4cm,yshift=.6cm]{\small\eqref{distalphaup:B}}
          (aux1) -- (aux2) node[midway,blue,xshift=-4cm,yshift=-.6cm]{\small\eqref{distalphaup:B}}
          (aux1) -- (aux2) node[midway,blue,yshift=1.6cm,xshift=-4cm]{\small\eqref{distalphaup:B}}
          (aux1) -- (aux2) node[midway,blue,yshift=-1.6cm,xshift=-3cm]{\small\eqref{distalphaup:D}}
          (aux1) -- (aux2) node[midway,blue,xshift=4cm,yshift=.6cm]{\small\eqref{distalphaup:C}}
          (aux1) -- (aux2) node[midway,blue,xshift=4cm,yshift=-.6cm]{\small\eqref{distalphaup:C}}
          (aux1) -- (aux2) node[midway,blue,yshift=1.6cm,xshift=4cm]{\small\eqref{distalphaup:C}}
          (aux1) -- (aux2) node[midway,blue,yshift=-1.6cm,xshift=5cm]{\small\eqref{distalphaup:E}}
          ;
        }]
        US(US\Psi\times\Phi) \arrow[r, "U\lambda"] \arrow[d, "U(\Id\times\eta)"]  & U(S(US\Psi\times\Phi)\otimes\I) \arrow[r, "U(\Id\otimes\alpha)"] \arrow[d,red,dashed,swap, "U(S(\Id\times\eta)\otimes\Id)"'] & U(S(US\Psi\times \Phi)\otimes\I) \arrow[r, "U\lambda^{-1}"] \arrow[d,red,dashed,"U(S(\Id\times\eta)\otimes\Id)"]  & US(US\Psi\times \Phi) \arrow[d, "U(\Id\times\eta)"]                      \\
        US(US\Psi\times US\Phi) \arrow[d,"USn"] \arrow[r,red,dashed, "U\lambda"]                                & \color{red}U(S(US\Psi\times US\Phi)\otimes\I) \arrow[d,red,dashed, "U(Sn\otimes\Id)"] \arrow[r,red,dashed, "U(\Id\otimes\alpha)"]                            & \color{red}U(S(US\Psi\times US\Phi)\otimes\I) \arrow[d,red,dashed, "U(Sn\otimes\Id)"] \arrow[r,red,dashed, "U\lambda^{-1}"]                            & US(US\Psi\times US\Phi) \arrow[d,"USn"]                                  \\
        USU(S\Psi\otimes S\Phi) \arrow[d,"USUm"] \arrow[r,red,dashed, "U\lambda"]                               & \color{red}U(SU(S\Psi\otimes S\Phi)\otimes\I) \arrow[d,red,dashed, "U(SUm\otimes\Id)"]                                                            & \color{red}U(SU(S\Psi\otimes S\Phi)\otimes\I) \arrow[d,red,dashed, "U(SUm\otimes\Id)"] \arrow[r,red,dashed, "U\lambda^{-1}"]                           & US(U(S\Psi\otimes S\Phi)) \arrow[d, "USUm"]                               \\
        USUS(\Psi\times\Phi) \arrow[d, "\mu_{\Psi\times\Phi}{\color{red}=U\varepsilon_{S(\Psi\times\Phi)}}"] \arrow[r,red,dashed, "U\lambda"] & \color{red}U(SUS(\Psi\times \Phi)\otimes\I) \arrow[r,red,dashed, "U(\Id\otimes\alpha)"] \arrow[d, "\color{red}U(\varepsilon_{S(\Psi\times\Phi)}\otimes\Id)"] & \color{red}U(SUS(\Psi\times \Phi)\otimes\I) \arrow[r,red,dashed, "U\lambda^{-1}"] \arrow[d, "\color{red}U(\varepsilon_{S(\Psi\times\Phi)}\otimes\Id)"] & USUS(\Psi\times \Phi) \arrow[d, "\mu_{\Psi\times\Phi}{\color{red}=U\varepsilon_{S(\Psi\times\Phi)}}"] \\
        US(\Psi\times\Phi) \arrow[swap,r, "U\lambda"]                                                      & U(S(\Psi\times\Phi)\otimes\I) \arrow[swap,r, "S(\Id\otimes\alpha)"]                                                                         & U(S(\Psi\times\Phi)\otimes\I) \arrow[swap,r, "U\lambda^{-1}"]                                                                         & US(\Psi\times \Phi)                                                      
      \end{tikzcd}
      }
    \end{center}
    \begin{enumerate}
    \item\label{distalphaup:B} Naturality of $\lambda$.
    \item\label{distalphaup:D} Naturality of $\lambda$ and the definition of monad given by an adjunction.
    \item\label{distalphaup:A} Functoriality of tensor.
    \item\label{distalphaup:C} Naturality of $\lambda^{-1}$.
    \item\label{distalphaup:E} Naturality of $\lambda^{-1}$ and the definition of monad given by an adjunction.
    \end{enumerate}

  \item[\rdistcazeror]
    $\Uparrow_r\z[S(S\Psi)\times\Phi]\lra\Uparrow_r\z[S\Psi\times\Phi]$. We have

    \[
      \vcenter{
        \infer{\vdash\Uparrow_r\z[S(S\Psi)\times\Phi]:S(\Psi\times\Phi)}
        {
          \infer{\vdash\z[S(S\Psi)\times\Phi]:S(S(S\Psi)\times\Phi)}{}
        }
      }
      \qquad\textrm{and}\qquad
      \vcenter{
        \infer{\vdash\Uparrow_r\z[S\Psi\times\Phi]:S(\Psi\times\Phi)}
        {
          {
            \infer{\vdash\z[S\Psi\times\Phi]:S(S\Psi\times\Phi)}{}
          }
        }
      }
    \]

    Then, 
    \begin{center}
      \begin{tikzcd}
US1 \arrow[r, "U\mathbf 0"]                        & US((US)^2\Psi\times\Phi) \arrow[r, "U(\Id\times\eta)"] & US((US)^2\Psi\times US\Phi) \arrow[ddd, "USn"] \\
1 \arrow[d, "\eta"] \arrow[u, "\eta"]              &                                                        &                                                \\
US1 \arrow[d, "U\mathbf 0"]                        &                                                        &                                                \\
US(US\Psi\times\Phi) \arrow[d, "U(\Id\times\eta)"] & USUS(US\Psi\times\Phi) \arrow[l, "\mu"]                & USU(SUS\Psi\otimes S\Phi) \arrow[l, "USUm"]    \\
US(US\Psi\times US\Phi) \arrow[r, "USn"]           & USU(S\Psi\otimes S\Phi) \arrow[r, "USUm"]              & USUS(\Psi\times\Phi) \arrow[d, "\mu"]          \\
                                                   &                                                        & US(\Psi\times\Phi)
\end{tikzcd}
    \end{center}
    This diagram commutes by the property of the map $\mathbf 0$.

  \item[\rdistcazerol]
    $\Uparrow_\ell\z[\Phi\times S(S\Psi)]\lra\Uparrow_\ell\z[\Phi\times S\Psi]$.
    Analogous to case $\rdistcazeror$.

  \item[\rcaneutr] If $u\in\tbasis$, $\Uparrow_r u\times
    v\lra  u\times v$. We have
    \[
      \vcenter{
        \infer{\vdash\Uparrow_r u\times v:S(\Psi\times \Phi)}
        {
          \infer{\vdash u\times v:S(S\Psi\times \Phi)}
          {
            \infer{\vdash u\times v:S\Psi\times \Phi}
            {
              \infer{\vdash u:S\Psi}{\vdash u:\Psi}
              &
              \vdash v:\Phi
            }
          }
        }
      }
      \qquad\textrm{and}\qquad
      \vcenter{
        \infer{\vdash u\times v:S(\Psi\times \Phi)}
        {
          \infer{\vdash u\times v:\Psi\times \Phi}
          {
            \vdash u:\Psi & \vdash v:\Phi
          }
        }
      }
    \]
    Then
    \begin{center}
      \begin{tikzcd}[execute at end picture={
          \path (\tikzcdmatrixname-1-2) -- (\tikzcdmatrixname-2-2) coordinate[pos=0.5] (aux)
          (aux) -- (\tikzcdmatrixname-1-4) node[midway,blue,xshift=-17mm]{\small\eqref{neutrup:A}};
          \path (\tikzcdmatrixname-1-4) -- (\tikzcdmatrixname-3-4) coordinate[pos=0.5] (aux)
          (\tikzcdmatrixname-2-2) -- (aux) node[midway,blue,xshift=1cm]{\small\eqref{neutrup:B}};
          \path (\tikzcdmatrixname-2-2) -- (\tikzcdmatrixname-3-2) coordinate[pos=0.5] (aux)
          (aux) -- (\tikzcdmatrixname-3-4) node[xshift=-1cm,yshift=1mm,midway,blue]{\small\eqref{neutrup:C}};
          \path (aux) -- (\tikzcdmatrixname-3-4) node[xshift=-2.75cm,yshift=2.5mm,midway,blue]{\small\eqref{neutrup:D}};
        }]
        {1}\approx{1}^2 \arrow[r, "u\times v"] & \Psi\times\Phi \arrow[d, "\eta"] \arrow[r, "\eta\times\Id"] & US\Psi\times\Phi \arrow[r, "\eta"]  & US(US\Psi\times \Phi) \arrow[dd, "U(\Id\times\eta)"] \\
        & US(\Psi\times \Phi) \arrow[rru,red,dashed, "US(\eta\times\Id)"] \arrow[rrd,red,dashed, "US(\eta\times\eta)"] \arrow[d, "US\eta",red,dashed, bend left=50] & & \\
        & USUS(\Psi\times \Phi) \arrow[u, "\mu"] & US(U(S\Psi\otimes S\Phi)) \arrow[l, "USUm"] & US(US\Psi\times US\Phi) \arrow[l, "USn"]            
      \end{tikzcd}
    \end{center}
    \begin{enumerate}
    \item\label{neutrup:A} Naturality of $\eta$.
    \item\label{neutrup:B} Functoriality of product.
    \item\label{neutrup:D} Remark that $\mu\circ US\eta = \Id$.
    \item\label{neutrup:C} $\eta$ is a monoidal linear transformation.
    \end{enumerate}

  \item[\rcaneutl] If $v\in\tbasis$, $\Uparrow_\ell u\times
    v\lra  u\times v$. Analogous to case $\rcaneutr$.

  \item[\rcaneutzr] $\Uparrow_r\z[S(\B^n)\times\Phi]\lra\z[\B^n\times\Phi]$.
    We have
    \[
      \vcenter{
        \infer{\vdash\Uparrow_r\z[S(\B^n)\times\Phi]:S(\B^n\times\Phi)}
        {
          \infer{\vdash\z[S(\B^n)\times\Phi]:S(S(\B^n)\times\Phi)}{}
        }
      }
      \qquad\textrm{and}\qquad
      \vcenter{
        \infer{\vdash\z[\B^n\times\Phi]:S(\B^n\times\Phi)}{}
      }
    \]

    Then, 
    \begin{center}
      \begin{tikzcd}[column sep=15mm]
        US1\arrow[r,"U\mathbf 0"] & US(US\B^n\times\Phi) \arrow[r, "U(\Id\times\eta)"] & US(US\B^n\times
        US\Phi) \arrow[d, "USn"] \\
	1 \arrow[u,"\eta"]\arrow[d,"\eta"] & & USU(S\B^n\otimes S\Phi) \arrow[d, "USUm"] \\
        US1\arrow[r,"U\mathbf 0"] & US(\B^n\times\Phi) & USUS(\B^n\times\Phi) \arrow[l, "\mu"]
      \end{tikzcd}
    \end{center}
    Remark that $\mathbf 0 = f\circ \mathbf 0$ for any $f$.

  \item[\rcaneutzl] $\Uparrow_\ell\z[\Phi\times S(\B^n)]\lra\z[\Phi\times\B^n]$.
    Analogous to case \rcaneutzr.
    
  \item[Contextual rules] Trivial by composition law.
    \qed
  \end{description}
\end{proof}

\xrecap{Lemma}{Adequacy}{lem:Adequacy}{
  If $\Gamma\vdash t:A$ and $\sigma\vDash\Gamma$, then $\sigma t\in\com A$.}
\begin{proof}
  We proceed by structural induction on the derivation of $\Gamma\vdash t:A$.
  \begin{itemize}
  \item $\vcenter{\infer[^\tax] {\Gamma^\B,x:\Psi\vdash x:\Psi} {}}$\qquad
    Since $\sigma\vDash\Gamma^\B,x:\Psi$, we have $\sigma x\in\com\Psi$.
  \item $\vcenter{\infer[^{\tax_{\mathbf 0}}] {\Gamma^\B\vdash \z:S(A)} {}}$\qquad
    By definition, $\sigma\z[A]=\z[A]\in\com{S(A)}$.

  \item $\vcenter{\infer[^{\tax_{\ket 0}}] {\Gamma^\B\vdash\ket 0:\B} {}}$\qquad
    By definition, $\sigma\ket 0=\ket 0\in\com{\B}$.

  \item $\vcenter{\infer[^{\tax_{\ket 1}}] {\Gamma^\B\vdash\ket 1:\B} {}}$\qquad
    By definition, $\sigma\ket 1=\ket 1\in\com{\B}$. 

  \item $\vcenter{\infer[^{\alpha_I}] {\Gamma\vdash \alpha.t:S(A)}
      {\Gamma\vdash t:S(A)}}$

    By the induction hypothesis, $\sigma t\in\com{S(A)}$, hence, one of the
    following cases occur:
    \begin{itemize}
    \item $t\in\mathcal S\com A$, then $t=\sum_i\beta_ir_i$ with $r_i\in\com A$.
      Since $\alpha.\sum_i\beta_ir_i\lra^*\sum_i(\alpha\times\beta_i)r_i$, we
      have $\alpha.t\in\com{S(A)}$.
    \item $t\lra^* r$ with $r\in\mathcal S\com A$, then $t\in\overline{\mathcal
        S\com A}\subset\com{S(A)}$.
    \end{itemize}

  \item $\vcenter{\infer[^{+_I}] {\Gamma,\Delta,\Xi^\B\vdash\pair tu:S(A)}
      {\Gamma,\Xi^\B\vdash t:S(A) & \Delta,\Xi^\B\vdash u:S(A)}}$
    
    By the induction hypothesis, $\sigma_1\sigma t,\sigma_2\sigma
    u\in\com{S(A)}$, where $\sigma_1\vDash\Gamma$, $\sigma_2\vDash\Delta$, and
    $\sigma\vDash\Xi^\B$.

    By definition
    $\sigma_1\sigma t+\sigma_2\sigma u=\sigma_1\sigma_2\sigma(t+u)\in\com{S(A)}$.

  \item $\vcenter{\infer[^\tif]{\Gamma\vdash\ite{}tr:\B\Rightarrow
        A}{\Gamma\vdash t:A & \Gamma\vdash r:A}}$
    
    By the induction hypothesis, $\sigma t\in\com A$ and
    $\sigma r\in\com A$. Hence, for any $s\in\com\B$, $\ite s{\sigma
      t}{\sigma r}$ reduces either to $\sigma t$ or to $\sigma r$, hence it is in $\com A$, therefore, $\ite{}{\sigma t}{\sigma
      r}\in\com{\B\Rightarrow A}$.
  \item $\vcenter{\infer[^{\Rightarrow_I}] {\Gamma\vdash\lambda
        x{:}\Psi.t:\Psi\Rightarrow A} {\Gamma,x:\Psi\vdash t:A}}$
    
    Let $r\in\com\Psi$. Then, $\sigma(\lambda
    x{:}\Psi.t)r=(\lambda x{:}\Psi.\sigma t)r\rightarrow (r/x)\sigma
    t$. Since $(r/x)\sigma \vDash\Gamma,x:\Psi$, we have, by the induction
    hypothesis, that $(r/x)\sigma t\in\com A$. Therefore, $\lambda
    x{:}\Psi.t\in\com{\Psi\Rightarrow A}$.

  \item $\vcenter{\infer[^{\Rightarrow_E}] {\Delta,\Gamma,\Xi^\B\vdash tu:A}
      {\Delta,\Xi^\B\vdash u:\Psi & \Gamma,\Xi^\B\vdash t:\Psi\Rightarrow A}}$
    
    By the induction hypothesis, $\sigma_1\sigma u\in\com\Psi$ and $\sigma_2\sigma
    t\in\com{\Psi\Rightarrow A}$, where $\sigma_1\vDash\Delta$,
    $\sigma_2\vDash\Gamma$, and $\sigma\vDash\Xi^\B$. Then, by definition, $\sigma_1\sigma
    t\sigma_2\sigma r=\sigma_1\sigma_2\sigma(tr)\in\com A$.

  \item $\vcenter{\infer[^{\Rightarrow_{ES}}] {\Delta,\Gamma,\Xi^\B\vdash tu:S(A)}
      {\Delta,\Xi^\B\vdash u:S\Psi & \Gamma,\Xi^\B\vdash t:S(\Psi\Rightarrow
        A)}}$
    
    By the induction hypothesis
    $\sigma_1\sigma t\in\com{S(\Psi\Rightarrow A)}=\overline{S\com{\Psi\Rightarrow
        A}}$ and $\sigma_2\sigma u\in\com{S\Psi}=\overline{S\com\Psi}$, where $\sigma_1\vDash\Gamma$,
    $\sigma_2\vDash\Delta$, and $\sigma\vDash\Xi^\B$.
    Let $\sigma_1\sigma t\lra^*\sum_i\alpha_it_i$ with $t_i\in\com{\Psi\Rightarrow A}$ and
    $\sigma_2\sigma u\lra\sum_j\beta_ju_j$, with $u_j\in\com\Psi$. Then
    $\sigma_1\sigma_2\sigma (tu)=(\sigma_1\sigma t)(\sigma_2\sigma u)\lra^*\sum_{ij}\alpha_i\beta_jt_iu_j$
    with $t_iu_j\in\com A$, therefore, $\sigma_1\sigma_2\sigma (tu)\in\com{S(A)}$.

  \item $\vcenter{\infer[^{S_I}] {\Gamma\vdash t:S(A)} {\Gamma\vdash t:A}}$

    By the induction hypothesis, $\sigma t\in\com A\subseteq S\com A\subseteq\com{S(A)}$.
    
  \item $\vcenter{\infer[^{\times_I}] {\Gamma,\Delta,\Xi^\B\vdash t\times
        u:\Psi\times\Phi} {\Gamma,\Xi^\B\vdash t:\Psi & \Delta,\Xi^\B\vdash
        u:\Phi}}$
    
    By the induction hypothesis, $\sigma_1\sigma t\in\com\Psi$ and
    $\sigma_2\sigma u\in\com\Phi$, hence,
    $\sigma_1\sigma t\times\sigma_2\sigma u=\sigma_1\sigma_2\sigma (t\times u)\in\com\Psi\times\com\Phi\subseteq\com{\Psi\times\Phi}$.

  \item $\vcenter{\infer[^{\times_{Er}}] {\Gamma\vdash \head~t:\B} {\Gamma\vdash
        t:\B^n & {\scriptstyle n>1}}}$
    
    By the induction hypothesis, $\sigma
    t\in\com{\B^n}=\overline{\com\B\times\com{\B^{n-1}}}=\{u\mid
    u\lra^*u_1\times u_2\textrm{ with }u_1\in\com\B\textrm{ and
    }u_2\in\com{\B^{n-1}}\}$. Hence, $\sigma(\head\ t)=\head\ \sigma
    t\lra^*\head(u_1\times u_2)\lra u_1\in\com\B$.

  \item $\vcenter{\infer[^{\times_{El}}] {\Gamma\vdash \tail~t:\B^{n-1}} {\Gamma\vdash t:\B^n & {\scriptstyle n>1}}}$

    By the induction hypothesis, $\sigma
    t\in\com{\B^n}=\overline{\com\B\times\com{\B^{n-1}}}=\{u\mid
    u\lra^*u_1\times u_2\textrm{ with }u_1\in\com\B\textrm{ and
    }u_2\in\com{\B^{n-1}}\}$. Hence, $\sigma(\tail\ t)=\tail\ \sigma
    t\lra^*\tail(u_1\times u_2)\rightarrow u_2\in\com{\B^{n-1}}$.

  \item $\vcenter{\infer[^{\Uparrow_r}] {\Gamma\vdash \Uparrow_rt:S(\Psi\times
        \Phi)} {\Gamma\vdash t:S(S\Psi\times \Phi)}}$

    By the induction hypothesis, we have that
    $\sigma t\in\com{S(S\Psi\times\Phi)}$. Therefore, $\sigma
    t\in\overline{S(\overline{\overline{S\com\Psi}\times\com\Phi})}$.

    Hence, $\sigma t\lra^*\sum_i\alpha_i (
    (\sum_{j_i}\beta_{ij_i}r_{ij_i})
    \times u_i)$ with
    $u_i\in\com\Phi$ and
    $r_{ij_i}\in\com\Psi$.

    Hence, $\Uparrow_r t\lra^*\sum_{j_i}(\alpha_i\beta_{ij_i})\Uparrow_r(r_{ij_i}\times u_i)\in\com{S(\Psi\times\Phi)}$.

  \item $\vcenter{\infer[^{\Uparrow_\ell}] {\Gamma\vdash \Uparrow_\ell
        t:S(\Psi\times \Phi)} {\Gamma\vdash t:S(\Psi\times S\Phi)}}$\qquad
    Analogous to previous case.
    \qed
  \end{itemize}
\end{proof}

\end{document}